\documentclass[12pt,draftcls,onecolumn,web,confmode]{ieeecolor}
\usepackage{graphicx}

\usepackage{booktabs} % For formal tables

\usepackage{hyperref}
\usepackage{url}
\usepackage[english]{babel}
\usepackage{algorithm}
\usepackage[noend]{algpseudocode}
\usepackage{amssymb,amsmath}
\usepackage{subcaption}
\usepackage{graphicx}
\usepackage{comment}
\usepackage{thmtools, thm-restate}
\usepackage{multirow}
\usepackage{paralist}

\usepackage{chngcntr}

\usepackage{xcolor} % Required for specifying colors by name
\definecolor{ocre}{RGB}{243,102,25} % Define the orange color used for highlighting throughout the book

\usepackage{avant} % Use the Avantgarde font for headings
\usepackage{mathptmx} % Use the Adobe Times Roman as the default text font together with math symbols from the Sym­bol, Chancery and Com­puter Modern fonts

\usepackage{microtype} % Slightly tweak font spacing for aesthetics
\usepackage[utf8x]{inputenc} % Required for including letters with accents
\usepackage[T1]{fontenc} % Use 8-bit encoding that has 256 glyphs

\usepackage{tikz}
\usetikzlibrary{shapes}
\RequirePackage{ifthen}
\RequirePackage{amsmath}
\RequirePackage{amssymb}
\RequirePackage{xspace}
\RequirePackage{graphics}
\usepackage{color}
\RequirePackage{textcomp}
\usepackage{keyval}
\usepackage{xspace}
\usepackage{mathrsfs}
\let\labelindent\relax
\usepackage{enumitem}

\newcommand{\sayan}[1]{\textcolor{blue}{#1}}
\newcommand{\hussein}[1]{\textcolor{red}{#1}}

\newcommand{\reals}{{\mathbb{R}}} 

\newcommand{\nnreals}{{\mathbb{R}^+}}
\newcommand{\dom}{\relax\ifmmode {\sf dom} \else ${\sf dom}$\fi} 

\def\A{{\mathcal{A}}} % HA

\newcommand{\dur}{\mathit{dur}} %Code of a TM
\newcommand{\Reach}{\mathtt{Reach}} %Code of a TM
\newcommand{\init}{\mathit{init}} %Code of a TM
\newcommand{\Rtube}{\mathtt{Reach}} % ReachTb
\newcommand{\ARtube}{\mathtt{Reach}} % ReachTb
\newcommand{\AReach}{\mathtt{Reach}}  % ReachTb
\newcommand{\ha}{\mathcal{A}}
\newcommand{\hb}{\mathcal{B}}

\newcommand{\Symcache}{\textsc{Symcache}}
\newcommand{\Nosym}{\textsc{Nosym}}
\newcommand{\Symvir}{\textsc{SymVir}}

\newcommand{\SC}{\textsc{SC}}
\newcommand{\NS}{\textsc{NS}}
\newcommand{\SV}{\textsc{SV}}

\newcommand{\node}{\mathit{node}}

\newcommand{\reachset}{\mathtt{Reach}}

\newcommand{\concat}{\bigwedge} % {\mathbin{^{\frown}}} % concatenation 

\newcommand{\disctrans}{D}
\newcommand{\edgeset}{E}
\newcommand{\edgeinstance}{e}
\newcommand{\stateset}{X}
\newcommand{\stateinstance}{x}

\newcommand{\exec}{\texttt{Exec}}

\newcommand{\safetycache}{\mathit{safetycache}}
\newcommand{\symcompute}{\mathtt{symComputeReachtube}} %Unsafe set

\newcommand{\safe}{\mathit{safe}}

\newcommand{\state}{{\bf x}}
\newcommand{\statey}{{\bf y}}
\newcommand{\rtube}{\mathit{rtube}}
\newcommand{\ourtacastool}{\sf{CacheReach}}
\newcommand{\ourtool}{\sf{VirReach}}

\newcommand{\modeinitset}{K}
\newcommand{\modereachset}{R}
\newcommand{\initset}{X_{\mathit{init}}}
\newcommand{\initmode}{p_{\mathit{init}}}
\newcommand{\initsetf}{X_{\mathit{init,1}}}
\newcommand{\initmodef}{p_{\mathit{init,1}}}
\newcommand{\initsets}{X_{\mathit{init,2}}}
\newcommand{\initmodes}{p_{\mathit{init,2}}}
\newcommand{\initsetv}{X_{\mathit{init,v}}}
\newcommand{\initmodev}{p_{\mathit{init,v}}}
\newcommand{\initsetwp}{X_{\mathit{init,W}}}

\newcommand{\wpha}{W} % \mathit{wp}
\newcommand{\roha}{R} %{{\mathit{rd}}}
\newcommand{\origtrans}{\mathit{ot}}
\newcommand{\coortrans}{\mathit{ct}}
\newcommand{\fsr}{\mathcal{R}_{\mathit{rv}}}

\newcommand{\guard}{\mathit{guard}}
\newcommand{\guardsingle}{fset}
\newcommand{\resetsingle}{sset}
\newcommand{\reset}{\mathit{reset}}
\newcommand{\ltube}{\mathit{len}}

\newcommand{\ftime}{\mathit{ftime}}
\newcommand{\etime}{\mathit{etime}}
\newcommand{\fstate}{\mathit{fstate}}
\newcommand{\lstate}{\mathit{lstate}}
\newcommand{\resrelguard}{\mathit{relguard}}
\renewcommand{\path}{\mathit{pseq}}
\newcommand{\pathset}{\mathtt{Paths}}
\newcommand{\pathinstance}{\mathit{path}}
\newcommand{\timebound}{\mathit{tbound}}
\newcommand{\TR}{\mathit{TR}}
\newcommand{\expth}{\mathit{th}}
\newcommand{\expst}{\mathit{st}}
\newcommand{\expnd}{\mathit{nd}}
\newcommand{\exprd}{\mathit{rd}}

\usepackage{setspace}
\usepackage{listings}
\usepackage{hyperref}
\hypersetup{
	pdftitle={Modeling and Verification for CPS},
	pdfauthor={Sayan Mitra},
	colorlinks=true,
	citecolor={blue},
	linkcolor = {blue},
	pagecolor={blue},
	backref={true},
	bookmarks=true,
	bookmarksopen=false,
	bookmarksnumbered=true
}

\newcounter{theorems}
\newtheorem{theorem}[theorems]{Theorem}%[theorem]
\newtheorem{corollary}[theorems]{Corollary}%[theorem]
\newtheorem{lemma}[theorems]{Lemma}%[theorem]
\newcounter{remarks}

\newtheorem{remark}[remarks]{Remark}%[remark]f
\newcounter{examples}
\newtheorem{example}[examples]{Example}%[remark]
\newtheorem{definition}[remarks]{Definition}%[definition]
\newtheorem{proposition}[remarks]{Proposition}%[definition]
\usepackage[normalem]{ulem}

\newcommand{\computereachset}{{\sf computeReachset}}

\newcommand{\unboundedverif}{{\sf unboundedVerif}}

\newcommand{\computerealfromvirtual}{{\sf computeRealSetfromVirtualSet}}

\newcommand{\permodedict}{{\mathit{Dictionary}}}
\newcommand{\inmodes}{\mathit{inModes}}

\newcommand{\globalR}{\mathit{globalR}}

\newcommand{\checkFixedPoint}{\mathit{checkFixedPoint}}

\newcommand{\src}{\mathit{src}}
\newcommand{\dest}{\mathit{dest}}

\newcommand{\pin}{P_\mathit{in}}
\newcommand{\pd}{P_\mathit{d}}
\newcommand{\pout}{P_{\mathit{out}}}

\newcommand{\pdin}{p_\mathit{d}}

\usepackage{tikz}
\usetikzlibrary{fit}
\usepackage{twoopt}

\newsavebox\IBoxA \newsavebox\IBoxB \newlength\IHeight
\newcommand\TwoFig[6]{% Image1 Caption1 Label1 Im2 Cap2 Lab2
	\sbox\IBoxA{\includegraphics[width=0.45\textwidth]{#1}}
	\sbox\IBoxB{\includegraphics[width=0.45\textwidth]{#4}}%
	\ifdim\ht\IBoxA>\ht\IBoxB
	\setlength\IHeight{\ht\IBoxB}%
	\else\setlength\IHeight{\ht\IBoxA}\fi
	\begin{figure}[!htb]
		\minipage[t]{0.45\textwidth}\centering
		\includegraphics[height=\IHeight]{#1}
		\caption{#2}\label{#3}
		\endminipage\hfill
		\minipage[t]{0.45\textwidth}\centering
		\includegraphics[height=\IHeight]{#4}
		\caption{#5}\label{#6}
		\endminipage 
	\end{figure}%
}

\usepackage{framed}
\definecolor{light-gray}{gray}{0.95}
\colorlet{shadecolor}{light-gray}
\colorlet{framecolor}{black}
{\endMakeFramed}

\newenvironment{tBox*}{%
	\MakeFramed {\advance\hsize-\width \FrameRestore}}%
{\endMakeFramed}

\newcommandtwoopt\Textbox[5][2.5cm][2cm]{%
	\begin{tikzpicture}[remember picture,overlay]
	\coordinate (aux) at ([xshift=#1]#4);
	\node[inner ysep=3pt,yshift=0.6ex,draw=gray,thick,
	fit=(#3) (aux),baseline] 
	(box) {};
	\node[text width=#2,anchor=north east,
	font=\sffamily\footnotesize,align=right] 
	at (box.north east) {#5};
	\end{tikzpicture}%
}
\newcommand{\ball}{\mathit{B}}

\newcommand{\rv}{\mathit{rv  }}

\usepackage{generic}
\usepackage{cite}
\usepackage{amsmath,amssymb,amsfonts}
\usepackage{graphicx}
\usepackage{textcomp}
\def\BibTeX{{\rm B\kern-.05em{\sc i\kern-.025em b}\kern-.08em
    T\kern-.1667em\lower.7ex\hbox{E}\kern-.125emX}}
\begin{document}
\title{Symmetry Abstractions for Hybrid Systems and their Applications}
\author{Hussein Sibai and
	Sayan Mitra \\ 
	University of Illinois, Urbana IL 61801, USA \\
	\{sibai2,mitras\}@illinois.edu
}
\maketitle

\begin{abstract}
A symmetry of a dynamical system is a map that transforms one trajectory to another trajectory. 
We introduce a new type of abstraction for hybrid automata based on symmetries. The abstraction combines different modes
in a concrete automaton $\ha$,  whose trajectories are related by symmetries, into a single mode in the abstract automaton $\hb$. 
The abstraction sets the guard and reset of an abstract edge to be the union of the symmetry-transformed guards and resets of the concrete edges.
We establish the soundness of the abstraction using a forward simulation relation (FSR) and present several examples.
Our abstraction results in simpler automata, that are more amenable for formal analysis and design.
We illustrate an application of this abstraction in making reachability analysis faster and enabling unbounded time safety verification.
We show how 
a fixed point of the reachable set computation of $\hb$ can be used to answer reachability queries for $\ha$, even if the latter visits an infinite and unbounded sequences of modes. 
We present our implementation of the abstraction construction, the fixed point check, and the map that transforms abstract reachable sets to concrete ones in a software tool. Finally, we show the advantage of our method over existing ones, and the different aspects of our abstraction, in a sequence of experiments including scenarios with linear and nonlinear agents following waypoints.
\end{abstract}

\begin{IEEEkeywords}
hybrid systems, abstractions, symmetry, formal methods, reachability analysis.
\end{IEEEkeywords}

\section{Introduction}
\label{sec:intro}

	\begin{comment}
\item model description: hybrid dynamics
\item symmetry definitions.
\item virtual model description: definition + theorem 
\item (contribution) Algorithm for the construction of unbounded virtual reachset
\item algorithm description + pseudo code
\item correctness guarantees

\item (contribution) Algorithm for safety checking using clustering of waypoints and unbounded virtual reachset
\item clustering waypoints using balls with unbounded reachset radii
\item bloating waypoints using virtual reachsets

\item Experiments
\item visualization of the unbounded virtual reachset
\item safety checking for swarm of drones

\begin{enumerate}
\item possible extensions: 3d plots: 2d vs. 3d symmetries,
\item effect of having more dimensions of the state that are not part of the mode -> harder to have a fixed point (related to the item above),
\item SMT solver for infinite safety check,
\item mapping tubes to the real system or mapping unsafe set to the virtual system
\item comparing results with the TACAS paper: improvement in time because we are copy pasting once we reach a fixed point instead of computing the intersection and then going to but lost in quality.
\item refining in the virtual model (needs mapping of the unsafe set to the virtual system): refining for all segments at once.
\item no cache miss theorem
\item people to cite: stanley bak and Murat (Berkeley)
\item examples should not contain the previous mode, this is handled by the virtual system construction.
\item levels of abstraction: one edge, two edges, ....
\end{enumerate}
\end{comment}

Hybrid system models bring together continuous and discrete behaviors~\cite{teel2012,vdSS2000,AlurBook}, and have proven to be useful in the design and analysis of a wide variety of systems, ranging from automotive, 
%~\cite{mit:prot,NimaHSCC,Powertrain:JinDKUB14,FanQM18,DuggiralaFM015},
%~\cite{hybrid_systems_manufacturing}, 
medical, to manufacturing
%~\cite{jiang2011cyber,KushnerBCFMS19,Sankaranarayanan17-Model-based,GrosuCAV},
and robotics.
%~\cite{alur00modular}. 
%Their design and analysis had been an active research topic for the past three decades~\cite{mit:hioa,ACHH93,Ant1993,alur95algorithmic,Spaceex,bak2017hylaa,verifbillion}. 
Exact algorithmic solutions for many of the synthesis and analysis problems for hybrid systems are known to be computationally intractable~\cite{HKPV98}. 
%
%For example, it has been shown that the reachability problem, the problem of finding which set of states will be visited by the system starting from a certain initial set, is undecidable, even for a very simple class of hybrid systems~\cite{HKPV98}.  
%
Therefore, one aims to develop approximate solutions, and the  main approach is to work with an {\em abstraction\/} of the hybrid system model. 
Roughly, an abstraction of a hybrid automaton $\ha$ is a simpler automaton $\hb$ that subsumes all behaviors of $\ha$. 
For example, $\hb$ may have fewer variables than $\ha$, or fewer modes, or it may have linear or rectangular dynamics approximating $\ha$'s nonlinear dynamics. 
 $\ha$ is called the {\em concrete}, and $\hb$ the {\em abstract} or {\em virtual} automaton.
 % We use the terms interchangeably. 
 Ideally, $\hb$ is simpler and  yet a useful over-approximation of $\ha$, and the formal analysis and synthesis of $\hb$ is tractable. 

%circumventing these hardness 
%To tackle such challenges, researchers often design {\em abstractions}~\cite{alurdiscrete,HsuMMS18,HARE:RoohiP017,Alur:2006,Clarke:2003:VHS,Julius06}. 

% The abstraction would be accompanied with a {\em forward simulation relation (FSR)}: a relation that maps behaviors of $A$ to ones of $B$. Such a relation ensures that the abstraction is sound.
Several important verification and synthesis techniques for hybrid systems have relied on  abstractions. The early decidability results for verification of rectangular hybrid systems were based on creating discrete and finite state abstractions~\cite{alurdiscrete,alur95algorithmic,HKPV98}. 
Decidability results based on abstractions for more general classes were presented in~\cite{VPVD:HSCC2008,lafferriere2000minimal}.
In a sequence of papers~\cite{GIRARD2012947,TabuadaPL02,tabuadaabstraction},  metric-based abstractions were developed for more general hybrid models and shown to be useful for both verification and synthesis. 
%
%For a given task, such as control synthesis or formal verification, abstraction is a useful approach to defining summarized versions of existing hybrid models, with details that are not relevant for the task in hand being removed.
%Moving from correctness to value, the aim of an abstraction is to reduce the computational effort required to accomplish the task, without losing significantly in the quality of the results. 
Techniques have been developed for automatically making abstractions more precise based on data and counterexamples~\cite{FehnkerCJK05,alur2006counterexample,HARE:RoohiP017,HARE:PrabhakarDM015}.
%ontroller synthesis~\cite{GIRARD2012947,hsu2018multi,moor2002abstraction}, and complexity
% analysis~\cite{maler91from}.
Finally, several practical approaches have been proposed for  computing abstractions based on linearization~\cite{sriram_abstraction}, state space partitioning~\cite{alur2006counterexample,HsuMMS18}, and hybridization~\cite{thaoOderHybridization}.

An important characteristic of  dynamical systems that has {\em not\/} been explored for constructing abstractions in the literature is {\em symmetry\/}.
Symmetry in a dynamical system $\dot{x} = f(x,p)$, with parameter $p$,  is a map $\gamma$ that transforms solutions  (or {\em trajectories}) of the system to other trajectories. For example, consider a trajectory $\xi_0 = \xi(x_0, p, \cdot)$ of a vehicle, starting from $x_0$ and following waypoint $p$. When $\xi_0$ is {\em shifted\/} by $a$, the result $\gamma_a(\xi_0)$ is just the trajectory of the vehicle starting from $\gamma_a(x_0)$ following $p'$, $\xi_0' = \xi(\gamma_a(x_0), p',\cdot)$. Here $p'$ is $p$ shifted by $a$. That is, the symmetry $\gamma_a$ relates different trajectories of the system. 
%This statement would hold for not only a variety of systems like cars, airplanes, underwater vehicles, etc., but also for  other  transformations like rotation, translation, and permutation. 
% Generally, symmetry implies  similarity in  behavior under similar conditions, e.g. states or input. 
% This similarity of behaviors eases the analysis and control of the systems that posses it. 
%This capability of symmetry maps to relate different trajectories of the system allows 
% Symmetry plays a fundamental rule in the analysis of dynamical and control systems.
This property has been used for studying stability of feedback systems~\cite{symmetryandstabilityfeedback}, designing observers \cite{bonnabel2008symmetry} and  controllers~\cite{controlledsymmetries_passivewalking}, analyzing neural networks \cite{Gerard2006NeuronalNA}, and deriving conservation laws~\cite{ConservationLawsFromNoethers} using Noether's theorem~\cite{Noether1918}. 

%Symmetry is mathematically defined via functions, called {\em symmetry maps}, which act on the continuous and discrete state spaces of the automaton.  These functions would transform behaviors or {\em trajectories} of the system to other behaviors. They would map a particular trajectory starting from a certain initial continuous state, or just {\em state}, and having a particular discrete state, or {\em mode}, to another one starting from different initial state and mode. 
%The behaviors that are transformed versions of each other under this map, are said to be {\em similar}.

Since both of $\xi_0$ and $\xi_0'$ are solutions of the system, and $\gamma_a$ can compute $\xi_0'$ from $\xi_0$, then, in a sense, $f$ has some redundancy. A simpler version of $f$, would only have $\xi_0$ as a solution, and allows us to derive $\xi_0'$ using $\gamma_a$. 
% A hybrid automaton is roughly a dynamical system $f$ transitioning between multiple $p$'s. 
In this paper, we create such simpler versions by defining symmetry-based abstractions for hybrid automata.

% If a dynamical system is symmetric, then some of its trajectories similarity of behavior implies redundancy in modeling
%, which calls for abstractions. 
%Since symmetry mathematically defines similarity, it would be a natural tool to create a simpler model with similar behaviors grouped in representative ones, and hence defining abstractions. In this paper, we utilize this fact and present a symmetry-based abstraction of hybrid systems. 

 %in modeling, which calls for abstractions.                                                                                                                                                                                                                                                                                                                                                                                                                                                                                                                                                                                                                                                                                                                                                                                                                                                                                                                                                                                                                            
Given a hybrid automaton $\ha$ having a set of discrete states (or {\em modes}) $P$ and a family of symmetry maps $\Phi$, our abstraction partitions $P$ to create the abstract automaton $\ha_v$. Each {\em equivalent class} of the partition is represented by a single mode in $\ha_v$. 
Any trajectory $\xi$ of any mode $p \in P$, can be transformed using $\Phi$ to get a trajectory $\xi_v$ of the representative abstract mode $p_v$, and vice versa. 
%
%In each equivalent set, the trajectories of the different modes are transformed versions of each other, via $\Phi$. 
%The trajectories of that representative mode can be transformed later to obtain trajectories of any of the other modes in the equivalent set. 
Accordingly, all concrete edges between any two equivalent mode classes would be represented with a single abstract edge. 
A set of concrete edges represented with the same abstract edge forms an equivalent edge class. 
The edges of $\ha$ are annotated with {\em guards} and {\em resets}. These dictate when the discrete transitions over the edges can be taken and how the state would be updated, respectively. The abstraction transforms the guards and resets of all concrete edges using $\Phi$. Then, it unions
%to be part of the guard and reset of the representative abstract edge. 
all of  the transformed guards and resets of an edge equivalent class to get the guard and reset of the corresponding abstract edge. This means that an execution of $\ha_v$ would transition over an edge $e_v$ if any of the transformed guards of the edges of $\ha$ that $e_v$ represents is satisfied. Moreover, the execution would split into several executions after a reset. Each of these executions start from a transformed version of the state defined by the reset of an edge of $\ha$ that is represented by $e_v$. We establish the soundness of this abstraction using a FSR (Theorem~\ref{thm:fsr_concrete_virtual}).

With several examples related to vehicles, we show that symmetry abstractions are natural. The abstraction can be useful for solving several problems related to tractability of synthesis and verification. In this paper, we focus on a particular application that is {\em reachability analysis}.
Our abstraction accelerates reachability analysis and enables unbounded time safety verification because 
%modeling waypoint-following agents, such as robots and drones. 
 $\ha_v$ has fewer modes than $\ha$. 
% Moreover, neither the dimensions of the mode and state spaces, nor the complexity of the dynamics, changes in the abstraction. 
% In general, that makes the set of reachable states of $\ha_v$ to be smaller than that of $\ha$, and easier to compute.  
%
For safety verification, the reachset of $\ha$ ($\Reach_{\ha}$) rather than that of $\ha_v$ is what matters. We show that $\Reach_\ha$ can be retrieved from $\Reach_{\ha_v}$ using $\Phi$ (Section~\ref{sec:real_from_virtual_reachset}). In fact, we show in our experiments in Section~\ref{sec:experiments}, that computing $\Reach_{\ha_v}$ and then transforming it with $\Phi$, is computationally less expensive than computing $\Reach_{\ha}$ directly. Since $\Reach_{\ha_v}$ is expected to be smaller than $\Reach_\ha$, its computation would reach a fixed point earlier than that of $\Reach_{\ha}$.
Moreover, $\Reach_{\ha_v}$ might be a bounded set when $\Reach_{\ha}$ is not. This property enables unbounded safety verification. Using our method, the safety verification problem of $\ha$ changes from computing $\Reach_{\ha}$ and checking if it intersects with an unsafe set $U$, to checking if there exists a map in $\Phi$ that transforms $\Reach_{\ha_v}$ to intersect $U$ (see Algorithm~\ref{code:unbounded}). 
The search over $\Phi$ for a map that transforms a bounded reachset $\Reach_{\ha_v}$ to intersect $U$ would be easier than computing an unbounded reachset of a nonlinear automaton $\Reach_{\ha}$, where the latter might not even be feasible.
\subsection{Summary of contributions}
\begin{itemize}
\item We introduce a new type of abstraction for hybrid systems based on symmetries (Definition~\ref{def:hybridautomata_virtual}) and explain its construction with examples (Examples~\ref{sec:single_linear_example},\ref{sec:transformation_example},\ref{sec:virtual_example},\ref{sec:single_robot_example_edges},\ref{sec:transformation_example_edges},\ref{sec:virtual_example_roads}).
\item We show that it is an abstraction using a FSR, and therefore, enjoys all the properties of standard abstractions (Theorem~\ref{thm:fsr_concrete_virtual}).
\item We show the practical advantage of this abstraction in accelerating bounded verification. We also show that it enables unbounded-time safety verification of $\ha$, using a data-structure called $\permodedict$, that stores the per-mode reachsets of $\ha_v$ (Algorithm~\ref{code:unbounded} and Theorem~\ref{thm:unboundedcodecorrectness}). 
\item We present an implementation of our abstraction construction, a fixed-point check on the reachset computation, and the construction of $\permodedict$.
\item We evaluate our implementation with experiments 
%and show the different aspects that affect our abstraction quality 
on a sequence of waypoint-following examples with scenarios having linear and nonlinear continuous dynamics, following different paths, using translation only or the combination of translation and rotation symmetries.
\end{itemize}

\subsection{Reachability and symmetry: brief literature review}

Reachability analysis is an essential tool in formal verification of hybrid systems.
Significant strides have been made in reachability analysis of continuous time dynamical and hybrid systems in the past decade. Linear dynamical models with thousands, and even millions, of dimensions have been verified~\cite{Spaceex,bak2017hylaa,verifbillion}. Nonlinear models of realistic systems ranging from engine control systems~\cite{NimaHSCC,Powertrain:JinDKUB14,FanQM18,DuggiralaFM015} to biomedical processes have been analyzed~\cite{KushnerBCFMS19,Sankaranarayanan17-Model-based,GrosuCAV}. Software tools to solve the reachability problem have been developed~\cite{CAS13,Althoff2015a,bak2017hylaa,C2E2paper,FanKJM16:Emsoft,FanM15:ATVA,CAS13,Chen2015ReachabilityAO}. These developments rely on advances in data-structures and dynamical systems theory results that exploit  characteristics like sparsity and stability. 
Exploiting symmetries using the abstraction presented here will add a new methodology to this toolbox.
% There are many reachability analysis tools, for example, that can  for the bounded-time and number of transitions case. Still, computing such over-approximations is expensive as it requires solving non-trivial optimization problems and integrating non-linear functions.

%, and more recently, accelerating bounded time safety verification of dynamical systems~\cite{Sibai:ATVA2019,Sibai:TACAS2020} and attempting to solve the unbounded safety verification problem~\cite{Sibai:TACAS2020}.

% It makes obvious sense to exploit this property for reachability analysis and the way to do that is also obvious: When a trajectory or a reachable set is computed from a set of initial states $K_0$ we should {\em cache\/} it. The next time the state $K_0 + a$ is visited in some phase of the reachability analysis, for any shift $a$, we should not compute the new reachset afresh but just retrieve the cached reachset, and apply the shift to it. 

% This very simple but powerful idea
Symmetry was used to  accelerate the safety verification of dynamical systems achieving promising results, in some cases by orders of magnitude speedups~\cite{Sibai:ATVA2019,Sibai:TACAS2020,maidens2018exploiting}.
% before their computation when the systems posses continuous symmetries such as translation invariance, which resulted in up to 1000$\times$ speedup in their safety verification. 
% In \cite{Sibai:TACAS2020}, we extended the result to the parameterized dynamical systems and the multi-agent setting. 
%They used symmetry transformations of the continuous dynamics to create a representative autonomous dynamical system for all of the parameterized dynamics, which they call the {\em virtual} system. They used the virtual system, along with symmetry, to cache and share reachsets between all the parameterized dynamics. 
In \cite{maidens2018exploiting}, Maidens et al. used the Cartan moving frame method for symmetric nonlinear discrete-time dynamical systems to move from absolute representation of states to relative ones. That resulted in orders of magnitude speed up on the backward reachable set computation problem for checking if two Dubin vehicles would collide. In our paper, we consider hybrid systems, with discrete and continuous dynamics, and reduce the number of modes of the automaton, instead of considering pure-discrete dynamics and reducing the dimension of the state space as in \cite{maidens2018exploiting}.
% How to apply the same idea to hybrid systems? The challenge here is that each mode of a hybrid system may have different symmetries. Different modes may be shifted and transformed versions of each other, and we would like to take advantage of the caching across modes. For example, a left-turn at a constant speed will be  symmetric with a right turn, but a will not be symmetric with a hard brake. Therefore, reachability in the right-turn mode should be able to exploit the cached left-turn, but not necessarily a brake. 
% The challenge here is the dependability of modes on each other: the initial set of a mode depends on the reachset of the previous one.
% How to apply the same idea to hybrid systems? 
% The work of Sibai et al.

In \cite{Sibai:ATVA2019}, we utilized symmetry of an autonomous dynamical system to cache and retrieve its reachsets using symmetry transformations. 
%These reachsets start from different parts of the state space, but from which the behavior of the system is similar.  
When given a continuous family of symmetry transformations, we were able to reduce the dimensionality of the reachsets that had to be computed, achieving orders of magnitude speedup in verification time. In contrast, this paper, we focus on hybrid systems instead of continuous dynamical systems. 
In \cite{Sibai:TACAS2020}, we extended the result of \cite{Sibai:ATVA2019} to the parameterized dynamical systems and multi-agent settings. That allowed us to tackle hybrid automata as well. We viewed the different modes of an automaton as parameterized dynamical systems. We constructed a representative mode $p^*$ of all the modes, using symmetry, and then used $p^*$ as a proxy to share reachsets across different modes using a cache.
% Moreover, we tackled the reachset computation of hybrid systems problem in the usual way. It is to sequentially compute the per-mode reachsets to construct the full reachset. Each time a mode is visited, its reachset is computed along with the set of initial states for the next mode. When caching is used, the cache is checked if it has a transformed version of the mode reachset, or part of it, that can be retrieved using symmetry instead of being computed. And whenever a reachset is computed, it is transformed to be a reachset of $p^*$ using symmetry and stored in the cache for future use.
 However, in \cite{Sibai:TACAS2020}, we did not develop the general abstraction-based view of symmetry (Definition~\ref{def:hybridautomata_virtual} and Theorem~\ref{thm:fsr_concrete_virtual}). As a result, we did not have a fixed point analysis (Theorem~\ref{thm:fixed_point_condition} and Corollary~\ref{cor:vir_dict_to_real_reachset}), and we were not able to verify unbounded-time safety properties.% Moreover, unlike  and .

\section{Model and problem statement}
\label{sec:model}

\paragraph*{Notations}
We denote by $\mathbb{N}$, $\mathbb{R}$, and $\mathbb{R}^{\geq 0}$ the sets of natural numbers, real numbers and non-negative reals, respectively. Given a finite set $S$, its cardinality is denoted by $|S|$. The length of a finite sequence $\mathit{seq}$ is denoted by $\mathit{seq.len}$ and its elements between $i$ and $j$, inclusive, by $\mathit{seq}[i:j]$. Given $N \in \mathbb{N}$, we denote by $[N]$ the set $\{0, \dots, N-1\}$. 
%Given a vector $v \in \mathbb{R}^n$ and a set of indices $L \subseteq [n]$, we denote the restriction of $v$ to the indices in $L$ by $v[L]$. 
Given two vectors $v \in \mathbb{R}^n$ and $u \in \reals^m$, we define $[v,u]$ to be the vector of length $n+m$ that results from appending $u$ to $v$.
Given $\epsilon \in (\reals^{\geq 0})^n$, we denote by $\ball(v,\epsilon)$ the $n$-dimensional hyper-rectangle centered at $v$ with $i^{\mathit{th}}$ dimension sides having length $\epsilon[i]$, for all $i \in [n]$. 
Given a hyper-rectangle $H \subseteq \mathbb{R}^n$ and a set of indices $L \subseteq [n]$, we denote the restriction of $H$ to the indices in $L$ by $H[L]$. 
We denote $\mathit{diag}(v)$ to be the diagonal matrix with diagonal $v$.
%  We define an $n$-dimensional hyper-rectangle by a 2d-array specifying its bottom-left and upper-right corners. 
% Given a  hyperrectangle $H$ and a set of indices specified by $L$,  the projection of $H$ to   $L$ is a $|L|$- dimensional hyperrectangle, which we denote as  $H[L]$.
%
%
Given a function $\gamma: \mathbb{R}^n \rightarrow \mathbb{R}^n$ and a set $S \subseteq \mathbb{R}^n$, as usual, we lift the function to subsets of $\reals^n$, and define 
$\gamma(S) = \{ \gamma(x)\ |\ x \in S\}$. 
We also lift it to vectors of $\reals^n$ and define $\gamma(v) = [\gamma(v_1), \gamma(v_2), \dots, \gamma(v_k)]$, for any $v \in \mathbb{R}^{n \times k}$.
We define $\arctan_2(y,x)$ to be the phase of the complex number $x + j y$.

\subsection{Hybrid dynamics}
\label{sec:agent:model}

In this paper, we will use a standard hybrid automaton modeling framework for defining cyber-physical systems~\cite{ACHH93,TIOAmon,Mitra07PhD}. 

\begin{definition}
\label{def:hybridautomata}
A {\em hybrid automaton\/} is a  tuple $\ha = \langle \stateset, P, \initset, \initmode, \edgeset, \guard, \reset, f \rangle$, where
\begin{enumerate}[label=(\alph*)]
%\item $V = X \cup \{loc\}$ is a set of variables. Here $loc$ is a discrete variable of finite type $Loc$. Valuations of $loc$ are called {\em locations}. Each $x \in X$ is a continuous variable of type $\mathbb{R}$. Elements of $val(V)$ are called states.
\item $\stateset \subseteq \reals^n$ is the continuous state space and $P$ is a (possibly infinite) set of discrete states. Continuous states are simply called {\em states} and the discrete ones are called {\em modes} or {\em parameters}.
\item $\initset \subseteq \stateset$ is a set of possible initial states and $\initmode \in P$ is the initial mode,
%\item $\timebound: P \rightarrow \mathbb{R}^+$ is an upper bound on the time that can be spent in each mode, 
\item $\edgeset \subseteq P \times P$ is the set of directed edges over modes that define mode transitions, 
%where a transition $(p_1, p_2) \in \edgeset$ is written as $p_1 \rightarrow p_2$,  \label{item:def_real_edge_set}
\item $\guard: \edgeset \rightarrow 2^{\stateset}$ gives the set of states from which an edge transition is enabled,  \label{item:def_guard}
\item $\reset: \stateset \times \edgeset \rightarrow 2^\stateset$ gives the updated (post) state after a transition is taken, \label{item:def_reset} and
%\item $\disctrans \subseteq (S \times P) \times (S \times P)$ is the set of discrete transitions. A transition $((s_1,p_1),(s_2, p_2)) \in \disctrans$ is written as $(s_1,p_1) \rightarrow (s_2, p_2)$, and
%The discrete transitions are described by finitely many guards and reset maps involving $s \in S$ and $p \in P$.
\item $f: \stateset \times P \rightarrow \stateset$ is a {\em dynamic function\/} that defines the continuous evolution. It is Lipschitz continuous in the first argument.
%\item %
%\item %
% It defines how the state evolves continuously with time for any given mode $p \in P$.
%\item $\trajs$ is a set of trajectories for $x$ which is closed under suffix, prefix and concatenation (see \cite{HuangM:HSCC2014} for details). For each $p \in P$, a set of trajectories $\trajs_p$ for parameter $p$ are specified by differential equation (\ref{sys:input}) and an invariant $I_l \subseteq val(X)$. Over any trajectory $\tau \in \trajs_l$, $\loc$ remains constant and the variables in $X$ evolve according to $E_l$ such that for all $t$ in the domain of $\tau$, $\tau(t)$ satisfies the invariant $I_l$.
\end{enumerate}
\end{definition}
%Now we define few functions that will be useful for our analysis throughout the paper.

% We say $(p_1,p_2)$ belongs to the edge set $\edgeset$ if and only if $\guard(p,p') \neq \emptyset$. 

% We denote the set of modes that appear in $\edgeset$ by $N$ and call it the set of {\em nodes}.  

% Now that we defined the automaton, we define what its executions look like.

%\begin{definition}[semantics]
%\label{def:ha_exec}
The dynamic function $f(\cdot,p)$ 
%and the $\timebound(p)$  
define the continuous state evolution in each mode $p\in P$. 
A function $\xi: \stateset \times P \times \mathbb{R}^{\geq 0} \rightarrow \stateset$ is a {\em trajectory of $\ha$} if $\xi$ is differentiable in its third argument, and given an initial state $\stateinstance_0 \in \stateset$ and a mode $p \in P$, $\xi(\stateinstance_0,p,0) = \stateinstance_0$ and for all $t \in \reals^{\geq 0}$, % $ [0, \timebound(p)]$,
\begin{align}
\label{sys:input}
\frac{d }{d t}\xi(\stateinstance_0,p,t) = f(\xi(\stateinstance_0,p,t), p).
\end{align}
The trajectory $\xi$ is the unique solution of (\ref{sys:input}) starting from $x_0$, since $f$ is Lipschitz continuous.
We say that $\xi(\stateinstance_0,p,t)$ is the state of (\ref{sys:input}) at time $t$  starting from $\stateinstance_0$ in mode $p$. When $x_0$ and $p$ are clear from context, we denote $\xi(x_0,p,t)$ by $\xi(t)$, for simplicity.
%\sayan{How is this defined if $t$ is a transition time?}
For any time-bounded trajectory $\xi$, i.e. defined over a finite interval in the third argument,
 $\dur(\xi)$ is its last time point. The first and last state in such a trajectory are denoted by $\xi.\mathit{fstate}$ and $\xi.\mathit{lstate}$, respectively. 

The edge set $E$, the $\guard$, and the $\reset$ together define the discrete transitions. For simplicity, 
for an edge $e = (p,p') \in E$, we denote its source mode $p$ by $\edgeinstance.\src$ and its destination mode $p'$ by $\edgeinstance.\dest$.
A sequence of modes $p_0, p_{1}, \dots$, where for all $i\geq 0$, $(p_i,p_{i+1}) \in \edgeset$ is called a {\em path}.
Moreover, we abuse notation and denote $\guard((p,p'))$ by $\guard(p,p')$. Then, $\guard(p,p')$ is the set from which a transition from mode $p$ to mode $p'$ is possible. From a state $x \in \guard(p,p')$, the post-state $x'$ after the transition has to be in $\reset(x,(p,p'))$. Such state-mode pairs $((x,p),(x',p'))$ define the transitions of $\A$ and we write  $(\stateinstance,p) \rightarrow (\stateinstance' ,p')$. 
 Note that there are no  urgent transitions here, and guards may be ignored.
 % up to a point\footnote{This is a simplification for analysis but would not hurt generalizability as the same mode can be repeated several times. The results of the paper extend naturally to models with  urgent transitions.}.  
 %The system may continue to evolve in a given mode $p \in P$ beyond a state $\stateinstance$ that satisfies $\guard(p,p')$, 
 
 %However, once the total time in mode $p$  reaches $\timebound(p)$, the trajectory must stop. 
 
 % that is, if \exists t \in \xi.dom, such that t = \tbound(p), then t = \xi.\ltime. 

 %A sequence of modes $\{p_i\}_i$, where for all $i\geq 0$, $(p_i,p_{i+1}) \in \edgeset$ is called a {\em path} and the set of modes in $P$ that appear in any path by $P_\edgeset$. Moreover, for any $p \in P_\edgeset$, the set of all modes $p' \in P_\edgeset$, where $(p',p) \in \edgeset$, is denoted by $\inmodes(p)$.

The semantics of a hybrid automaton is defined by executions which are  sequences of trajectories and transitions. 
An {\em execution\/} of $\ha$ is a sequence of pairs of trajectories and modes $\sigma = (\xi_0,p_0), (\xi_1,p_1), \dots $, where 
\begin{inparaenum}[(a)]
	\item each  $\xi_i$ is a trajectory of $\ha$, 
	\item each $(\xi_i.\lstate,p_i)\rightarrow (\xi_{i+1}.\fstate,p_{i+1})$ is a transition as defined above. 
% WE defined all of this above, why repeat and complicate.
%	\item edge $\edgeinstance_i = (p_i, p_{i+1}) \in \edgeset$,
%	\item  $\xi_i.\mathit{lstate} \in \guard(e_i)$, 
%	\item $\dur(\xi_i) \leq \timebound(p_i)$, and
%	\item $\xi_{i+1}.\mathit{fstate} \in \reset(\xi_{i}.\mathit{lstate}, e_i)$.  
\end{inparaenum}
A {\em finite} and {\em time-bounded} execution has a finite number of discrete transitions and all of its trajectories are time-bounded.
The duration of a finite and time-bounded execution $\sigma = (\xi_0,p_0), (\xi_1, p_1) \dots (\xi_k, p_k)$ is $\dur(\sigma) = \sum_{i}\dur(\xi_i)$ and its last state is $\sigma.\lstate = \xi_{k}.\lstate$. 

Finally, 
%for a fixed path $\path$, $\exec_{\ha,\path}$ is the corresponding set of executions of $\ha$.  
fix $J \in \mathbb{N}$;
$\exec_\ha(J)$  is the set of all executions of $\ha$ with at most $J$ transitions; When the transitions are unbounded, it is denoted by 
$\exec_\ha$.
We define the set of {\em reachable states\/} as:
\begin{align}
\reachset_\ha = \{\stateinstance \in \stateset \ |\ \exists\ \sigma \in \exec_\ha, \sigma.\lstate= x\}.
\end{align}
$\reachset_\ha(J)$ is the reachset restricted to executions with at most $J$ transitions.

\begin{example}[Robot following waypoints]
	\label{sec:single_linear_example}
	% modified from~\cite{Sibai:TACAS2020}
	Consider a robot following a sequence of waypoints $\{w_i \in \reals^2 \ |\ i\in [4]\}$ on the plane connected with directed roads $\{r_i \in \reals^4\ |\ i \in [5]\}$ forming an axis-aligned rectangle centered at the origin (see Figure~\ref{fig:problem_description}).
	% The path has horizontal roads, $r_1$ and $r_3$, of length 5 meters and vertical roads, $r_2$ and $r_4$, of length 3 meters, as shown in Figure~\ref{fig:problem_description}. 
	The robot starts from an arbitrary point in some initial set $\stateset_\init \subset \mathbb{R}^3$. We fix rectangles with dimensions $\epsilon_0 \in (\reals^{\geq 0})^2$ or $\epsilon_1 \in (\reals^{\geq 0})^2$.
	% , an axis-aligned square with length 1 meter sides centered at $[-4.5,-2.5]$, with an orientation between $-\pi/2$ and $0$, with respect to the $x[0]$-axis.
	 We say that it reached the first waypoint $w_0$ following road $r_0$ if it is located in the rectangle $\ball(w_0, \epsilon_0)$. If it was following road $r_1$ instead, we say it reached $w_0$ if it is located in the smaller rectangle $\ball(w_0, \epsilon_1)$.
	 % ellipse $\ball(w_0, [0.8, 0.6])$,
	 Moreover, for any $i \in \{1,2,3\}$, we say that it reached the waypoint $w_i$ following road $r_i$ if it is located in the rectangle $\ball(w_i, \epsilon_1)$.
	 % ellipse $\ball(w_i, [1, 0.8])$.
	
	To formalize the dynamics of the robot in this scenario, we construct a corresponding hybrid automaton. We use the four waypoints as four modes of the automaton, i.e. $P = \{w_1, w_2, w_3, w_4\}$. Consequently, whether starting from $\initset$ or coming from $w_3$, i.e. following road $r_0$ or road $r_4$, the robot would be in the same mode which corresponds to following $w_0$. 
	%Hence, a single edge in the automaton $(w_0,w_1)$ would represent the two cases, as shown in Figure~\ref{fig:waypoint_modes_original_state_machine}. 
	Thus, one should not confuse the roads of the path in Figure~\ref{fig:problem_description} with the edges of the automaton that we will construct.
	Now, we fix $\epsilon_0 = [1, 1.4]$ and $\epsilon_1 = [0.6,1]$. The resulting hybrid automaton is shown in Figure~\ref{fig:waypoint_modes_original_state_machine} and would be formally defined as: $W = \langle \stateset, P, \initset, \initmode,$
	 %\timebound,
	$  \edgeset, \guard, \reset, f \rangle$: %in Definition~\ref{def:hybridautomata} notation:
	\begin{enumerate}[label=(\alph*)]
		\item $\stateset \subseteq \reals^3$, representing the position and orientation with respect to the $x[0]$-axis, and $P = \{p_{i} = w_{i}\ |\ i\in [4] \}$,
		% and $\pd = \pout = \pin \textbackslash \{(-5,-5)\} \cup \{(-5,5)\}$.  
		% Continuous states are simply called {\em states} and the discrete ones are called {\em modes} or {\em parameters}.
		% \item $\timebound(p_{i})$ equals 5 seconds for $i \in \{0,2\}$ and equals 10 seconds for $i \in \{2,3\}$, 
		\item $\initset = \ball([-4.5,-0.5, -\frac{\pi}{4}], [0.8,0.8, \frac{\pi}{2}])$, $\initmode = p_0 = w_0$,
		\item $\edgeset = \{e_0=(p_{0},p_{1}), e_1 = (p_{1}, p_{2}), e_2 = (p_{2}, p_{3}), e_3 = (p_{3}, p_{0}) \}$, 
		%where a transition $(p_1, p_2) \in \edgeset$ is written as $p_1 \rightarrow p_2$,  \label{ite:def_real_edge_set}
		\item 
		\begin{align*}\guard(e_i) =
		 \begin{cases}
		 (\ball(w_{i}, \epsilon_0) \cup \ball(w_{i}, \epsilon_1)) \times \reals, \text{ if } i=0, \\
		 \ball(w_{i}, \epsilon_1) \times \reals, \text{ if } i=\{1,2,3\},
		 \end{cases} 
		\end{align*} 
		\label{eg_item:def_guard}
		\item $\forall \stateinstance \in \stateset, \edgeinstance \in \edgeset, $
		\begin{align}\reset(\stateinstance,\edgeinstance) = \{\stateinstance\},
		\end{align} 
		is the identity map, \label{eg_item:def_reset} and
		\item $\forall \stateinstance \in \stateset, \forall p \in P$,
		\begin{align}
		\label{eq:robot_dynamics}
		f(x, p) = \frac{dx}{dt} = 
		\left[
		\begin{matrix}
		v \cos (\stateinstance[2]) \\
		v \sin (\stateinstance[2])\\
		2v \sin (\alpha) / L
		\end{matrix}
		\right], \text{where}
		\end{align}
		%&\frac{dx[0]}{dt} = v \cos \theta, \frac{dx[1]}{dt} = v \sin \theta , \frac{dx[2]}{dt} = 2v \sin (\alpha) / L, \\
		$\alpha = \arctan_2(p[1] - x[1],p[0] - x[0]) - x[2]$ and
		$v$ and $L$ are the fixed speed and length of the robot~\cite{BLAZIC20111001}.
		 %f(x,p) =  A (\stateinstance - p)$, where $A = [[-3, 1], [1, -2]]$ is a stable matrix.
	\end{enumerate}
% Moreover, the initial set $K$ is $\ball([-4.5,-0.5], [0.5,0.5])$ which is the square with sides $0.5$ centered at $[-4.5,-0.5]$ in the first two coordinates and the interval $[-\pi/2, 0]$ in the third one.  Finally, the initial mode is $p_0$.
 %Figure~\ref{fig:waypoint_modes_original_state_machine}, represents the state machine representing the transitions between the different modes of the system
% In Figure~\ref{fig:waypoint_modes_original}, the different colors of the waypoints means that they represent the different modes. The color of the perimeter of a guard, whether a square or an ellipse, denotes the color of the source mode of the edge it corresponds to. Moreover, the color of a guard interior represents the color of the destination mode of the edge it corresponds to. Note that if the robot is following the orange waypoint, whether it was coming from the blue waypoint or starting from the initial set, it would be in mode $p_0$. That is why both the rectangle and the ellipse around the orange waypoint have an orange perimeter and violet interior. Finally, Figure~\ref{fig:waypoint_modes_original_state_machine}, represents the state machine representing the transitions between the different modes of the system.
\end{example}

\begin{comment}
\vspace{\floatsep}
\begin{subfigure}[t]{0.5\textwidth}
\centering
\includegraphics[width=\textwidth]{Figures/waypoint_modes/waypoint_modes_original.png} 	
\caption{ \label{fig:waypoint_modes_original}}
\end{subfigure}
\end{comment}

\begin{figure}[t!]
	%\caption{Reachtubes for drone~\ref{fig:linear_nosym_1_agents_sym} and linear~\ref{fig:linear_nosym_3_agents_sym} models  using {\sf Sym-Flow*}. Three agents \label{fig:flow_linear}}
	\centering
	\begin{subfigure}[t]{0.5\textwidth}
		\centering
		\includegraphics[width=\textwidth]{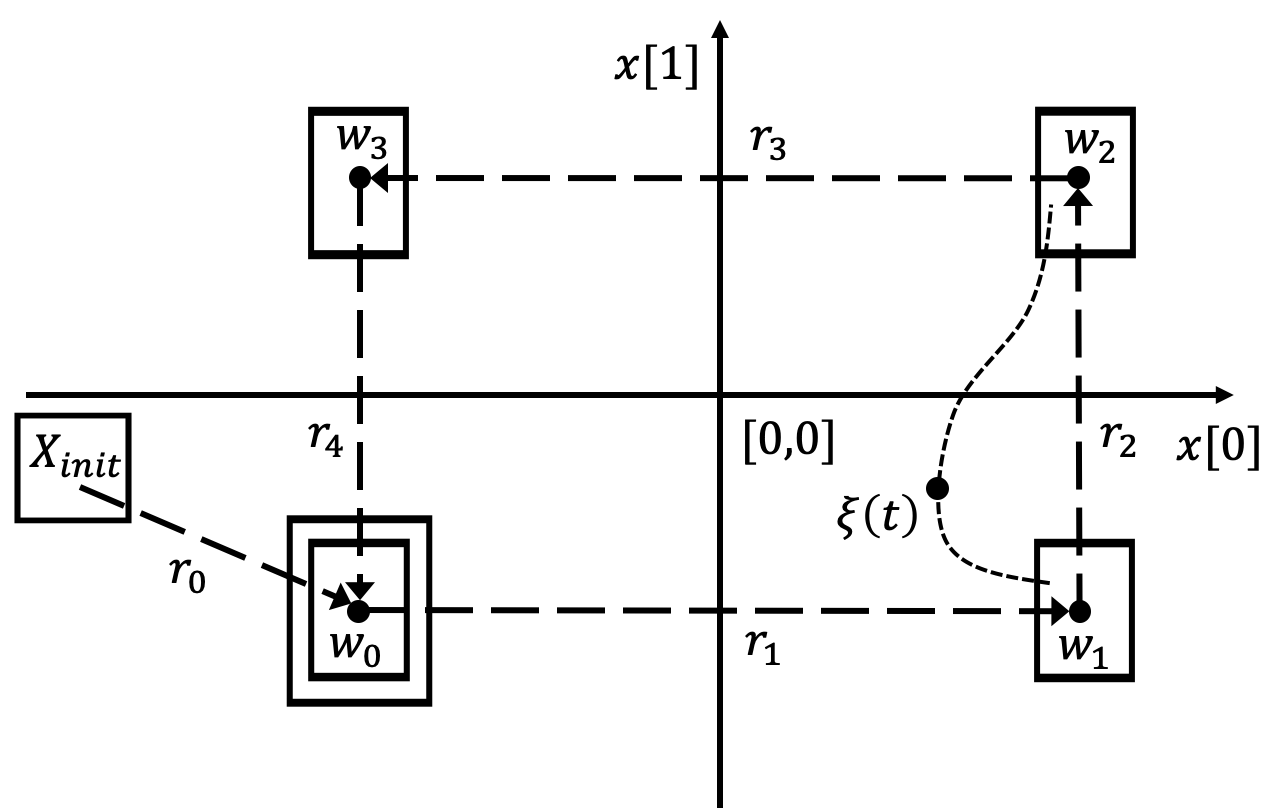} 
		\caption{ \label{fig:problem_description}}
		\vspace{\floatsep}
	\end{subfigure}
	\begin{subfigure}[t]{0.5\textwidth}
		\centering
		\includegraphics[width=\textwidth]{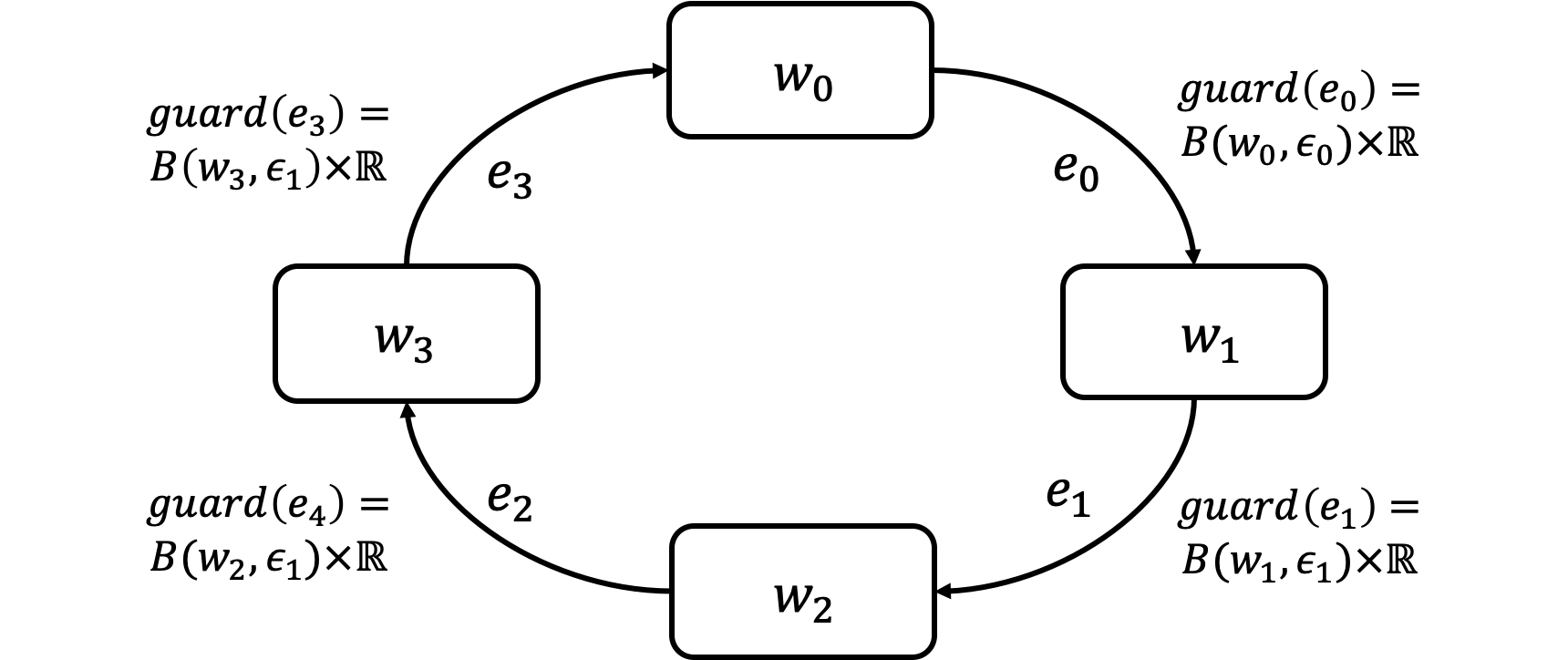}
		\caption{\label{fig:waypoint_modes_original_state_machine}}
		\vspace{\floatsep}
	\end{subfigure}
	\caption{\scriptsize (\ref{fig:problem_description}) A robot, with a state containing its position and orientation, following a sequence of 2D waypoints forming a rectangle starting from its initial set $\initset$. It reaches a waypoint if it reaches the rectangle centered at it. It has to reach the larger rectangle centered at $w_0$ when starting from $\initset$. (\ref{fig:waypoint_modes_original_state_machine}) The state machine representing the discrete transitions of the hybrid automaton $W$ describing the scenario in (\ref{fig:problem_description}). The resets are omitted since they are the identity map for all the edges. \label{fig:waypoint_modes}
	}
\end{figure}

\section{Symmetry and equivariant dynamical systems}
\label{sec:symdef}

%\chuchu{chuchu: can you use English to describe what symmetry means for controlled systems?}
 In this section, we present an existing definition of symmetry for dynamical systems with parameters and a sufficient condition for a map to be a symmetry.

 A symmetry map $\gamma$ acts on the state space $\stateset$, i.e. $\gamma: \stateset \rightarrow \stateset$, such that given a solution of the system, it maps it to another valid solution.
 % This is formalized in the following definition.

\begin{definition}[Definition 2 in \cite{russo2011symmetries}]
	\label{def:symmetry}
Let $\Gamma$ be a group of maps acting on $\stateset$. We say that $\gamma \in \Gamma$ is a symmetry of (\ref{sys:input}) if it is differentiable, invertible, and for any solution $\xi(x_0,p, \cdot)$, $\gamma (\xi(\stateinstance_0,p, \cdot))$ is a solution as well.
 % Furthermore, if $\gamma  (\xi(\state_0,p,t)) = \xi(\stateinstance_0,p,t)$ for any $t$, we say that the solution $\xi(x_0,p,\cdot)$ is $\gamma$-symmetric.
\end{definition} 

Fortunately, a map can be checked if it is a symmetry of a system by analyzing whether it commutes with its dynamic function. 
\begin{definition}[\cite{russo2011symmetries}]
	\label{def:equivariance_input}
	The dynamic function $f: \stateset \times P \rightarrow \stateset$ is said to be $\Gamma$-equivariant if for any $\gamma \in \Gamma$, there exists $\rho: P \rightarrow P$ such that, 
	\begin{align}
	\label{eq:equivariance_condition}
	 \forall\ \stateinstance \in \stateset,\forall\ p \in P,\ \frac{\partial \gamma}{\partial \stateinstance} f(\stateinstance, p) = f(\gamma(\stateinstance), \rho(p)).
	\end{align}
\end{definition}
	%\chuchu{ chuchu: this is sentence is not readable: Moreover, the parameter can be a varying input signal $u: \mathbb{R}^{\geq 0} \rightarrow P$ instead of a constant $p \in P$ and $\rho$ can be non-linear and a function of the state $\stateinstance$ and time $t$, i.e. $\rho : P \times S \times \mathbb{R}^{\geq 0} \rightarrow P$ and the theorem would still hold.}
%\chuchu{chuchu: $\Gamma$ above is just a group of smooth operators, but now it is the class of symmetry operators. I am a bit confused now. Is there a way to make it consistent? Also, why did not $\rho$ show up in Definition 4?}

The following theorem shows that it is enough to check the condition in equation (\ref{eq:equivariance_condition}) to prove that a map is a symmetry. %of~(\ref{sys:input}). 
% In this paper, we fix a set of maps $\Gamma$ and assume that the dynamic functions of the hybrid automata that we consider are $\Gamma$

\begin{theorem}[part of Theorem 10 of \cite{russo2011symmetries}]
	\label{thm:sol_transform_input_nonlinear}
	 If $f$ is $\Gamma$-equivariant, then all maps in $\Gamma$ are symmetries of (\ref{sys:input}). Moreover, for any $\gamma \in \Gamma$, map $\rho: P \rightarrow P$ that satisfies equation (\ref{eq:equivariance_condition}), $x_0 \in \stateset$, and $p \in P$, $\gamma(\xi(x_0, p,\cdot)) = \xi(\gamma(x_0), \rho(p), \cdot)$.  
\end{theorem} 

Note that if $\gamma$ in Theorem~\ref{thm:sol_transform_input_nonlinear} is a linear function of the state, i.e. $\gamma(x) = A x$, for some $A \in \mathbb{R}^{n \times n}$, the condition in equation~(\ref{eq:equivariance_condition}) for equivariance becomes $\gamma(f(\stateinstance,p)) = f(\gamma(\stateinstance),\rho(p))$.

\begin{example}[Robot origin translation symmetry]
	%, borrowed and modified from \cite{Sibai:TACAS2020}]
	\label{sec:transformation_example}
	Consider the robot presented in Example~\ref{sec:single_linear_example}, and maps $\gamma_\origtrans$ and $\rho_\origtrans$ that translate the origin of the plane to a new origin $p^* \in \mathbb{R}^2$.
	Let $\gamma_{\origtrans} : \mathbb{R}^3 \rightarrow \mathbb{R}^3$ and $\rho_{\origtrans} : \mathbb{R}^2 \rightarrow \mathbb{R}^2$ be defined as:
	\begin{align}
	\gamma_{\origtrans}(x) &=  [x[0:1] - p^*, x[2]], \\ % {\bf R}
	\rho_{\origtrans}(p) &=  (p - p^*), % {\bf R}
	\end{align}
	\begin{comment}
	where \begin{align}
	{\bf R} =
	\left[
	\begin{matrix}
	&\cos(\theta) &\sin(\theta) \\
	&-\sin(\theta) &\cos(\theta)
	\end{matrix}
	\right]
	\end{align} 
	
	is the rotation matrix with angle  $\theta = \arctan_2( p[0] -  \mathit{src}[0], p[1] - \mathit{src}[1])$.
	\end{comment} 
	% \gamma(\stateinstance) &= [(\stateinstance[0] - p[0])\cos(\theta) - (\stateinstance[1] - p[1])\sin(\theta),\nonumber\\
	% &\hspace{0.5in} -(\stateinstance[0] - \mathit{src}[0])\sin(\theta)  + (\stateinstance[1] - \mathit{src}[1])\cos(\theta)] \text{ and }\\
	%\begin{align}
	%\rho(p) &= [(p[0] - \mathit{src}[0])\cos(\theta) + (p[0] - 
	% &\hspace{0.5in} -(p[0] - \mathit{src}[0])\sin(\theta) + (p[1] - \mathit{src}[1])\cos(\theta)],
	% \end{align}
	% Observe that $\gamma(x)$ only changes $x[0:1]$ while leaving $x[2]$ unaffected. But, in equation~(\ref{eq:robot_dynamics}), the only term that depends on $x[0:1]$ is $\alpha$.
	
	Then, $\frac{\partial \gamma_{\origtrans}}{\partial x} f(x, p) = \frac{\partial \gamma_{\origtrans}}{\partial x} \frac{dx}{dt} = I_3 \times f(x,p) = f(x,p)$, where
	$\alpha$ is as in equation~(\ref{eq:robot_dynamics}) and $I_3$ is the $3 \times 3$ identity matrix. Moreover, 
	\begin{align}
	f(\gamma_{\origtrans}(x), \rho_{\origtrans}(p)) =
	\left[
	\begin{matrix}
	v \cos (\stateinstance[2])\\
	v \sin (\stateinstance[2])\\
	2v \sin (\alpha') / L
	\end{matrix}
	\right], \text{where}
	\end{align}
	$\alpha' = \arctan_2(p[1] - p^*[1] - (x[1] - p^*[1]), p[0] - p^*[0] - (x[0] - p^*[0])) - x[2] = \alpha$. 
	Then, for all $\stateinstance \in \stateset$ and $p \in P$, $\frac{\partial \gamma}{\partial x} f(x, p) = f(\gamma_{\origtrans}(\stateinstance), \rho_{\origtrans}(p))$, and $\gamma_{\origtrans}$ is a symmetry of $f$.
	% The transformation $\gamma$ would translate the origin of $\stateset$ from zero to $[p^*[0], p^*[1], 0]$. 
	% Then, it would rotate its axes counter-clockwise by $\theta$, so that the $x$-axis is aligned with the segment connecting $\src$ and $p$ waypoints. 
	% Moreover, $\rho$ would translate the origin of the parameter space $P$ to $p^*$.
	%
	Figure~\ref{fig:waypoint_modes_transformation} shows the new state $\xi(t)-p^*$ and new parameter $w_2 - p^*$ representing mode $p_{2}$ of the robot of Figure~\ref{fig:problem_description} after translating the origin to $p^*$ using $\gamma_{\origtrans}$ and $\rho_{\origtrans}$ of this example.
	% For the aircraft, this means translating and rotating the plane so that the segment connecting the aircraft and the waypoint positions reside would be the $y$-axis. 
	% If translation would be considered alone, $[\stateinstance[0], \stateinstance[1]]$ would be left unchanged and  $[\stateinstance[2], \stateinstance[3]]$ as well as $[p[2], p[3]]$, would be translated by $[-\mathit{src}[0],-\mathit{src}[1]]$.
	
	% SIMPLIFICATION
	% We use the same transformations for the linear example. The only difference is that it has no heading angle. Hence, in case of translation, the first two coordinates are translated by $[-\mathit{src}[0],-\mathit{src}[1]]$, and in case of rotation, the axes of these two coordinates get rotated as the last two state components of the aircraft. The third coordinate is left intact.
	
	%In \cite{Sibai:TACAS2020}, the authors set the first two coordinates of $\rho(p)$ to zero, independent of $p$. Here we use them to represent the relative position of the src waypoint with respect to the destination one.
\end{example}

\begin{figure}[t!]
	%\caption{Reachtubes for drone~\ref{fig:linear_nosym_1_agents_sym} and linear~\ref{fig:linear_nosym_3_agents_sym} models  using {\sf Sym-Flow*}. Three agents \label{fig:flow_linear}}
	\centering
		\includegraphics[width=0.5\textwidth]{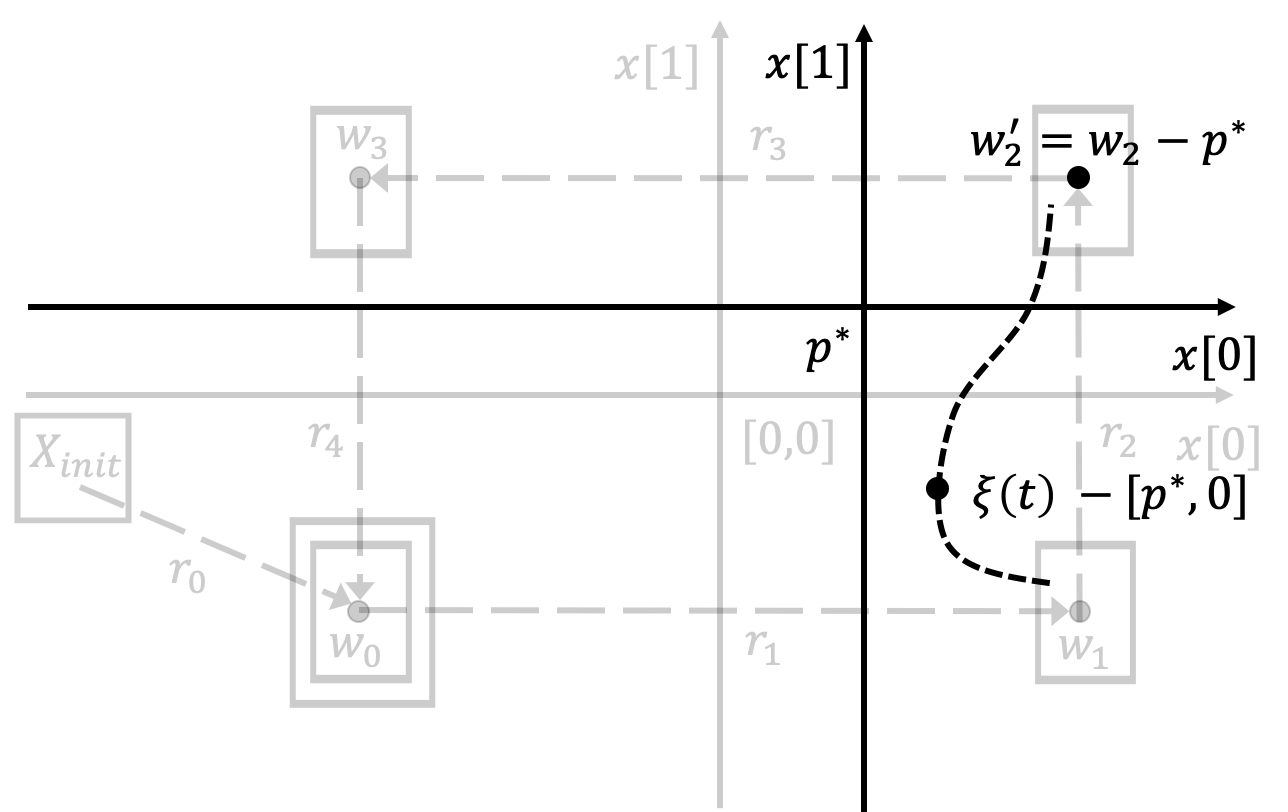} 
		\caption{\scriptsize Changing the origin of the coordinate system to $p^*$ in Figure~\ref{fig:problem_description}, does not affect the intrinsic behavior of the robot, but only translates the states in its trajectories. Such a translation is a valid symmetry of the robot dynamics. \label{fig:waypoint_modes_transformation}}
\end{figure}

\section{Virtual or abstract hybrid automaton}

%The challenge in safety verification of multi-agent systems is that the dimensionality of the problem grows too rapidly with the number of agents to be handled by any of approach. However, often agents share the same dynamics. 
%
%For instance, several drones of the same model and same manufacturer may behave similarly while having different initial conditions and follow different waypoints. 
%This commonality has been exploited in developing specialized proof techniques~\cite{JM:2012:small}. For reachability analysis, using symmetry transforms of the previous section, reachtubes of one agent in one mode can be used to get the reachtubes of other modes and even other agents.
In this section, we present our symmetry-based abstraction
%\marginpar{\scriptsize{\sayan{*Do we establish this abstraction  somewhere? That will be nice to do, e.g., using a fwd sim rel? Abstraction give more than reachset containment.}}}
%
 of hybrid automata along with the corresponding FSR.
 % These symmetries would be provided by the user in this paper. 
 Our abstract automata have fewer numbers of modes and edges, than their concrete counterparts. % Consequently, our abstraction reduces the time and space complexities of the computations done for these automata.
 %We do that for the sake of making unbounded-time safety verification tractable.  
 
 Our abstraction is an extension of the concept of virtual system for parameterized dynamical systems that we defined in \cite{Sibai:TACAS2020}, to hybrid systems. 
 % In that sense, we borrow the abstraction of the continuous dynamics from \cite{Sibai:TACAS2020} and add that of the discrete transitions in this paper. 
 %We follow that notation and call the abstract model the {\em virtual} system and that of Definition~\ref{def:hybridautomata}, the {\em concrete} one. 
 Throughout the paper, we use subscript $v$ to denote the variables and functions of the abstract (or virtual) automaton and no subscript for those of the concrete one.

\subsection{Creating the virtual model}
\label{sec:virtualmodel}
In order to create the virtual model, a family of symmetries is needed.
This is formalized below.
\begin{definition}[virtual map]
	\label{def:virtual_map}
	Given a hybrid automaton $\ha$, a {\em virtual map} is a set  
	\begin{align}
	\Phi = \{(\gamma_p, \rho_p) \}_{p \in P},
	\end{align}
	where for every $p \in P$,
	$\gamma_p: \stateset \rightarrow \stateset$, 
	%is a symmetry of the dynamic function $f$ of $\ha$,
	 $\rho_p: \reals^m \rightarrow \reals^m$, and they satisfy equation~(\ref{eq:equivariance_condition}).
\end{definition}
% Now, we define the new model formally. 

%\sayan{I am finding this def hard to read. The first ``exists'' could be written as $\gamma_p$ is invertible? Satisfies Def 3, is same as saying $(\gamma_p, \rho_p)$ is equivariant? But equivariance was defined for a set of symmetries. Why the subscript $p$ for $\gamma$ and $\rho$? Actually $p$ is not a bound variable. What does $\rho_p(p) =p^*_v$ mean? What is the $(p)$?}
%
%

Given $\ha$ and a virtual map $\Phi$, 
% to $p^*$
 each $\gamma_{p}$ is called a {\em virtual state map\/} and each $\rho_{p}$ is called a {\em virtual mode map\/}. 
From Theorem~\ref{thm:sol_transform_input_nonlinear}, it follows that 
these maps transform trajectories in mode $p$ to trajectories in  mode
$\rho_p(p)$.
Using a virtual map $\Phi$ of $\ha$, we can define the function $\rv: \mathbb{R}^m \rightarrow \mathbb{R}^m$, where for all $p \in P$:
\begin{align}
 \rv(p) = \rho_{p}(p).
\end{align}

This function will be used in Definition~\ref{def:hybridautomata_virtual} to map concrete modes
to virtual ones. Moreover, its inverse will be used to map
virtual modes to the sets of concrete modes, or the equivalent mode classes, that they represent.

\begin{definition}[virtual model]
	\label{def:hybridautomata_virtual}
	Given a hybrid automaton $\ha$, 
	% an initial set of states $K_0 \subseteq \stateset$, an initial mode $p_0 \in P$,  
	and a virtual map $\Phi$, 
	% to $p^*$,
	 the resulting {\em abstract (virtual) hybrid automaton} is: 
	% \timebound_v,
	$$\ha_v = \langle \stateset_v, P_v, \initsetv, \initmodev,  
	\edgeset_v, \guard_v, \reset_v, f_v \rangle \text{, where } $$% \edgeset' = \edgeset \cup \{(p_{-1},p_0)\} \text{ and}$$
%\sayan{I am not changing the subscripts here...decide the notation and carry it forward.}
	\begin{enumerate}[label=(\alph*)]
	\item $\stateset_v = \stateset$ and 
	$P_v = \rv(P)$% \{ p_v \in P\ |\ \exists\ p \in P, p_v= \rho_{p}(p) \}$, 
	\label{item:def_virtual_mode_set}
	%\item for any $p_v \in P_v$, $\timebound_v(p_v) = \max\limits_{\substack{p \in P, \rho_{p}(p) = p_v}} \timebound(p), $\label{item:def_virtual_timebound}
	\item $\initsetv = \gamma_{\initmode}(\initset)$ and $\initmodev = \rv(\initmode)$, \label{item:def_virtual_initial}
	\item $\edgeset_v = \rv(\edgeset) = \{ (\rv(p_1),\rv(p_2))\ |\ \edgeinstance = (p_1,p_2) \in \edgeset \}$% \bigcup\limits_{(p_1,p_2)\in \edgeset} \{ (\rv(p_1), \rv(p_2)) \}$, 
	\label{item:def_virtual_edgeset}
	%
	%$\edgeset_v \subseteq P_v \times P_v$, where $(\rho_{p_2}(p_1), \rho_{p_3}(p_2)) \in \edgeset_v$ if $(p_1, p_2)$ and $(p_2,p_3) \in \edgeset$, 
	%\item $\guard_v(e_v) = \cup_{p_1,p_2,p_3 \in P}\ \bigg\{ \gamma_{p_2}\big(\resrelguard(p_1,p_2,p_3)\big)\big|
	%\rho_{p_2}(p_1) = p_{v,1}, \rho_{p_3}(p_2) = p_{v,2}\bigg\}$,
	%
	\item $\forall \edgeinstance_v \in \edgeset_v, $
	\begin{align*}
	\guard_v(e_v) = \bigcup\limits_{\edgeinstance \in \rv^{-1}(\edgeinstance_v)}  \gamma_{\edgeinstance.\src}\big(\guard(\edgeinstance)\big),
	\end{align*}
	%\begin{align*}
	%\begin{cases}
	%&
	
	%\bigcup\limits_{\substack{(p_1,p_2) \in \edgeset, \rho_{p_1}(p_1) = p_{v,1}, \rho_{p_2}(p_2) = p_{v,2}}}\   \gamma_{p_1}\big(\guard((p_1,p_2))\big),
	% , \text{ if } p_{v,1} \neq \bot, 
	% \text{ and} \\
	%&\bigcup\limits_{p_1\in P, \rho_{p_1}(p_0) = p_{v,2} } \gamma_{p_0}\big(\guard((p_0,p_1))\big), \text{ otherwise},
	%\end{cases}	
	%\end{align*}
	 \label{item:def_virtual_guard}
	%\item $\reset_v(\stateinstance_v, e_v) = \cup_{p_1,p_2,p_3 \in P}\ \bigg\{\gamma_{p_3}\big(\reset\big(\gamma_{p_2}^{-1}(x_v), e_2\big)\big)\big| \rho_{p_2}(p_1) = p_{v,1}, \rho_{p_3}(p_2) = p_{v,2}\bigg\}$, and
	% \item $D_v = \{ (\gamma_{p_2}(s_2), \rho_{p_2}(p_1), \gamma_{p_3}(s_4), \rho_{p_3}(p_2))\ |\ (s_1,p_1,s_2,p_2), (s_3,p_2,s_4,p_3) \in D, \exists\ t\leq T, s_3 = \xi(s_2, p_2, t)\}$, and
	\item $\forall \stateinstance_v \in \stateset_v, \edgeinstance_v \in \edgeset_v, $ 
	\begin{align*}
	\reset_v(\stateinstance_v, e_v) = \bigcup\limits_{\edgeinstance \in \rv^{-1}(\edgeinstance_v)}\ \gamma_{\edgeinstance.\dest}\big(\reset\big(\gamma_{\edgeinstance.\src}^{-1}(x_v), \edgeinstance \big)\big), \text{ and}
	\end{align*}
	%\begin{align*} 
	%\begin{align*}
	%\begin{cases}
	%&
	%
	% \text{ if } p_{v,1} \neq \bot, \text{ and} \\ 
	%&\bigcup\limits_{p_1 \in P, \rho_{p_1}(p_0) = p_{v,2}} \gamma_{p_1}\big(\reset\big(\gamma_{p_0}^{-1}(x_v), (p_0,p_1)\big)\big), \text{ otherwise.}
	% \end{cases}
	% \end{align*} 
	\label{item:def_virtual_reset}
	\item %$f_v: \stateset \rightarrow \stateset$, where $\forall\ \stateinstance \in \stateset,\ 
	$\forall p_v \in P_v$, $\forall \stateinstance \in \stateset, f_v(\stateinstance, p_v) = f(\stateinstance, p_v)$. \label{item:def_virtual_dynamics}
	\end{enumerate}
% Finally, the initial set of states and the initial mode of $\ha_v$ would be $K_{v,0} = \gamma_{p_0}(K_0)$ and $\rv(p_0)$, respectively.
% \{s_{v} \in \stateset\ |\ \exists\ s_{v,2} \in S_v, (s_{v,1}, p_{v,1},s_{v,2},p_{v,2}) \in \disctrans_v\}\nonumber \\
%  &= \{s_{v,2} \in S_v\ |\ \exists\ s_{v,1} \in S_v, (s_{v,1},p_{v,1},s_{v,2},p_{v,2}) \in \disctrans_v\} \nonumber \\
% the function $\resrelguard: P\times P \times P \rightarrow 2^{S}$ is defined as follows: 
%\begin{align}
%\resrelguard(p_1,p_2,p_3) = \Rtube(\reset(\guard(e_1),e_1),p_2,0) \cap \guard(e_2).
%\end{align}
\end{definition}  
%\end{tBox*}

% \resrelguard(p_1,p_2,p_3) = \Rtube(\reset(p_1,p_2),p_2,T,0) \cap \guard(p_2,p_3).
%\reset_v(\stateinstance_v, e_v) &= \cup_{p_1,p_2,p_3 \in P}\ \bigg\{\gamma_{p_3}\big(\reset\big(\resrelguard(p_1,p_2,p_3), e_2\big)\big)\ | \nonumber \\

%The virtual mode $\bot$ in the definition above is an assumed mode that the system was in before $p_0$ which resulted in $K_0$. It can be an arbitrary mode as long as it is not part of the paths of the concrete model, i.e. does not belong to $P_\edgeset$. It is introduced to have a mode in the virtual system corresponding to the initial mode of the concrete one.
%
The trajectories and executions of the virtual hybrid automaton $\ha_v$ are defined in the same way as in Definition~\ref{def:hybridautomata}.

\begin{example}[Robot virtual system]
	%, continuous dynamics part borrowed and modified from \cite{Sibai:TACAS2020} while we added set of virtual modes, virtual edges, guards, and resets]
	\label{sec:virtual_example}
	Consider the scenario described in Figure~\ref{fig:problem_description} and its corresponding hybrid automaton $W$ defined in Example~\ref{sec:single_linear_example}. To construct its virtual automaton, we need a virtual map $\Phi$ first.
	% and the corresponding transformation described in Example~\ref{sec:transformation_example}. 
	% Fix $p \in P$, we set $\src$ in the transformation of Example~\ref{sec:transformation_example} to $[p[2],p[3]]$ and $\theta$ to $\arctan_2(p[0] - p[2], p[3]-p[1])$ and let $\gamma_p$ and $\rho_p$ be the resulting transformations. 
	% Then, for all $p \in P$, $\rho_{p_2}(p_1) = [0,0, (p_1[2] - p_2[2])\cos(\theta) + (p_1[2] - p_2[3])\sin(\theta), -(p_1[3] - p_2[2])\sin(\theta)  + (p_1[3] - p_2[3])\cos(\theta)]$.
	% It changes the position of the destination waypoint $[p_1[2],p_1[3]]$ of the input mode $p_1$ to the relative position with respect to the coordinate system where $[p_2[2],p_2[3]]$ is the origin and the line connecting them is the $y$-axis.
	% Hence, $\forall p \in P$, $\rho_p(p) = [0,0,0,0]$. 
	% We choose $p^* = 0$.
	% and the virtual dynamics $f_v$ is that of Example~\ref{sec:single_linear_example} with the parameter $p$ being $p^*$. 
	% For any mode $p$ in the set of four modes $P$ representing the waypoints, we choose the virtual state and mode maps $\gamma_{p}$ and $\rho_{p}$ as in Example~\ref{sec:transformation_example} with $p^*$ being $p$. Hence, the maps translate the origin of the plane to the waypoint under consideration.
	For every $p \in P$, we define $\gamma_p$ and $\rho_p$ to be the origin translation maps $\gamma_{\origtrans}$ and $\rho_{\origtrans}$ that we presented in Example~\ref{sec:transformation_example}. Recall that in that example, we needed to define $p^*$ that we want to translate the origin to. For this example, for each $p \in P$, we choose $p^*$ to be equal to $p$.
	Figure~\ref{fig:waypoint_modes_virtual_transformation} shows a visualization of $\gamma_{p_2}$ and $\rho_{p_2}$. In that figure, the waypoint $w_2$, which is the concrete mode $p_{2}$, becomes the origin $p_v$ after applying $\rho_{p_{2}}$. Hence, $p_2$ gets represented in the virtual automaton by the virtual mode $p_v$, as shown in Figures~\ref{fig:waypoint_modes_virtual} and \ref{fig:waypoint_modes_virtual_statemachine}. 
	
	The set $\guard(\edgeinstance_2)$ of $\ha$, becomes $\ball(p_v, \epsilon_1) \times \mathbb{R}$, after applying $\gamma_{p_2}$. This rectangle will be part of the guard of the virtual mode $p_v$, as specified in Definition~\ref{def:hybridautomata_virtual} part~\ref{item:def_virtual_guard}. Its projection to the first two dimensions is shown as the rectangle $\ball(p_v,\epsilon_0)$ centered at the origin $p_v$ in Figure~\ref{fig:waypoint_modes_virtual}.
	
	Figure~\ref{fig:waypoint_modes_virtual_transformation} also shows that the rectangle $B(w_1,\epsilon_1)$, which is $\guard(\edgeinstance_1)[0:1]$, becomes $B(w_1-w_2, \epsilon_0)$ a rectangle centered at $w_1 - w_2$, after applying $\gamma_{p_2}$. Recall that the reset of any edge of $W$ is just the identity map. Hence, the rectangle centered at $w_1 - w_2$ represents $\gamma_{p_2}(\reset(\guard(e_1)), e_1)[0:1]$. This will be part of the set of possible reset states per part~\ref{item:def_virtual_reset} of Definition~\ref{def:hybridautomata_virtual}. It is shown as the rectangle in the negative side of the $x[1]$-axis in Figure~\ref{fig:waypoint_modes_virtual}.
	
	The illustration above for $p_2$ would be repeated for every $p \in P$, to construct the virtual system shown in Figure~\ref{fig:waypoint_modes_virtual}.
	
	% Looking at our choice of virtual map from a geometric perspective, one would see that we are changing the absolute perspective in Figure~\ref{fig:problem_description} to the relative one in Figure~\ref{fig:waypoint_modes_virtual}. In that sense, it is as if the robot is always going toward the origin, the single virtual mode. Whenever it reaches the origin, it is as if it reached the waypoint in the concrete automaton, it jumps to the position relative to the next waypoint. This jump is handled by the reset of the virtual automaton.  
	
	In summary, the resulting virtual automaton would be: $W_{v} = \langle \stateset_{v}, P_{v}, \initsetv, \initmodev, $ 
	%\timebound_{v}, 
	$\edgeset_{v}, \guard_{v}, \reset_{v}, f_{v} \rangle$, where
	
		\begin{enumerate}[label=(\alph*)]
		\item $\stateset_v = \stateset = \mathbb{R}^3$ and 
		$P_v = \{ p_v = [0, 0] \}$,
		%as shown in Figure~\ref{fig:waypoint_modes_virtual} and \ref{fig:waypoint_modes_virtual_statemachine}, 
		\label{item:def_eg_virtual_mode_set}
		\item $\initsetv$ is $\gamma_{\initmode}(\initset)$ which is the translation of the center of $\initset$ from $[-4.5,0.5,-\frac{\pi}{4}]$ to $[-2,1, -\frac{\pi}{4}]$, $\initmodev$ is the only mode $p_v$, which is the origin, 
		% \item $\timebound_v(p_v) = \max\{5, 10\} = 10$, since all the modes are mapped to the same mode $p_v$, the origin, we take the maximum of time bounds of all modes, \label{rg_item:def_eg_virtual_timebound}
		\item $\edgeset_v = \{e_v = [p_v,p_v]\} $, since all modes are mapped to $p_v$, all the edges map to the same virtual edge $e_v$, \label{item:def_eg_virtual_edgeset}
		%
		%$\edgeset_v \subseteq P_v \times P_v$, where $(\rho_{p_2}(p_1), \rho_{p_3}(p_2)) \in \edgeset_v$ if $(p_1, p_2)$ and $(p_2,p_3) \in \edgeset$, 
		%\item $\guard_v(e_v) = \cup_{p_1,p_2,p_3 \in P}\ \bigg\{ \gamma_{p_2}\big(\resrelguard(p_1,p_2,p_3)\big)\big|
		%\rho_{p_2}(p_1) = p_{v,1}, \rho_{p_3}(p_2) = p_{v,2}\bigg\}$,
		%
		\item 
		\begin{align}
		\guard_v(e_v) = \ball(p_v, \epsilon_0) \times \reals.
		\end{align}
		It is the union of all guards of all the edges mapped to rectangles centered at the origin. The guards of $\edgeinstance_1, \edgeinstance_2,$ and $\edgeinstance_3$ would be mapped to $\ball(p_v, \epsilon_1) \times \reals$, while that of $\edgeinstance_0$ would be mapped to $\ball(p_v, \epsilon_0) \times \reals$, 
		% which are rectangles centered at the waypoints, which would be
		%  
		 \label{item:def_eg_virtual_guard}
		%\item $\reset_v(\stateinstance_v, e_v) = \cup_{p_1,p_2,p_3 \in P}\ \bigg\{\gamma_{p_3}\big(\reset\big(\gamma_{p_2}^{-1}(x_v), e_2\big)\big)\big| \rho_{p_2}(p_1) = p_{v,1}, \rho_{p_3}(p_2) = p_{v,2}\bigg\}$, and
		% \item $D_v = \{ (\gamma_{p_2}(s_2), \rho_{p_2}(p_1), \gamma_{p_3}(s_4), \rho_{p_3}(p_2))\ |\ (s_1,p_1,s_2,p_2), (s_3,p_2,s_4,p_3) \in D, \exists\ t\leq T, s_3 = \xi(s_2, p_2, t)\}$, and
		\item $\forall \stateinstance_v \in \stateset_v$, \begin{align*}
		\reset_v(\stateinstance_v, e_v) = \{\gamma_{p_{1}}(\gamma_{p_0}^{-1}(x_v)), \gamma_{p_{2}}(\gamma_{p_1}^{-1}(x_v)),\\
		\gamma_{p_{3}}(\gamma_{p_2}^{-1}(x_v)),
		\gamma_{p_{0}}(\gamma_{p_3}^{-1}(x_v)) \}, \text{ and}
		\end{align*} 
		% \end{align*} 
	\label{item:def_eg_virtual_reset}
		\item %$f_v: \stateset \rightarrow \stateset$, where $\forall\ \stateinstance \in \stateset,\ 
		$\forall \stateinstance_v \in \stateset_v, f_v(\stateinstance, p_v) = f(\stateinstance, p_v)$, it is the dynamics of equation~(\ref{eq:robot_dynamics}) in Example~(\ref{sec:single_linear_example}) going to the origin. \label{item:def_eg_virtual_dynamics}
	\end{enumerate} 
The reset of the guard is the set of all possible reseted states. It is shown as the rectangles along the axes in Figure~\ref{fig:waypoint_modes_virtual}.

\begin{figure}[t!]
	%\caption{Reachtubes for drone~\ref{fig:linear_nosym_1_agents_sym} and linear~\ref{fig:linear_nosym_3_agents_sym} models  using {\sf Sym-Flow*}. Three agents \label{fig:flow_linear}}
	\centering
	\begin{subfigure}[t]{0.5\textwidth}
		\centering
		\includegraphics[width=\textwidth]{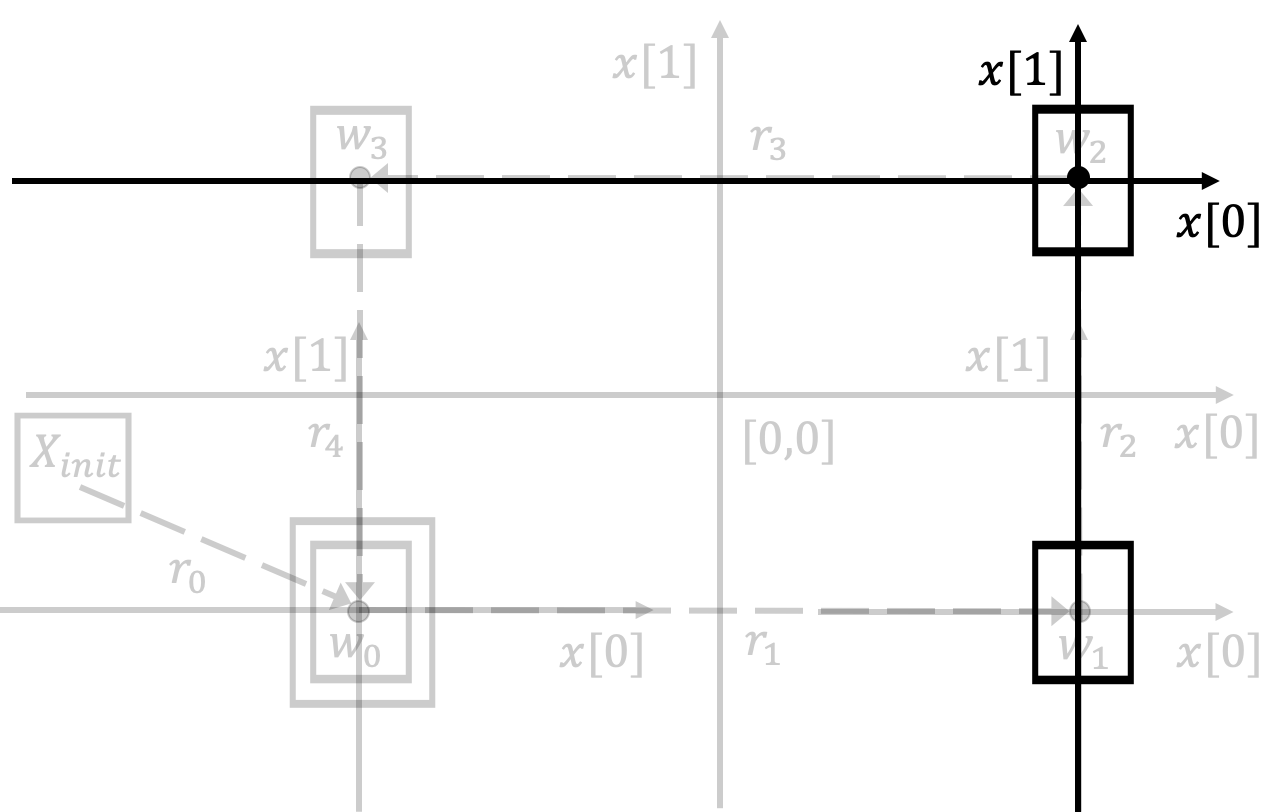} 	
		\caption{ \label{fig:waypoint_modes_virtual_transformation}}
		\vspace{\floatsep}
	\end{subfigure}
	\begin{subfigure}[t]{0.5\textwidth}
		\centering
		\includegraphics[width=\textwidth]{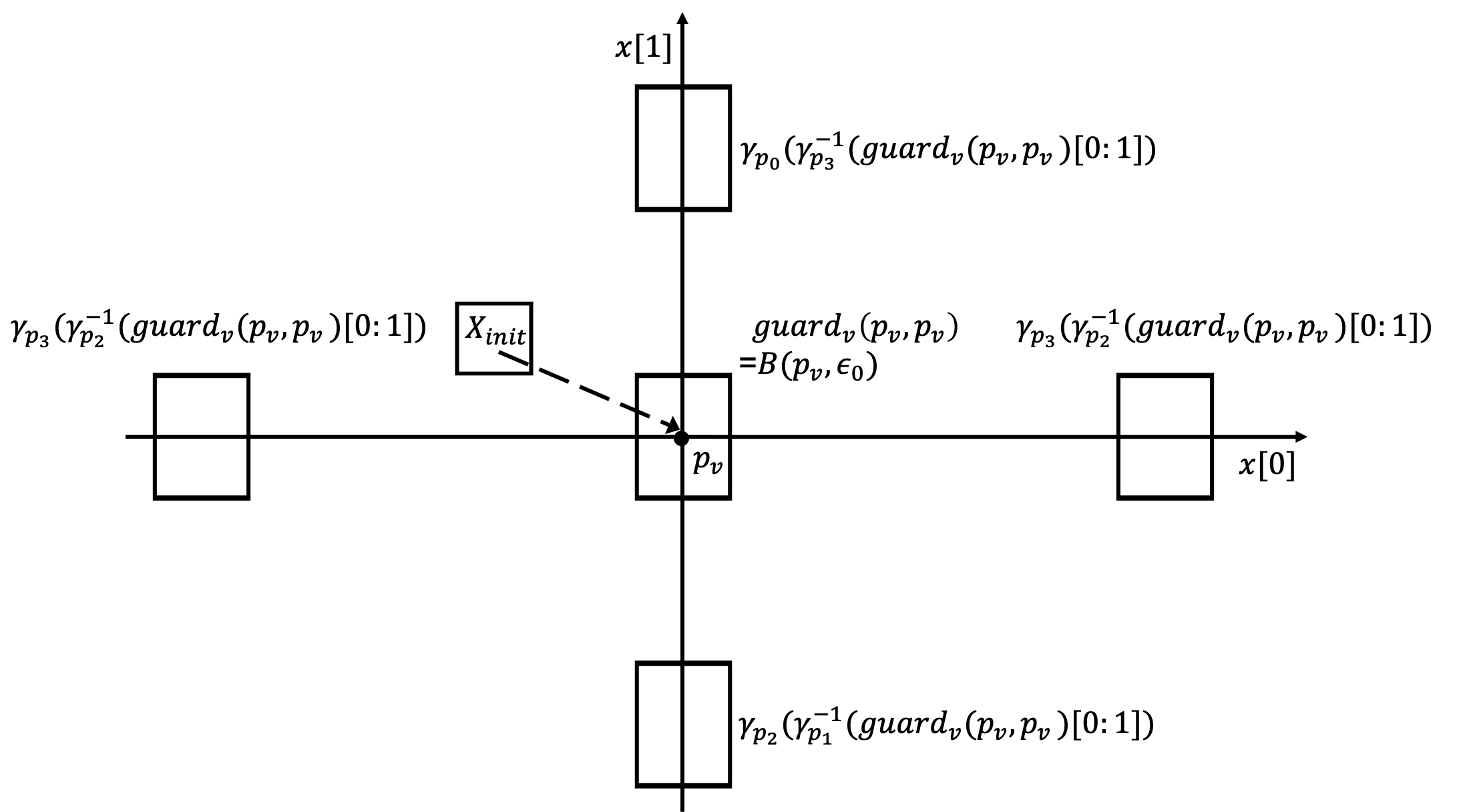}
		\caption{\label{fig:waypoint_modes_virtual}}
		\vspace{\floatsep}
	\end{subfigure}
	\begin{subfigure}[t]{0.5\textwidth}
		\centering
		\includegraphics[width=\textwidth]{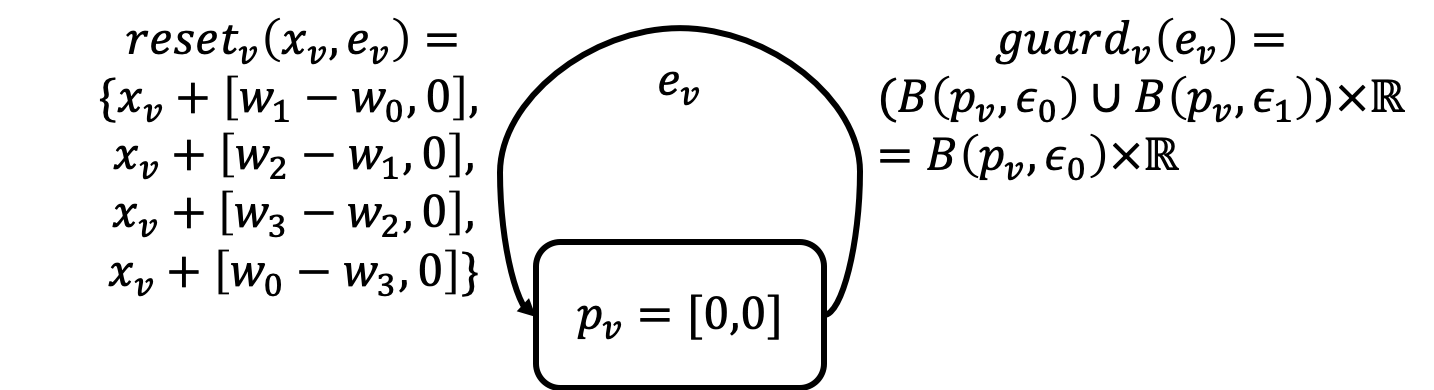}
		\caption{\label{fig:waypoint_modes_virtual_statemachine}}
		\vspace{\floatsep}
	\end{subfigure}
	\caption{\scriptsize (\ref{fig:waypoint_modes_virtual_transformation}) the symmetry transformation that changes the origin of the plane of Figure~\ref{fig:problem_description} to the waypoint $w_2$. (\ref{fig:waypoint_modes_virtual}) shows the virtual initial sets, guards and reseted guards of the virtual automaton of Figure~\ref{fig:problem_description} after choosing the virtual map to be the set of origin translations to the waypoints. 
	(\ref{fig:waypoint_modes_virtual_statemachine}) shows the resulting hybrid automaton $\ha_v$. \label{fig:waypoint_modes_virtual_automaton}}
\end{figure}

\end{example}

\subsection{Forward simulation relation (FSR): from concrete to virtual}
\label{sec:fsr}
In this section, we establish a correspondence from the executions of the concrete system to those of the virtual one through a  FSR~\cite{TIOAmon,Mitra07PhD,GJP:IFAC2006}. A  FSR is a standard approach to describe the similarity of behavior of two different hybrid automata. 

\begin{definition}[FSR \cite{TIOAmon}]
	\label{def:fsr}
	% A {\em forward simulation relation} from hybrid automaton $\ha_1$ to another one $\ha_2$, with initial sets of states and modes $\Theta_1 \subseteq (\stateset_1 \times P_1)$ and $\Theta_2 \subseteq (\stateset_2 \times P_2)$,   is a relation $\mathcal{R} \subseteq (\stateset_1 \times P_1) \times (\stateset_2 \times P_2)$, such that
	A {\em forward simulation relation} from hybrid automaton $\ha_1$ to another one $\ha_2$,
	% with initial sets of states and modes $K_{0,1} \subseteq \stateset_1$ and $p_{0,1} \in P_1$ and $K_{0,2} \subseteq \stateset_2$ and $p_{0,2}\in P_2$,
	 is a relation $\mathcal{R} \subseteq (\stateset_1 \times P_1) \times (\stateset_2 \times P_2)$, such that
	\begin{enumerate}[label={(\alph*)}]
		% \item for any initial pair $ (\stateinstance_1, p_1) \in \Theta_1$, there exists a pair $ (\stateinstance_2, p_2) \in \Theta_2$, where $(\stateinstance_1, p_1, \stateinstance_2, p_2) \in \mathcal{R}$,
		\item for any initial  $\stateinstance_{0,1} \in \initsetf$, there exists a state $\stateinstance_{0,2} \in \initsets$, such that $(\stateinstance_{0,1}, \initmodef) \mathcal{R}(\stateinstance_{0,2}, \initmodes)$, \label{def:fsr_first_cond}
		\item For any discrete transition $(\stateinstance_1, p_1) \rightarrow (\stateinstance_1',p_1')$ of $\ha_1$ and $(\stateinstance_2,p_2) \in \stateset_2 \times P_2$, 
			where $(\stateinstance_1, p_1) \mathcal{R}(\stateinstance_2,p_2)$, there exists $(\stateinstance_2',p_2') \in \stateset_2 \times P_2$ such that $(x_2,p_2)\rightarrow (x_2',p_2')$ is a discrete transition of $\ha_2$ and 
			$(\stateinstance_1',p_1') \mathcal{R}(\stateinstance_2',p_2')$, and 
			\label{def:fsr_second_cond}
		\item For any solution $\xi_1(\stateinstance_1,p_1,\cdot)$ of $\ha_1$ and pair $(\stateinstance_2,p_2) \in \stateset_2 \times P_2$, such that $(\stateinstance_1,p_1)\mathcal{R}(\stateinstance_2,p_2)$, there exists a solution $ \xi_2(\stateinstance_2,p_2,\cdot)$, where $\dur(\xi_1) = \dur(\xi_2)$ and %$\forall\ t \in \dur(\xi_1)$, 
		$(\xi_1.\lstate,p_1) \mathcal{R} (\xi_2.\lstate,p_2).$
		\label{def:fsr_third_cond}
	\end{enumerate}
\end{definition}

Existence of a FSR implies that for any execution of $\A_1$ there is a   corresponding related execution of $\A_2$. The following theorem is an adoption of Corollary 4.23 of \cite{TIOAmon} into our hybrid modeling framework.
\
\begin{theorem}[executions correspondence \cite{TIOAmon}]
	\label{thm:correspondence}
	If there exists a forward simulation relation $\mathcal{R}$ from $\ha_1$ to $\ha_2$, then for every execution $\sigma_1$ of $\ha_1$, there exists a corresponding execution $\sigma_2$ of $\ha_2$ such that 
	\begin{enumerate}[label=(\alph*)]
	\item $\sigma_1.\mathit{len} = \sigma_2.\mathit{len}$,
	\item $\forall\ i \in [\sigma_1.\mathit{len}]$, $\dur(\xi_{1,i}) = \dur(\xi_{2,i})$, and 
	\item $\forall\ i \in [\sigma_1.\mathit{len}]$,
	 %$\forall\ t \in [0,\dur(\xi_1)]$,
	  $(\xi_{1,i}.\lstate,p_{1,i}) \mathcal{R} (\xi_{2,i}.\lstate, p_{2,i})$. 
	\end{enumerate}
\end{theorem}

% Before introducing the FSR from the concrete system to the virtual one, we define their initial sets $\Theta$ and $\Theta_v$, respectively. The initial continuous state of the concrete system can be any state $\stateinstance \in \stateset$. However, none of the modes in $P_{-1}$ that are assumed to have generated the initial states in Remark~\ref{rm:K_initial_parameter} can be an initial mode. Moreover, an initial mode should belong to the graph, i.e. it should belong to an edge in $\edgeset$. Hence, the initial set of states and parameters is \begin{align}
%\Theta = \stateset \times \{p \in P\textbackslash P_{-1}\ |\ \exists\ p' \in P, \text{ s.t. }(p',p)\in \edgeset\}.
%\end{align}
%On the other hand,
Now we introduce a FSR from the concrete hybrid automaton to the virtual one.

\begin{theorem}[FSR: concrete to virtual]
	\label{thm:fsr_concrete_virtual}
%iven a hybrid automaton $\ha$, an initial set of states $K_0 \subseteq \stateset$, an initial mode $p_0 \in P$, a virtual map $\Gamma$, and the resulting virtual hybrid automaton $\ha_v$ from Definition~\ref{def:hybridautomata_virtual} with its initial set $K_{v,0}$ and mode $p_{v,0}$,
Consider the relation $\fsr \subseteq (\stateset \times P)\times (\stateset_v \times P_v)$ defined as $(\stateinstance, p) \fsr (\stateinstance_v,p_v)$ if and only if: 
\begin{enumerate}[label=(\alph*)]
\item $x_v = \gamma_p(x)$, and
\item $p_v = \rv(p)$.
\end{enumerate}
Then, $\fsr$ is a forward simulation relation from  $\ha$ to  $\ha_v$.
\end{theorem}
\begin{proof}
%The first condition in Definition~\ref{def:fsr} is satisfied since: given any $(\stateinstance, p) \in \stateset \times P$, by assumption, $\gamma_{p}$ exists, hence $\gamma_{p}(\stateinstance)$ exists as well, and there exists $p' \in P$ such that $(p',p) \in \edgeset$ and $\rho_{p'}(p) = p_v$, by part $(a)$ of Definition~\ref{def:hybridautomata_virtual}.
$\fsr$ satisfies  Definition~\hyperref[def:fsr_first_cond]{\ref*{def:fsr}.\ref*{def:fsr_first_cond}} since:
% by definition, $\gamma_{p_0}$ exists in $\Gamma$, hence 
for any $\stateinstance_0 \in \initset$, $\gamma_{\initmode}(\stateinstance_0) \in \initsetv$, and $\initmodev= \rv(\initmode)$, by Definition~\hyperref[item:def_virtual_initial]{\ref*{def:hybridautomata_virtual}.\ref*{item:def_virtual_initial}}. 
%Also, we have $(x_0, p_0, \gamma_{p_0}(x_0),\bot) \in \fsr$. Hence, for every $x_0 \in K_0$, there exists $x_{v,0} \in K_{v,0}$ such that $(x_0,p_0,x_{v,0}, \bot) \in \fsr$. %Hence, $(x_0,p_0,\gamma_{p}(x_0), p_{v,0}) \in \fsr$ and .

To prove that $\fsr$ satisfies Definition~\hyperref[def:fsr_second_cond]{\ref*{def:fsr}.\ref*{def:fsr_second_cond}}, fix a discrete transition $(\stateinstance, p) \rightarrow (\stateinstance',p')$ of $\ha$ and $(\stateinstance_v, p_v) \in \stateset_v \times P_v$ such that $(\stateinstance, p, \stateinstance_v, p_v) \in \fsr$. We will show that if we choose $\stateinstance_v' = \gamma_{p'}(\stateinstance_v)$ and $p_v' = \rv(p')$, then $(x_v',p_v') \in \stateset_v \times P_v$, $(x',p',x_v',p_v') \in \fsr$, and $(x_v,p_v) \rightarrow (x_v',p_v')$ is a valid discrete transition of $\ha_v$. 

First, $\stateinstance_v' \in \stateset_v$ since $\stateinstance' \in \stateset$, $\gamma_{p'}$ is a map from $\stateset$ to $\stateset$, and by Definition~\hyperref[item:def_virtual_mode_set]{\ref*{def:hybridautomata_virtual}.\ref*{item:def_virtual_mode_set}}, $\stateset =\stateset_v$. Moreover, $p_v' \in P_v$ by Definition~\hyperref[item:def_virtual_mode_set]{\ref*{def:hybridautomata_virtual}.\ref*{item:def_virtual_mode_set}}.
Second, $(x',p',x_v',p_v') \in \fsr$ since $\stateinstance_v' = \gamma_{p'}(\stateinstance')$ and $p_v' = \rv(p')$.
Third, fix $\edgeinstance = (p,p')$. Then, by the definition of discrete transitions of $\ha$, $\stateinstance \in \guard(\edgeinstance)$ and $\stateinstance' \in \reset(x, e)$. 
%Moreover, from the definition of the mode set $P_v$ of the virtual system in part \ref{item:def_virtual_mode_set} of Definition~\ref{def:hybridautomata_virtual}. 
%
% In addition, $p_v = \rho_{p}(p)$, or $p = p_0$ and $p_v = \bot$.
%
Also, from the definition of $\edgeset_v$ in Definition~\hyperref[item:def_virtual_edgeset]{\ref*{def:hybridautomata_virtual}.\ref*{item:def_virtual_edgeset}},
the edge $\edgeinstance_v = (p_v,p_v') \in \edgeset_v$.
Also, by the definition of $\fsr$ and the assumption that $\stateinstance$ and $\stateinstance_v$ are related under $\fsr$, $\stateinstance_v = \gamma_{p}(\stateinstance)$. That means that $\stateinstance_v \in  \gamma_p(\guard(e))$, since $\stateinstance \in \guard(e)$. But, by Definition 
\hyperref[item:def_virtual_guard]{\ref*{def:hybridautomata_virtual}.\ref*{item:def_virtual_guard}},
%of the virtual system guards in part~\ref{item:def_virtual_guard} of Definition~\ref{def:hybridautomata_virtual},
 $\gamma_p(\guard(e)) \subseteq \guard_v(\edgeinstance_v)$. Then, $\stateinstance_v \in \guard_v(e_v)$.
Moreover, since $\stateinstance' \in  \reset(\stateinstance, \edgeinstance))$ and $\stateinstance = \gamma_{p}^{-1}(\stateinstance_v)$, then $\stateinstance' \in  \reset(\gamma_{p}^{-1}(\stateinstance_v), \edgeinstance)$. Hence, $\stateinstance_v' = \gamma_{p'}(x') \in  \gamma_{p'}(\reset(\gamma_{p}^{-1}(\stateinstance_v), \edgeinstance))$. Using 
Definition~\hyperref[item:def_virtual_reset]{\ref*{def:hybridautomata_virtual}.\ref*{item:def_virtual_reset}},
%the reset definition of the virtual system in part~\ref{item:def_virtual_reset} in Definition~\ref{def:hybridautomata_virtual},
  we know that $ \gamma_{p'}(\reset(\gamma_{p}^{-1}(\stateinstance_v), \edgeinstance)) \subseteq \reset_v(\stateinstance_v, \edgeinstance_v)$. We have $\stateinstance_v' \in \reset_v(\stateinstance_v, \edgeinstance_v)$. Therefore, $(\stateinstance_v, p_v) \rightarrow (\stateinstance_v',p_v')$ is a valid discrete transition of $\ha_v$.

To prove that $\fsr$ satisfies Definition~\hyperref[def:fsr_third_cond]{\ref*{def:fsr}.\ref*{def:fsr_third_cond}},
 fix a solution $\xi(\stateinstance,p,\cdot)$ of $\ha$ and a pair $(\stateinstance_v,p_v) \in \stateset_v \times P_v$, such that $(\stateinstance,p,\stateinstance_v,p_v)\in \fsr$. Then, we will show that $\dur(\xi) = \dur(\xi_v)$ and $(\xi(\stateinstance,p,\dur(\xi)),p,$ $\xi_v(\stateinstance_v,p_v,\dur(\xi)), p_v) \in \fsr$.
Since $\stateinstance$ and $\stateinstance_v$ are related under $\fsr$, then $\stateinstance_v = \gamma_{p}(\stateinstance)$. Moreover, using Theorem~\ref{thm:sol_transform_input_nonlinear}, $\forall\ t \in \dur(\xi)$, $\xi(\gamma_{p}(x), \rho_{p}(p), t) = \gamma_p(\xi(\stateinstance,p,t))$. But, $\rv(p) = p_v$ and using Definition
\hyperref[item:def_virtual_dynamics]{\ref*{def:hybridautomata_virtual}.\ref*{item:def_virtual_dynamics}},
 $\xi(\gamma_{p}(x), \rho_{p}(p), \cdot) =$ $\xi_{v}(\gamma_{p}(x), p_v,\cdot)$, which is a solution of $\ha_v$ starting from $\gamma_p(x)= x_v$. In addition, from the assumption that the guards are optional, we can choose $\xi_v$ that does not transition before $\dur(\xi)$. 
% part~\ref{item:def_virtual_timebound} of Definition~\ref{def:hybridautomata_virtual}, $\timebound_v(p_v) \geq \timebound(p)$.
Therefore, $\forall\ t \in \dur(\xi)$, $(\xi(\stateinstance,p,t),p,$ $\xi_v(\stateinstance_v,p_v,t), p_v) \in \fsr$. 
%Hence, the theorem.
\end{proof}

\begin{definition}
Given a hybrid automaton $\ha$ and a virtual map $\Phi$, we denote the corresponding virtual automaton by $\ha_{v,\phi}$ and the resulting FSR of Theorem~\ref{thm:fsr_concrete_virtual}, by $\mathcal{R}_\phi$.
\end{definition}

%\subsection{Implications of the forward simulation relation}
%\label{sec:fsr_implications}

The following corollary is also an adoption of Theorem 4.2 of \cite{TIOAmon} into our hybrid automaton framework.

\begin{corollary}[Theorem 4.2 in \cite{TIOAmon}]
	\label{cor:composition}
% Fix a hybrid automaton $\ha$ and a corresponding virtual map $\Phi$. Let  $\ha_v$ and $\fsrf$ be the resulting virtual hybrid automaton and FSR, respectively. Now, consider another virtual map $\Phi_v$ of $\ha_v$ and let $\ha_{\mathit{vv}}$ and $\fsrs$ be the resulting virtual hybrid automaton of $\ha_v$ and FSR, respectively. Then, $\ha_{\mathit{vv}}$ is a virtual system of $\ha$ with $\fsrf \circ \fsrs$ and $\Phi \circ \Phi_v$  being the corresponding FSR and virtual map, respectively, where $\circ$ is the composition operator. 
Let $A, B$ and $C$ be three hybrid automata and $\Phi_{AB}$ and $\Phi_{BC}$ be two virtual maps such that $B = A_{v,\Phi_{AB}}$ and $C = B_{v,\Phi_{BC}}$ with corresponding FSRs $\mathcal{R}_{AB}$ and $\mathcal{R}_{BC}$. Then, $C = A_{v,\Phi_{AC}}$ is the virtual automaton of $A$ with FSR $\mathcal{R}_{AB} \circ \mathcal{R}_{BC}$ and virtual map $\Phi_{AC} = \Phi_{AB} \circ \Phi_{BC}$,  where $\circ$ is the composition operator. 
\end{corollary}

Corollary~\ref{cor:composition} shows that we can apply symmetries in sequence to get hierarchical levels of  abstractions of $\ha$.% the automaton under consideration.

%It follows from the theorem above and Theorem~\ref{thm:correspondence} that for every execution of the concrete model, there exists a corresponding execution of the virtual one. 
%
%However, 
It is worth noting that there may not be a forward simulation relation from $\A_v$  to $\A$. The guard and reset of an edge $\edgeinstance_v$ of $\A_v$  are the union of all the transformed versions of the guards and resets of the edges of $\A$ that get mapped to $e_v$. Hence, some discrete transitions in $\ha_v$ may not have corresponding ones in $\ha$. For example, consider two edges $\edgeinstance_1 = (p_{11}, p_{12})$ and $\edgeinstance_2 = (p_{21}, p_{22})$ of $\ha$ with $\rv(\edgeinstance_1) = \rv(\edgeinstance_2) = \edgeinstance_v = (p_{v1}, p_{v2})$, an edge of $\ha_v$.  Then, a transition over $\edgeinstance_v$ would be allowed in $\ha_v$ with reseted state being $\gamma_{p_{22}}(\reset(\stateinstance_v, \edgeinstance_2))$ if $\stateinstance_v \in \gamma_{p_{11}}(\guard(\edgeinstance_1))$. Such a transition may not have a correspondent one in $\ha$, since it resembles a transition from $p_{11}$ to $p_{22}$.
% That means that some transitions in $\ha_v$ may be due to the satisfaction of the guard of another edge.
% Moreover, the time bound on a virtual mode $p_v$ is the maximum of all time bounds of the modes that get mapped to $p_v$. 
% That would allow more discrete transitions in $\ha_v$ than in $\ha$. 
Thus, some executions of $\A_v$ may not have corresponding executions in $\A$.

\section{Different virtual maps lead to different abstraction}
\label{sec:virtual_automaton_edges}

In this section, we show that the same scenario can result in different abstractions when different symmetries are applied. This multitude of modeling approaches would serve different purposes for the abstraction user. 
%We show this Hence, depending on the abstraction, one might need to need t.   
% We consider the metrics for comparison to be the numbers of modes and edges of the virtual hybrid automaton \hussein{and the volume of the corresponding guards and resets (need to say something more precise)}.
% Those metrics have direct impact on the time and space complexity and quality of the computations done over the automaton.
%
% We show that by modifying the hybrid automaton modeling of the scenario of Figure~\ref{fig:problem_description} and using a new virtual map. The new automaton has the roads in Figure~\ref{fig:problem_description} as its modes and the new virtual map has rotation in addition to translation. 
%For clarity, we add a subscript of $\wpha$ when we refer to the hybrid automaton of Example~\ref{sec:single_linear_example} and subscript $\roha$ for the hybrid automaton in this section.

We follow the same sequence of presentation as that of  Examples~\ref{sec:single_linear_example}, \ref{sec:transformation_example}, and \ref{sec:virtual_example}: first, we show a new hybrid automaton modeling of scenario in Figure~\ref{fig:problem_description} in Example~\ref{sec:single_robot_example_edges}, we then show a corresponding symmetry map in Example~\ref{sec:transformation_example_edges}, and finally, construct a virtual map and the corresponding virtual automaton in Example~\ref{sec:virtual_example_roads}.
% the resulting virtual automaton from Definition~\ref{def:hybridautomata_virtual} has more modes and edges, but more compact possible reset sets. 

% We start by presenting the new automaton in the following example.
\begin{example}[Modeling the scenario in Figure~\ref{fig:problem_description} with roads as modes]
	\label{sec:single_robot_example_edges}
	Consider the same scenario described in Figure~\ref{fig:problem_description}. In this example, instead of defining the modes of the hybrid automaton to be the waypoints, suppose we define the modes to be the roads. 
	In each mode, the robot would follow the destination waypoint of the corresponding road. We annotate the components of automaton $W$ of Example~\ref{sec:single_linear_example} with a subscript $W$ and define
	the resulting automaton of this example as follows: $R =  \langle \stateset, P, \initset, \initmode,  \edgeset, \guard, \reset, f \rangle$
	%$\ha_\roha$ 
	is shown in Figure~\ref{fig:edges_modes_original_statemachine}, where
	% and would be  $\ha_\roha = \langle \stateset_\roha, P_\roha, \initsetro, \initmodero, \edgeset_\roha, \guard_\roha, \reset_\roha, f_\roha \rangle$:
	\begin{enumerate}[label=(\alph*)]
		\item $\stateset = \stateset_\wpha \subseteq \reals^3$, same as that of Example~\ref{sec:single_linear_example}, and $P = \{p_{i} = r_{i}\ |\ i\in [5] \}$, the set of roads in Figure~\ref{fig:problem_description}, 
		\item $\initset = \initsetwp$ and $\initmode = p_{ 0}$,
		% and $\pd = \pout = \pin \textbackslash \{(-5,-5)\} \cup \{(-5,5)\}$.  
		% Continuous states are simply called {\em states} and the discrete ones are called {\em modes} or {\em parameters}.
		% \item $\timebound(p_{i})$ equals 5 seconds for $i \in \{0,2\}$ and equals 10 seconds for $i \in \{2,3\}$, 
		\item $\edgeset = \{e_{0} = (p_{0}, p_{1}), e_{1} = (p_{1}, p_{2}), e_2 = (p_{2}, p_{3}), e_3 = (p_{3}, p_{4}), e_4 = (p_{4}, p_{1}) \}$,  
		%where a transition $(p_1, p_2) \in \edgeset$ is written as $p_1 \rightarrow p_2$,  \label{ite:def_real_edge_set}
		\item  
		\begin{align*}\guard(e_i) =
		\begin{cases} 
		\ball(w_{i}, \epsilon_0) \times \reals, \text{ if } i=0, \\
		\ball(w_{i}, \epsilon_1) \times \reals, \text{ if } i=\{1,2,3,4\}. 
		\end{cases} 
		\end{align*} 
		\label{eg_item:def_2_guard}
		\item $\forall \stateinstance \in \stateset, \edgeinstance \in \edgeset$,  \[\reset(\stateinstance,\edgeinstance) = \{\stateinstance\},\] is the identity map, \label{eg_item:def_2_reset} and
		\item $\forall \stateinstance \in \stateset, \forall p \in P$, 
		\begin{align}
		f(x,p) = f_\wpha(x,p.\dest).  \label{eg_item:def_2_robot_dynamics}
		\end{align}
		%, where $f_\wpha$ is the continuous dynamics in equation~(\ref{eq:robot_dynamics}) of Example~\ref{sec:single_linear_example}.
		%&\frac{dx[0]}{dt} = v \cos \theta, \frac{dx[1]}{dt} = v \sin \theta , \frac{dx[2]}{dt} = 2v \sin (\alpha) / L, \\
		%f(x,p) =  A (\stateinstance - p)$, where $A = [[-3, 1], [1, -2]]$ is a stable matrix.
	\end{enumerate} 
	\begin{figure}[t!]
		%\caption{Reachtubes for drone~\ref{fig:linear_nosym_1_agents_sym} and linear~\ref{fig:linear_nosym_3_agents_sym} models  using {\sf Sym-Flow*}. Three agents \label{fig:flow_linear}}
		\centering
		\includegraphics[width=0.5\textwidth]{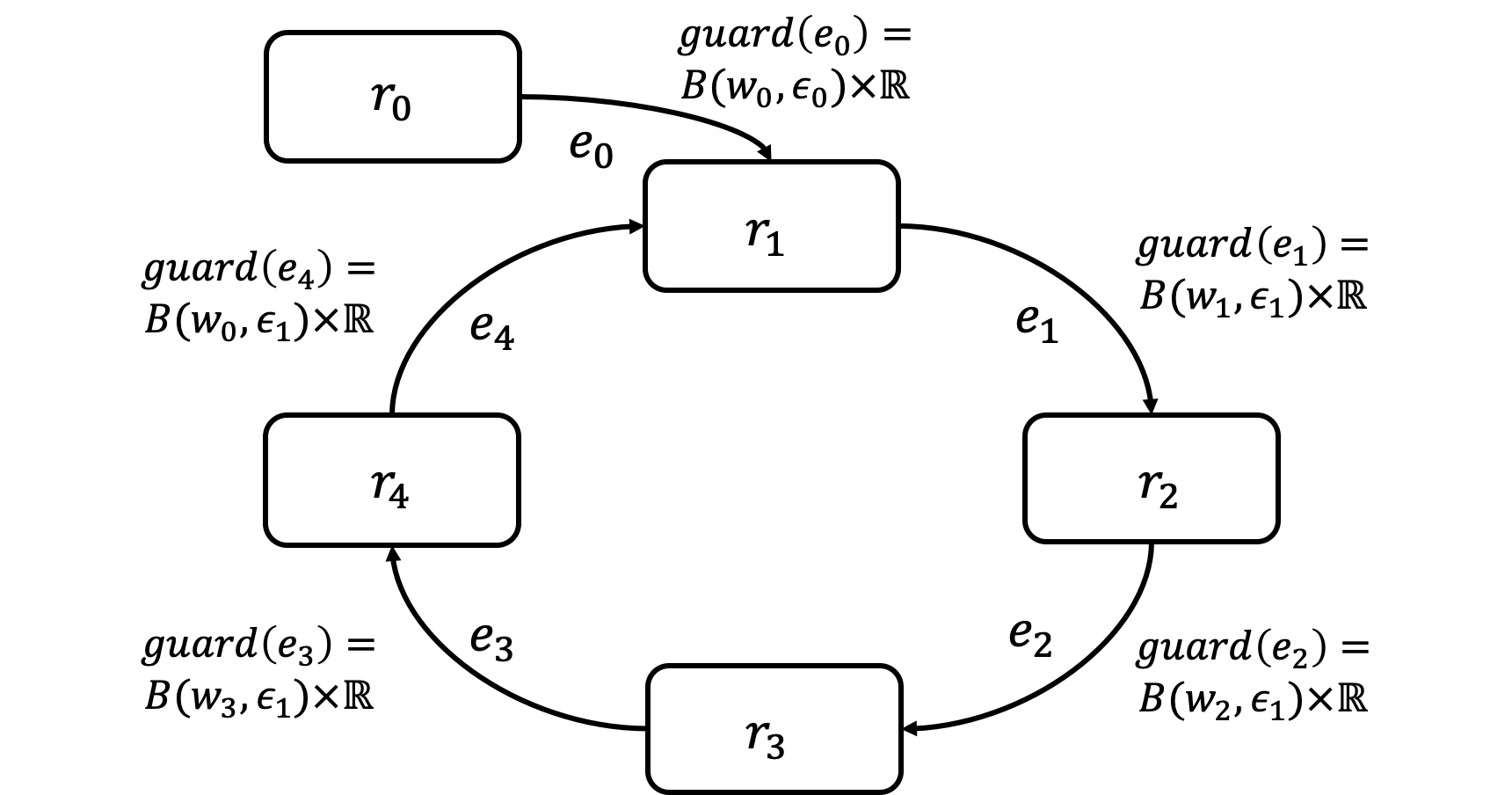} 
		\caption{\scriptsize The state machine representing the discrete transitions of the hybrid automaton $R$ describing the scenario in Figure~\ref{fig:problem_description} with the roads being the modes. The resets are omitted since they are just the identity map for all the modes. \label{fig:edges_modes_original_statemachine}}
	\end{figure}
	
\end{example}

\begin{example}[Robot coordinate transformation symmetry]
	%, borrowed and modified from \cite{Sibai:TACAS2020}]
	\label{sec:transformation_example_edges}
	%We modify the maps that were introduced in Example~\ref{sec:transformation_example} to fit the include rotation of axes in addition to translation of origin.
	We consider the new concrete model introduced in Example~\ref{sec:single_robot_example_edges} of the scenario described in Figure~\ref{fig:problem_description}. Fix a vector $p^* \in \reals^4$, where the first two coordinates $\mathit{p^*.\src}$ define the start point and the last two coordinates $p^*.\dest$ specifying the end point in the plane. Such a $p^*$ is similar to the roads in Figure~\ref{fig:problem_description}. We define $\gamma_\coortrans: \stateset \rightarrow \stateset$ and $\rho_\coortrans: P \rightarrow P$ to be the maps that transform the coordinate system of the plane where the robot and roads reside. These maps transform it so that  $p^*$ will be collinear with the $x[0]$-axis and $p^*.\dest$ be the origin of the system. 
	Formally, for every $\stateinstance \in \stateset$ and $p\in P$, 
	% Fix $\mathit{src} \in \mathbb{R}^2$.
	%$\gamma_\roha : \stateset_\roha \rightarrow \stateset_\roha$ and % $\rho_\roha : P_\roha \rightarrow P_\roha$ to be:
	\begin{align}
	\gamma_\coortrans(x) &= [{\bf R}_\theta (x[0:1] - p^*.\dest), x[2] - \theta], \\ % 
	\rho_\coortrans(p ) &=  [{\bf R}_\theta (p.\src - p^*.\dest), {\bf R}_\theta (p.\dest - p^*.\dest)], % {\bf R}
	\end{align}
	%\begin{comment}
	where $\theta = \arctan_2(p^*.\dest[1] - p^*.\src[1], p^*.\dest[0] -  p^*.\src[0])$ and
	\begin{align}
	{\bf R}_\theta =
	\left[
	\begin{matrix}
	\cos(\theta) &\sin(\theta) \\
	-\sin(\theta) &\cos(\theta)
	\end{matrix}
	\right]
	\end{align} 
	is the rotation matrix with angle $\theta$.
	%\end{comment} 
	% \gamma(\stateinstance) &= [(\stateinstance[0] - p[0])\cos(\theta) - (\stateinstance[1] - p[1])\sin(\theta),\nonumber\\
	% &\hspace{0.5in} -(\stateinstance[0] - \mathit{src}[0])\sin(\theta)  + (\stateinstance[1] - \mathit{src}[1])\cos(\theta)] \text{ and }\\
	%\begin{align}
	%\rho(p) &= [(p[0] - \mathit{src}[0])\cos(\theta) + (p[0] - 
	% &\hspace{0.5in} -(p[0] - \mathit{src}[0])\sin(\theta) + (p[1] - \mathit{src}[1])\cos(\theta)],
	% \end{align}
	Then, we can check with simple algebra, that for all $\stateinstance  \in \stateset$ and $p \in P$, $\frac{\partial \gamma_\coortrans}{\partial x} f(\stateinstance,p) = f (\gamma_\coortrans(\stateinstance ), \rho_\coortrans (p ))$. 
	%The transformation $\gamma$ would change the origin of $\stateset_\roha$ from zero to $[p_\roha^*[0], p_\roha^*[1], \theta]$. 
	% Then, it would rotate its axes counter-clockwise by $\theta$, so that the $x$-axis is aligned with the segment connecting $\src$ and $p$ waypoints. 
	% Moreover, $\rho$ would translate the origin of the waypoint space $\mathbb{R}^2$ to $p.\dest$.
	% For the mobile robot, this means translating and rotating the plane so that the segment connecting the aircraft and the waypoint positions reside would be the $y$-axis. 
	% If translation would be considered alone, $[\stateinstance[0], \stateinstance[1]]$ would be left unchanged and  $[\stateinstance[2], \stateinstance[3]]$ as well as $[p[2], p[3]]$, would be translated by $[-\mathit{src}[0],-\mathit{src}[1]]$.
	
	% SIMPLIFICATION
	% We use the same transformations for the linear example. The only difference is that it has no heading angle. Hence, in case of translation, the first two coordinates are translated by $[-\mathit{src}[0],-\mathit{src}[1]]$, and in case of rotation, the axes of these two coordinates get rotated as the last two state components of the aircraft. The third coordinate is left intact.
	
	%In \cite{Sibai:TACAS2020}, the authors set the first two coordinates of $\rho(p)$ to zero, independent of $p$. Here we use them to represent the relative position of the src waypoint with respect to the destination one.
	\begin{figure}[t!]
		%\caption{Reachtubes for drone~\ref{fig:linear_nosym_1_agents_sym} and linear~\ref{fig:linear_nosym_3_agents_sym} models  using {\sf Sym-Flow*}. Three agents \label{fig:flow_linear}}
		\centering
		\includegraphics[width=0.5\textwidth]{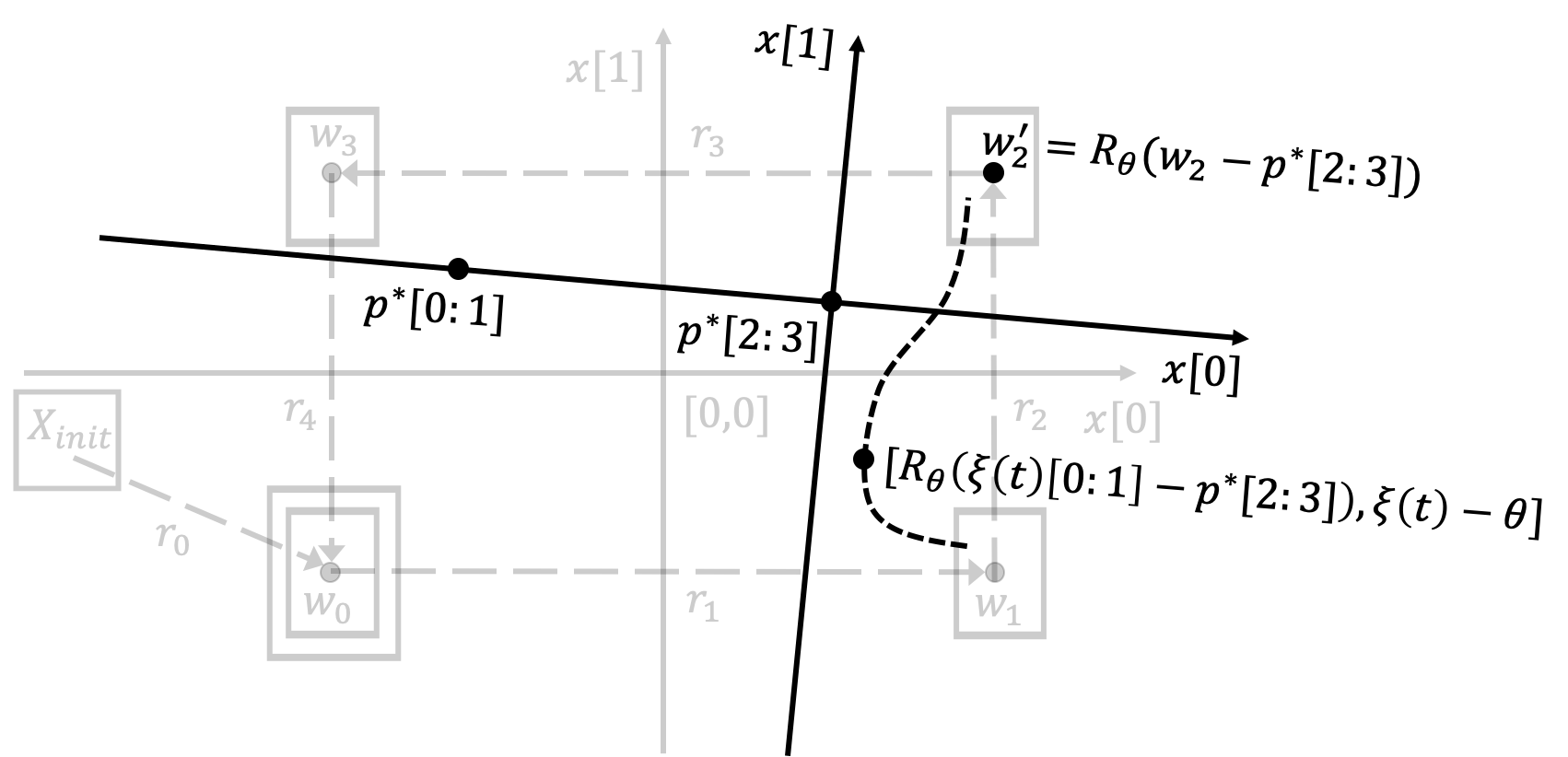} 
		\caption{\scriptsize Changing the axes of the coordinate system so that: the segment connecting $p^*.\src = p^*[0:1]$ to $p^*.\dest = p^{*}[2:3]$ is the new $\stateinstance[0]$-axis and $p^{*}.\dest$ is the new origin. Such a transformation does not affect the intrinsic behavior of the robot, but only transforms the states in its trajectories to conform with the new coordinate system. Such a coordinate transformation is a valid symmetry. \label{fig:edges_modes_transformation}}
	\end{figure}
\end{example}

\begin{example}[Robot virtual system with a coordinate transformation virtual map]
	%, continuous dynamics part borrowed and modified from \cite{Sibai:TACAS2020} while we added set of virtual modes, virtual edges, guards, and resets]
	\label{sec:virtual_example_roads}
	Consider the scenario described in Example~\ref{sec:single_linear_example} and Figure~\ref{fig:problem_description} and its hybrid automaton modeling in Example~\ref{sec:single_robot_example_edges}. To construct a virtual system, we use a virtual map $\Phi_{\roha}$ based on the transformations in Example~\ref{sec:transformation_example_edges}:
	for any $p \in P$, we choose the transformations $\gamma_{p }$ and $\rho_{p}$ to be $\gamma_{\coortrans}$ and $\rho_\coortrans$ of Example~\ref{sec:transformation_example_edges} with $p^*$ being $p$. The resulting virtual automaton would be:
	$R_{v} = \langle \stateset_{ v}, P_{ v}, \initsetv, \initmodev,\edgeset_{v}, \guard_{v}, \reset_{ v}, f_{v} \rangle$, where
	
	\begin{enumerate}[label=(\alph*)]
		\item $\stateset_{v} = \stateset  = \mathbb{R}^3$ and 
		$P_{ v} = \{ p_{ v, 0} = [-\sqrt{5}, 0, 0, 0], p_{v, 1} = [-3, 0,0,0], p_{v, 2} =[-5,0,0,0] \}$, \label{item:def_eg_2_virtual_mode_set}
		\item $\initsetv = \gamma_{\initmode}(\initset)$ and  $\initmodev = \rv (\initmode) = p_{v, 0}$, 
		% \item $\timebound_v(p_v) = \max\{5, 10\} = 10$, since all the modes are mapped to the same mode $p_v$, the origin, we take the maximum of time bounds of all modes, \label{rg_item:def_eg_virtual_timebound}
		\item $\edgeset_{v} = \{\edgeinstance_{{ v, 0}} = [p_{v, 0},p_{ v, 1}], \edgeinstance_{{v, 1}} = [p_{ v, 1},p_{ v, 2}], \edgeinstance_{{ v, 2}} = [p_{ v, 2},p_{v, 1}] \} $, \label{item:def_eg_2_virtual_edgeset}
		%
		%$\edgeset_v \subseteq P_v \times P_v$, where $(\rho_{p_2}(p_1), \rho_{p_3}(p_2)) \in \edgeset_v$ if $(p_1, p_2)$ and $(p_2,p_3) \in \edgeset$, 
		%\item $\guard_v(e_v) = \cup_{p_1,p_2,p_3 \in P}\ \bigg\{ \gamma_{p_2}\big(\resrelguard(p_1,p_2,p_3)\big)\big|
		%\rho_{p_2}(p_1) = p_{v,1}, \rho_{p_3}(p_2) = p_{v,2}\bigg\}$,
		%
		\item 
		\begin{align*}
		\guard_{v}(e_{v,i})= 
		\begin{cases}
		{\bf R}_{\arctan_2(-1,2)} \ball(p_{v,i}.\dest, \epsilon_0) \times \reals, \text{ if } i = 0,\\
		{\bf R}_{0} \ball(p_{v,i}.\dest, \epsilon_1) \times \reals, \text{ if } i = 1,\\ 
		{\bf R}_{\pi/2} \ball(p_{v,i}.\dest,\epsilon_1) \times \reals, \text{ if } i = 2,
		\end{cases} 
		\end{align*} 
		where $p_{v,i} = e_{v,i}.\src, \forall i \in [3]$,
		\label{item:def_eg_2_virtual_guard}
		%\item $\reset_v(\stateinstance_v, e_v) = \cup_{p_1,p_2,p_3 \in P}\ \bigg\{\gamma_{p_3}\big(\reset\big(\gamma_{p_2}^{-1}(x_v), e_2\big)\big)\big| \rho_{p_2}(p_1) = p_{v,1}, \rho_{p_3}(p_2) = p_{v,2}\bigg\}$, and
		% \item $D_v = \{ (\gamma_{p_2}(s_2), \rho_{p_2}(p_1), \gamma_{p_3}(s_4), \rho_{p_3}(p_2))\ |\ (s_1,p_1,s_2,p_2), (s_3,p_2,s_4,p_3) \in D, \exists\ t\leq T, s_3 = \xi(s_2, p_2, t)\}$, and
		\item $\forall \stateinstance_{v} \in \stateset_{v}$, 
		\begin{align*}
		\reset_{v}(\stateinstance_{v}, \edgeinstance_{v,i}) =
		\begin{cases}
		\{\gamma_{ p_{1}}(\gamma_{p_{0}}^{-1}(x_{v})) \}, \text{ if } i = 0, \\
		\{\gamma_{p_{2}}(\gamma_{p_{1}}^{-1}(x_{v})),  \gamma_{p_{4}}(\gamma_{p_{3}}^{-1}(x_{v}))\}, \text{ if } i = 1,\\
		\{\gamma_{p_{3}}(\gamma_{p_{2}}^{-1}(x_{v})),  \gamma_{p_{4}}(\gamma_{p_{1}}^{-1}(x_{v}))\}, \text{ if } i = 2,
		\end{cases}
		\end{align*}
		% \end{align*} 
		\label{item:def_eg_2_virtual_reset}
		\item %$f_v: \stateset \rightarrow \stateset$, where $\forall\ \stateinstance \in \stateset,\ 
		$\forall \stateinstance_v \in \stateset_v, \forall p_{v} \in P_{v},  f_{v}(\stateinstance_{v}, p_v) =  f(\stateinstance_{v}, p_{v}.\dest)$. \label{item:def_eg_2_virtual_dynamics}
	\end{enumerate} 
The resets of the guards of the three edges in $\edgeset_{v}$ constitute the set of all possible reseted states. They are shown as the rectangles on the $x[0]$-axis, but not at the origin, in Figure~\ref{fig:edge_modes_virtual}.
	
The new virtual automaton $R_{v}$ has three modes and three edges versus the single mode and single edge of  $W_{v}$. However, $R_{v}$ has guards and reseted guards of smaller volume. To see that, check Figure~\ref{fig:waypoint_modes_virtual} and compare it with Figure~\ref{fig:edge_modes_virtual}. In Figure~\ref{fig:waypoint_modes_virtual}, the reset of the guard of the only edge consists of four rectangles, from which the trajectory of the mode $p_v$ can start. That in addition to the initial set $\initset$. On the other hand, in Figure~\ref{fig:edge_modes_virtual}, $R_v$ has three modes and three edges. Yet, the guards and reseted guards are overlapping. This suggests that the reach set of $R_{v}$ has a smaller volume than that of $W_{v}$. That in turn means that the reach set computation would be generally easier and faster for $R_{v}$ than that for $W_{v}$, $W$, and $R$.

% Moreover, the volume of the guards and reseted guards are equal for $\ha_{\wpha}$ and $\ha_{\wpha,v}$, as can be seen by comparing Figures~\ref{fig:problem_description} and \ref{fig:waypoint_modes_virtual}. Thus, the abstraction at least did not increase the volume of the reachset.  

% We will show in the next section how to utilize the fact that computing reachsets of the virtual automaton may be faster than its concrete counterpart, for the sake of verifying the safety of the concrete automaton. 
	
\begin{figure}[t!]
	%\caption{Reachtubes for drone~\ref{fig:linear_nosym_1_agents_sym} and linear~\ref{fig:linear_nosym_3_agents_sym} models  using {\sf Sym-Flow*}. Three agents \label{fig:flow_linear}}
	\centering
	\begin{subfigure}[t]{0.5\textwidth}
		\centering
		\includegraphics[width=\textwidth]{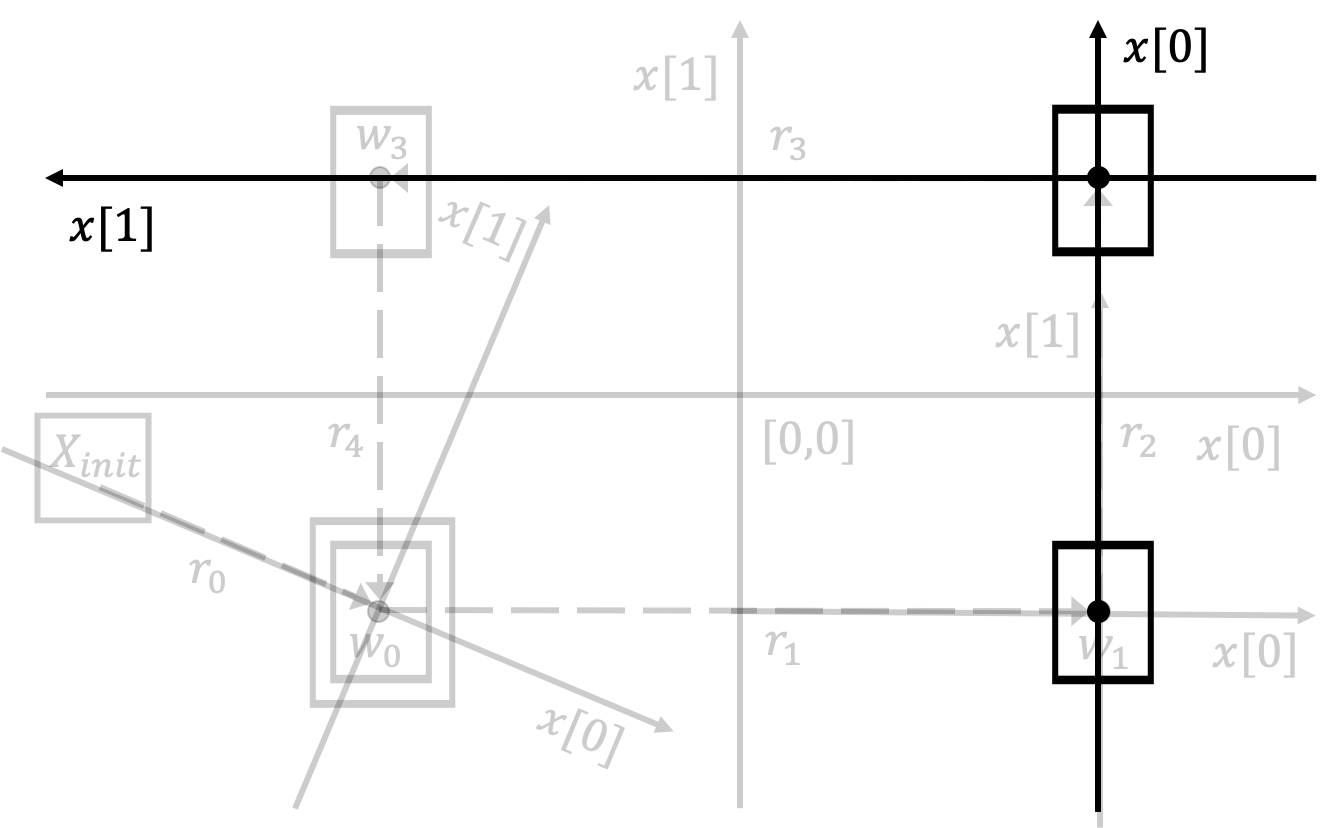} 	
		\caption{ \label{fig:edge_modes_virtual_transformation}}
		\vspace{\floatsep}
	\end{subfigure}
	\begin{subfigure}[t]{0.5\textwidth}
		\centering
		\includegraphics[width=\textwidth]{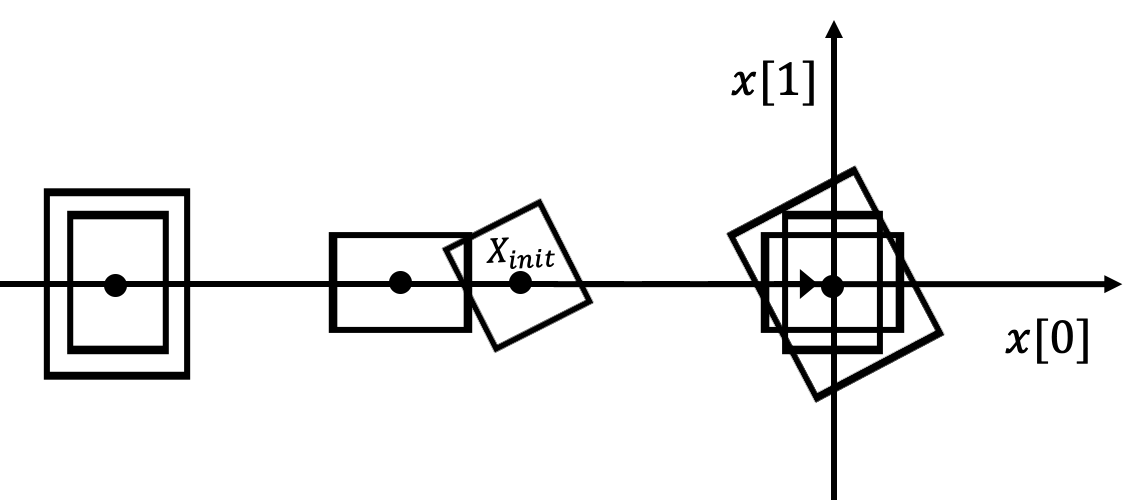}
		\caption{\label{fig:edge_modes_virtual}}
		\vspace{\floatsep}
	\end{subfigure}
	\begin{subfigure}[t]{0.5\textwidth}
		\centering
		\includegraphics[width=\textwidth]{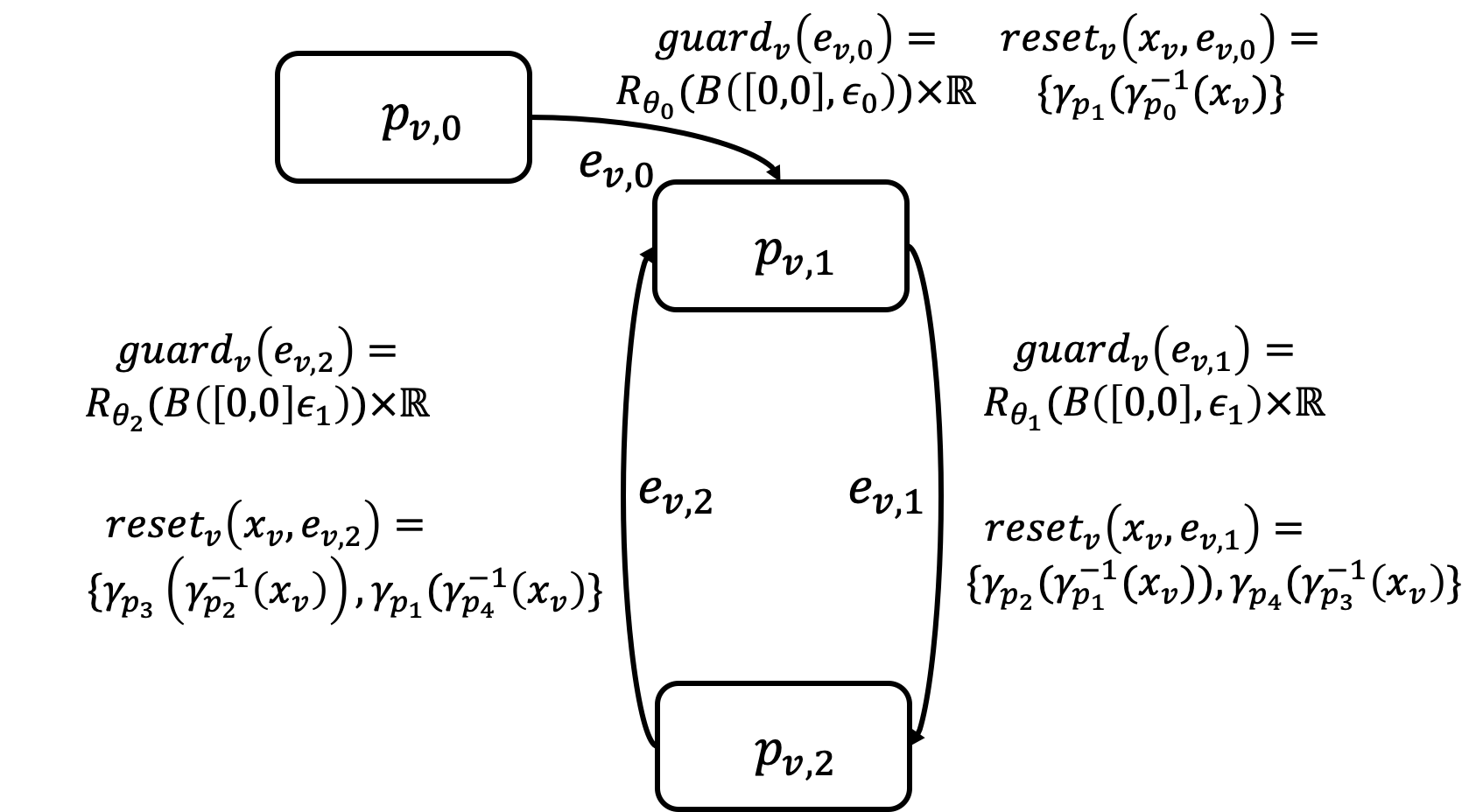}
		\caption{\label{fig:edge_modes_virtual_statemachine}}
		\vspace{\floatsep}
	\end{subfigure}
	\caption{\scriptsize (\ref{fig:edge_modes_virtual_transformation}) the symmetry transformation that changes the coordinate system of the plane of Figure~\ref{fig:problem_description} so that $r_2$ is the new $x_[0]-$axis and wayppoint $w_2$ is the new origin. (\ref{fig:edge_modes_virtual}) shows the virtual initial set, guards and reseted guards of the virtual automaton of Figure~\ref{fig:problem_description} after choosing the virtual map to be the set of coordinate transformations to the roads. 
		(\ref{fig:edge_modes_virtual_statemachine}) shows the resulting hybrid automaton $R_{v}$, where $\theta_0 = \arctan_2(r_0.\dest[1] - r_0.\src[1], r_0.\dest[0] - r_0.\src[0])$, $\theta_1 = \arctan_2(r_1.\dest[1] - r_1.\src[1], r_1.\dest[0] - r_1.\src[0])$, $\theta_2 = \arctan_2(r_2.\dest[1] - r_2.\src[1], r_2.\dest[0] - r_2.\src[0])$. \label{fig:edge_modes_virtual_automaton}}
\end{figure}
\end{example}

\section{Faster safety verification using Symmetry-based abstraction}
\label{sec:abstraction_application}
In this section, we show an example application of the abstraction in accelerating safety verification of hybrid automata. 

The key factors that motivate this section are:
\begin{inparaenum}[(a)]
\item in general, the computation of the reachset of the virtual automaton $\ha_v$ reaches a fixed point faster than that of the concrete one $\ha$,
\item one can obtain the reachset of the concrete automaton $\Reach_\ha$ by transforming the reachset of the virtual one $\Reach_{\ha_v}$, and
\item in general, transforming reachsets is computationally cheaper than computing them from scratch.
\end{inparaenum}

Before delving into the details, in this section, we add the assumption that the modes of $\ha$ have time bounds $\timebound: P \rightarrow \nnreals$. Thus, guards may be ignored only up to a point\footnote{This is a simplification for analysis but would not hurt generalizability as the same mode can be visited several times. The results of the paper extend naturally to models with  urgent transitions.}.  
%The system may continue to evolve in a given mode $p \in P$ beyond a state $\stateinstance$ that satisfies $\guard(p,p')$, 
However, once the total time in mode $p$  reaches $\timebound(p)$, the trajectory must stop. If the time bound is reached and some of the guards are satisfied, then $\ha$ transitions over any of the corresponding edges nondeterministically. Otherwise, the execution stops.
Accordingly, we adjust the definition of $\ha_v$ in Definition~\ref{def:hybridautomata_virtual} to include time bounds $\timebound_v: P_v \rightarrow \nnreals$ on the virtual modes. Specifically, for any $p \in P_v$, 
\begin{align}\timebound_v(p_v) = \max\limits_{\substack{p \in \rv^{-1}(p_v)}} \timebound(p).
\label{item:def_virtual_timebound}
\end{align}
It is easy to see that the FSR in Theorem~\ref{thm:fsr_concrete_virtual} is still valid.

\begin{comment}
# SHRINKING
This section is organized as follows: first, we define the safety verification problem, its existing solutions, and its challenges in Section~\ref{sec:safety_verification_problem}. Then, in Section~\ref{sec:relating_vir_real_reachsets},  we present an existing theorem  that shows that transforming the reachset of a mode in the concrete automaton $\ha$ results in a reachset of a mode of the virtual one $\ha_v$, using the virtual map $\Phi$.  Third, we show the fixed point condition that we adopt in the computation of hybrid automaton reachset and show how to get the concrete automaton reachset from that of the virtual one.
%a cache where we store the per-mode reachsets
Finally, in Section~\ref{sec:unbounded-verification}, we show that reaching a fixed point in the computation of the reachset  of $\ha_v$ %means that no more computation from scratch of reachsets is needed, which
 allows for unbounded-time safety verification of $\ha$. 
\end{comment}

%, for any $i \in \mathbb{N}$, one can compute the $i^{\mathit{th}}$ mode reachset in any path of the concrete automaton

\subsection{Safety verification problem: definition, existing solutions, and their challenges}
\label{sec:safety_verification_problem}
The {\em bounded safety verification}  problem is to check if any state reachable by $\ha$ within fixed number of transitions is unsafe.
That is, 
given maximum number of transitions $J$ 
and an unsafe set $U \subseteq \stateset$, the problem is to check whether: 
$\reachset_{\ha}(J) \cap U = \emptyset$.
In the {\em unbounded} version of the problem, the number of transitions
may be infinite, 
and we replace the bounded reachset with the unbounded one $\Rtube_\ha$.

% \sayan{The rest of this stuff should move to later when needed.}

% 
%If we restrict the executions to the parameters we get the set:
%\marginpar{\scriptsize{\sayan{Some of these paths will not be realized by any execution right?}}}
%If the number of switches of each execution is equal to $J$ and the maximum duration of any of its trajectories, i.e. time spent in any mode, is bounded by $T$, we denote the reachset 
%within the time interval $[\ftime, \etime]$ 
%by $\reachset(K_0, p_0, J, T)$. 
% Moreover, we denote the set of all paths of executions
\begin{comment}
\begin{align}
\pathset(K_0,p_0) = \{\path(\sigma) \ |\ \sigma \in \exec(K_0,p_0) \},
\end{align}
and the corresponding sets $\pathset(K_0,p_0,T)$ and $\pathset(K_0,p_0, T, J)$.
\end{comment}

\begin{comment}
\hussein{A path that is realized by some execution (visiting a sequence of states) is a {\em feasible path\/}.
Given a sequence of sets of states $\{K_i\}_{i=0}^{J-1}$, a path 
$\pathinstance \{p_0, \ldots, p_{J-1}\} \in \pathset(p_0,J)$ is $\{K_0,\ldots, K_{J-1}\}$-{\em feasible} if for all $i \geq 1$, for any $x_i \in K_i, \exists\ x_{i-1} \in K_{i-1}$, $\exists\ t\leq \timebound(p_{i-1})$, such that $\xi(x_{i-1}, p_{i-1}, t) = x_{i-1}'$, $x_{i-1}' \in \guard(e_{i-1})$ and $x_i \in \reset(x_{i-1}', e_{i-1})$, where $e_{i-1}=(p_{i-1}, p_i)$. we might not need this anymore??} 
\end{comment}

For any 
 path $\path = \{p_i\}_{i \in [J]}$, existing reachability analysis tools compute the reachsets of the modes sequentially. Roughly, for each index $i \in [J]$ in the path, they compute the reachset, or an over-approximation, of the $i^{\mathit{th}}$ mode reachset $\Reach_{\ha,\path,i}$,
 %, $T_i$ being $\timebound_\ha(\path[i])$,
 %$T_0$ being the time bound $T$, 
 intersect $\Reach_{\ha,\path,i}$ with $\guard(\path[i],\path[i+1])$, and then apply $\reset$ to the result of the intersection to get the initial set of states for the next mode $\path[i+1]$ in the path.
 %, then subtract the time it took to reach the guard from $T_i$, to get $T_{i+1}$. 
 Consequently, the reachset of the path would be:
\begin{align}
\label{eq:path_reachset}
	\Reach_{\ha,\path} = \cup_{i \in [J]} \Reach_{\ha,\path,i}.
\end{align} 
We call the modes reachsets $\Reach_{\ha,\path,i}$ {\em reachset segments}.
Moreover, the reachset of the hybrid automaton $\Reach_\ha$ would be:
\begin{align}
\label{eq:automaton_reachset}
 \Reach_{\ha}(J) = \cup_{\path \in \pathset_\ha(J)} \Reach_{\ha,\path},
\end{align} 
where $\pathset_\ha(J)$ is the set of all paths of length $J$ of $\ha$.

%We say that the computation of the reachset $\Reach_{\ha, \path}$ reached a {\em fixed point} at iteration $i^*$, if the 
%$\computereachset$ 
% If it is not deterministic, even more reachset segments should be computed.
As $J$ grows, sequential computation of the per-mode reachsets becomes infeasible. Some existing theorems for unbounded reachset computation, such as Theorem 4.4 in \cite{FanQMV:CAV2017}, assume that the unbounded-time reachset is a bounded set, so that a fixed point can be checked.
However, the reachset $\Reach_\ha$ may become unbounded, which motivates the search for new approaches for unbounded safety verification. That is what we tackle in this section.

\begin{comment}
SHRINKING
Next, we show how constructing a virtual automaton $\ha_v$ and computing its reachset can accelerate the computation of the bounded reachset of the concrete automaton $\ha$. Moreover, we show how that makes the unbounded safety verification problem of $\ha$ feasible, even when its reachset is not a bounded set. The key properties of the virtual automaton behind both of these advantages are: the fewer number of modes, the fewer number of edges, the more concentrated guards, and the more concentrated reseted guards of $\ha_v$ compared to $\ha$.
\end{comment}

\subsection{Relation between concrete and virtual automata reachset segments}
\label{sec:relating_vir_real_reachsets}

% In this section, we show how to over-approximate a mode reachset of the concrete hybrid automaton $\ha$
%, following a particular path $\path = \{p_j\}_{j=0}^{J-1}$, 
% using the reachset of the virtual one $\ha_v$.
\begin{comment}
# SHRINKING
In this section, we show the advantage of transforming reachsets over computing them from scratch. Moreover, we show how to map concrete reachsets to virtual ones, using a theorem from~\cite{Sibai:TACAS2020}.
\end{comment}

\begin{comment}
SHRINKING: Transforming vs. computing
Transforming reachsets is an easier problem and has a faster solution than computing them from scratch. For example, if the symmetry maps are linear, e.g. translation, rotation, permutation etc., and the representation of the reachsets is linear as well, e.g. list of polytopes, the transformation would have a polynomial time complexity after the  computation from scratch was hard~\cite{HENZINGER199894}. 
\end{comment}
 % Theorem~\ref{thm:tube_trans_input} allows computing new reachset segements by transforming existing ones.

The ability of symmetry maps to transform solutions to other solutions of the system extends to transforming reachsets to other reachsets of the system. The following theorem, restated from \cite{Sibai:TACAS2020}, formalizes this for parameterized dynamical systems. 
% It will be the essential building block and reason behind constructing the virtual system in the following section.

\begin{theorem}[Theorem 2 in \cite{Sibai:TACAS2020}]
	\label{thm:tube_trans_input}
	If (\ref{sys:input}) is $\Gamma$-equivariant, then for any $\gamma \in \Gamma$ and its corresponding $\rho$, any initial set $K \subseteq \stateset$, mode $p\in P$, and $T\geq 0$,
	%\chuchu{
	\[
	\reachset(\gamma(K), \rho(p), T) = \gamma (\reachset(K,p, T)),
	\]
	where $\Reach(K,p,T)$ is the reachset of (\ref{sys:input}) starting from  $K$, having mode $p $, and a time bound $T$.
	%}
\end{theorem}

% In this paper, we show how to compute and represent the unbounded reachset as a pair of a $\permodedict$ of an abstraction of the system and family of symmetry transformations. 
% This pair can be used to get any reachset segment of the concrete system without computing all of the segments of the modes leading to it.

% and the set of modes in $P$ that appear in any edge in $\edgeset$ by $P_\edgeset$.
% The set of all paths that start from $p_0$ is denoted by $\pathset(p_0)$. Similarly, when the lengths of the paths is upper bounded by $J$, the set is denoted by  $\pathset(p_0,J)$. 
% For any $p \in P_\edgeset$, the set of all modes $p' \in P$, where $(p',p) \in \edgeset$, is denoted by $\inmodes(p)$ and the set of all modes $p' \in P_\edgeset$, where $(p,p') \in \edgeset$, is denoted by $\outmodes(p)$.
%  sequence of sets that contain the over-approximation of the reachset over small time intervals.

\subsection{Relation between concrete and virtual automata reachsets}
\label{sec:real_from_virtual_reachset}
%In this section, we show how to transform the reachset of the virtual automaton to get the reachset of the concrete one using the virtual map. This is analogous to mapping virtual executions to concrete ones in Section~\ref{sec:fsr}.

\begin{comment}
SHRINKING
In this section, we show how we compute the reachset of the virtual automaton and how we check for the fixed point. Then, we show how to over-approximate the reachset of the concrete automaton, after the fixed-point of the reachset of the virtual automaton has been reached.
\end{comment}

Consider any algorithm or tool that computes the reachset of a hybrid automaton as in equations~(\ref{eq:path_reachset}) and (\ref{eq:automaton_reachset}). Lets call it $\computereachset$ and call its output $\globalR$. We use $\computereachset$ to compute the reachset of the virtual automaton $\ha_v$.
We annotate $\computereachset$ to keep track of the union of the initial sets and the union of the reachsets being computed corresponding to each unique mode $p_v \in P_v$, in a data structure and call it $\permodedict$. For every $p_v \in P_v$, it stores the initial set in the first argument, which we denote by $\permodedict[p_v].\modeinitset$, and the reachset in the second argument, which we denote by $\permodedict[p_v].\modereachset$. After each computation of a new reachset, or an over-approximation thereof,  $\Reach_{\ha_v,\path_v,i}$ of a mode in a path, it gets added to the entry corresponding to $\path_v[i]$ in $\permodedict$. Then, $\permodedict$ is used to check if a fixed point of the automaton reachset $\globalR$,
% that is an over-approximation of $\Reach_{\ha}$, being computed
  has been reached.  This checks if continuing running of $\computereachset$ would not change $\globalR$, and then stopping the run accordingly. 

Our check for the fixed point is implemented in the function $\checkFixedPoint$. It takes as input $\permodedict$ and a hybrid automaton for which we are computing the reachset, here $\ha_v$, and returns $\mathtt{True}$ or $\mathtt{False}$.
% following a corresponding path 
% per the FSR defined in Theorem~\ref{thm:fsr_real_virtual}.

The function $\checkFixedPoint(\permodedict)$ returns $\mathtt{True}$ if and only if: $\initsetv \subseteq \permodedict[\initmodev].K$ and $\forall\ p_v \in P_v,$
\begin{align}
\label{eq:fixed_point_condition}
&\cup_{p_v' \in \inmodes (\mathit{p_v})} \reset(\permodedict[p_v'].R \cap  \guard(p_v',p_v), (p_v',p_v)) \nonumber \\
 &\hspace{0.5in}\subseteq \permodedict[p_v].K,
\end{align}
where $\inmodes(\mathit{p_v})$ $= \{p_v' \in P_v\ |\ (p_v',p_v) \in \edgeset_v\}$, and returns $\mathtt{False}$, otherwise.

\begin{theorem}
	\label{thm:fixed_point_condition}
	If $\checkFixedPoint(\permodedict)$ returned $\mathtt{True}$,
	then for all $\path_v \in \pathset_{\ha_v}$ and $i \in [\path_v.\mathit{len}]$,
	%  \Reach_{\ha,\path,i} &\subseteq \permodedict[].K, \\
	\begin{align}
	\label{eq:fixed_point_puzzle_pieces} 
	\Reach_{\ha_v,\path_v,i} \subseteq \permodedict[\path_v[i]].\modereachset,  \text{ and} 
	\end{align}
	 $\reachset_{\ha_v} \subseteq \globalR$.
\end{theorem}
\begin{proof}(Sketch)
Recall that we assume that $\computereachset$ is a sound algorithm for computing reachsets of any hybrid automaton. Thus, for every $p_v \in P_v$, $\permodedict[p_v].\modereachset$ is an over approximation of the reachset of (\ref{eq:robot_dynamics}) starting from $\permodedict[p_v].\modeinitset$ and running for time bound $\timebound_v[p_v]$.
 Now, we fix any path $\path_v. \in \pathset_{\ha_v}$.
 % Then, $\Reach_{\ha,\path,i} \permodedict[i].\modereachset \cap \guard(\path[i],\path[i+1])$
	The proof is by induction over the indices of $\path_v$. 
	
	Base condition: for $i=0$, $\Reach_{\ha_v,\path_v,i} \subseteq \permodedict[i].\modereachset$, by the soundness of $\computereachset$ and the assumption that $\initsetv \subseteq \permodedict[\initmodev].\modeinitset$ in equation~(\ref{eq:fixed_point_puzzle_pieces}).
	
	Hypothesis: fix an index $i\geq 0$. Assume that the theorem is satisfied for all previous indices $i' \leq i$.
	
	Induction: we want to prove that $\Reach_{\ha_v,\path_v,i+1} \subseteq \permodedict[\path_v[i+1]].\modereachset$.
	%and $\globalR_{i+1} \subseteq R^*$.
	
	We know from equation~(\ref{eq:path_reachset}) and the discussion there, that the initial set of the $(i+1)^{\mathit{th}}$ mode in $\path_v$ is $\reset(\Reach_{\ha_v,\path_v,i} \cap \guard(\path_v[i],\path_v[i+1]),$ $(\path_v[i],\path_v[i+1]))$. But, we know from the induction assumption that $\Reach_{\ha_v,\path_v,i} \subseteq$ $\permodedict[\path_v[i]]$. Moreover, we know from equation~(\ref{eq:fixed_point_condition}) that  $\reset(\permodedict[\path_v[i]] \cap \guard(\path_v[i],\path_v[i+1]),$ $(\path_v[i],\path_v[i+1])) \subseteq $ $\permodedict[\path_v[i+1]].\modeinitset$. Hence, $\reset(\Reach_{\ha_v,\path_v,i} \cap \guard(\path_v[i],\path_v[i+1]),$ $(\path_v[i],\path_v[i+1])) \subseteq$  $\permodedict[\path_v[i+1]].\modeinitset$.
	Using the soundness of $\computereachset$ again results in $\Reach_{\ha_v,\path_v,i+1} \subseteq \permodedict[\path_v[i+1].\modereachset$. Hence, (\ref{eq:fixed_point_puzzle_pieces}) is satisfied.

	Since the path chosen is arbitrary and the automaton reachset is the union of the path reachsets per equation~(\ref{eq:automaton_reachset}), then $\Reach_{\ha_v} \subseteq \globalR$.
\end{proof}

\begin{corollary}
	\label{cor:vir_dict_to_real_reachset}
	Let $\permodedict$ be the result  after
	%from running $\computereachset$ on the virtual automaton $\ha_v$, and
	 reaching a fixed point in computing $\Reach_{\ha_v}$. Then, for any $J \in \mathbb{N}$ and path $\path \in \pathset_\ha(J)$ of the concrete automaton $\ha$, 
	\begin{align}
	\label{eq:vir_dict_to_real_reachset}
	\Reach_{\ha,\path} \subseteq \cup_{i \in [J]} \gamma_{\path[i]}^{-1}(\permodedict[\path_v[i]]),
	\end{align} 
	where $\path_{v}[i] = \rv(\path[i])$.
\end{corollary}
\begin{proof}
	It follows from Theorem~\ref{thm:tube_trans_input}, Theorem~\ref{thm:fixed_point_condition}, and the invertiblity assumption on $\gamma_p$, $\forall p \in P$.
\end{proof}

Hence, we can compute an over-approximation of the reachset of the concrete automaton $\ha$ by running $\computereachset$ to compute $\Reach_{\ha_v}$ and get its $\permodedict$, iterate over all the paths of $\ha$, transform back the corresponding reachset from $\permodedict$ at each mode visited to get the reachset segment.
However, we do not need to iterate over them sequentially as they appear in the path, since for any $i \in [J]$, the $i^\mathit{th}$ entry of the sequence on the right-hand-side of (\ref{eq:vir_dict_to_real_reachset}) only depends on $\path[i]$. Therefore, we can get the $i^{\mathit{th}}$ reachset segment of the path by just transforming the corresponding $\permodedict$ entry. The transformation is done using the symmetry in $\Phi$ corresponding to the $i^{\mathit{th}}$ mode. Thus, there is no need to compute the whole reachset before reaching the $i^{\mathit{th}}$ mode, to get its reachset.
This will be the key result that helps getting unbounded safety results in the next section.

\subsection{Unbounded safety verification of hybrid automata}
\label{sec:unbounded-verification}
In this section, we cultivate all the theorems presented to verify the safety of hybrid automata with symmetric continuous dynamics.
Our approach is summarized in Algorithm~\ref{code:unbounded}, which we name $\unboundedverif$. We use the symbol $\rightsquigarrow$ to denote that there exists a path from the source mode to the destination one in $\ha$.

 \begin{algorithm}
	\small
	\caption{$\unboundedverif$}
	\label{code:unbounded}
	%\begin{algorithm}	
	\begin{algorithmic}[1]
		\State \textbf{input:} $\ha, \Phi, U, J$
		\State $\ha_v	\gets {\sf constructVirtualModel} (\ha, \Phi)$
		\State $\permodedict_v \gets \computereachset(\ha_v, J)$
		\If{not $\exists\ p \in P, \initmode \rightsquigarrow p,\text{ and }  \gamma_{p}^{-1}(\permodedict[\rv(p)].\modereachset) \cap U \neq \emptyset, $}
		\State \textbf{return: } $\safe$
		\Else
		\textbf{ return: } $\mathit{unknown}$
		\EndIf
	\end{algorithmic}
\end{algorithm}

The following theorem shows the soundness of $\unboundedverif$. It follows from Theorem~\ref{thm:fixed_point_condition} and Corollary~\ref{cor:vir_dict_to_real_reachset}. 
  
\begin{theorem}[main theorem]
\label{thm:unboundedcodecorrectness}
Given any hybrid automaton $\ha$, virtual map $\Phi$, an unsafe set $U \subseteq \stateset$, and a $J \in \mathbb{N} \cup \{\infty\}$,
if $\unboundedverif$ returned $\safe$, then $\Reach_{\ha}(J) \cap U = \emptyset$.
\end{theorem}

There is still a possibility that $\computereachset$ would not reach a fixed point in the case of infinite $J$, in that case, the result is also unknown as well. In that case, the computed reachset should be refined to reduce its over-approximation error, for example, by partitioning the initial set~\cite{C2E2paper,FanQMV:CAV2017} or having higher order Taylor series approximations~\cite{Flow}. Or, the abstraction should be refined, for example, by utilizing more symmetry properties. 

The significance of the theorem is that, once a fixed point in computing $\Reach_{\ha_v}$ is reached and resulted in $\permodedict$, the safety verification problem gets reduced from computing a sequence of reachset segments and intersections with the guards, to searching over the modes if their transformations of their corresponding segments in $\permodedict$ would intersect the unsafe set. For example, assume that the fixed point has been reached after computing five reachset segments. We would be able to compute the reachset segment, or an over-approximation of it, of the hundredth mode in the path of the concrete system, without computing the ninety nine segments of the modes visited before reaching it, but only by transforming one of the five stored segments.
For infinite $J$, one can use a Satisfiability Modulo Theory (SMT) solver, for example, to search for a symmetry map $\gamma_p$ in $\Phi$ corresponding to a mode $p$ in an infinite path of $\ha$, that makes $\gamma_p^{-1}(\permodedict[\rv(p)].R) \cap U \neq \emptyset$. 
\section{Experimental evaluation}
\label{sec:experiments}

We implemented $\unboundedverif$ in Section~\ref{sec:unbounded-verification} in Python 3, for finite $J$.
% on top of our tool $\ourtacastool$~\cite{Sibai:TACAS2020}.
Our implementation includes:
\begin{inparaenum}[(1)]
 \item the function ${\sf constructVirtualModel}$, that constructs the virtual automaton $\ha_v$ from a given concrete automaton $\ha$ and virtual map $\Phi$, 
 \item the function $\computereachset(\ha_v, J)$, which computes the reachset of $\ha_v$ using $\ourtacastool$~\cite{Sibai:TACAS2020} for finite $J$ transitions while checking for fixed point as in equation~(\ref{eq:fixed_point_condition}), and
 \item  the data structure $\permodedict$ which caches computed reachset segments as in Section~\ref{sec:unbounded-verification}.
\end{inparaenum} 
% We call the new version of the tool: $\ourtool$. 

We ran several experiments to illustrate the usability and advantage of our method over existing ones.
% and compare the results
%, its advantage over $\ourtacastool$, and its advantage over not using symmetry to construct reachsets. 
Moreover, we illustrate the key parameters that affect the effectiveness of our abstraction and the quality of our results. 

% We tested it on a linear dynamical system and the aircraft model of Example~\ref{sec:single_drone_example}  using DryVR~\cite{FanQMV:CAV2017} and Flow*~\cite{CAS13} as reachability subroutines.
 \begin{comment}
\begin{enumerate}
\item 
\item
% This is shown as well by comparing the linear model, where the symmetry map change all of the three dimensions, and the fixed-wing aircraft one, where symmetry 
\item 
\item \hussein{Unbounded safety results by checking the safety of the linear example with infinite S-shaped path against an unsafe set representing a building near the $100^{\mathit{th}}$ segment.}
\item \hussein{Our method is time efficient. We compare the verification running time needed by our tool $\ourtool$ with $\ourtacastool$, of \cite{Sibai:TACAS2020}}. 
\end{enumerate}
\end{comment}

\subsection{Implementation details}
\label{sec:implementation_description}

The
% input-output interface of $\ourtacastool$ is kept without alteration. Basically, its
 inputs of our implementation are a  {\em scenario} and {\em dynamics} files describing the hybrid automaton $\ha$ and the virtual map $\Phi$. Its outputs are the virtual automaton $\ha_v$ and its reachset segments stored in $\permodedict $, the reachset of $\ha$, and the safety decision. 

The scenario file specifies $\A$ that is: initial set $\initset$ as a hyper-rectangle, unsafe set $U$ as a list of hyperrectangles, set of modes $P$ as a list of tuples, path as a finite sequence of modes $\pathinstance = \{p_i\}_{i=0}^{J-1}$, guards as a list of hyper-rectangles, and time bounds $\timebound$ as a list of floats. 

The dynamics file specifies the dynamic function $f$, which given a state $x$ and mode $p$, returns $f(x, p)$. It has also three other functions that implement the virtual map $\Phi$.  The first two functions, given a mode $p \in P$ and a set of states $S$ that is represented as a Tulip polytope\footnote{\url{https://tulip-control.github.io/polytope/}}~\cite{Tulip}, one returns $\gamma_p(S)$ and the other $\gamma_p^{-1}(S)$.
The third function, given a mode $p \in P$, would return $\rv(p) = \rho_{p}(p)$.

Our implementation has 3 options for  computing $\Reach_\ha$:
\begin{itemize}
\item the standard method without symmetry (\Nosym\  or \NS),
\item the method of~\cite{Sibai:TACAS2020} using symmetry and caching (\Symcache\ or \SC), and
%to construct $\Reach_\ha$. But, unlike Algorithm~\ref{code:unbounded} of Section~\ref{sec:unbounded-verification}, we did not construct $\ha_v$ or compute its reachset $\Reach_{\ha_v}$ explicitly, check for fixed point, or store $\ha_v$ reachset segments in a data structure $\permodedict$;
\item $\unboundedverif$ (Algorithm~\ref{code:unbounded}, \Symvir\ or \SV).
%, that is building $\ha_v$, computing $\ha_v$ reachset till reaching a fixed point, and transforming it back to get the reachset of $\ha$. While computing the reachset of $\ha_v$, we cache and reuse reachset segments whenever possible.
\end{itemize}
%\hussein{this can't be understood}
% In \cite{Sibai:TACAS2020},
% In computing the reachset of $\ha_v$, it uses caching as in \cite{Sibai:TACAS2020}. Using caching is possible since all the virtual modes have the same continuous dynamics, they follow the origin as the waypoint. 
%It caches and reuses computed reachset segments for different virtual modes.
% as the tool $\ourtacastool$ in \cite{Sibai:TACAS2020}, since all the modes share the same continuous dynamics.
{\em Reachability analysis:} 
%the first method, \NS, is the one we described in Section~\ref{sec:safety_verification_problem}.
To fix the over-approximation error added by the reachability analysis tool used, our implementation grids both the state space $\stateset$ of $\ha$ and the state space $\stateset_{v}$ of $\ha_v$ into equally-sized cells. 
Then, for \NS, while computing $\Reach_\ha$ as in Section~\ref{sec:safety_verification_problem}, for any initial set $\modeinitset$ for which a reachset segment has to be computed, it checks which cells of the grid over $\stateset$ are occupied by $\modeinitset$. For \SC, it would transform $\modeinitset$ using $\Phi$ to get the initial set $\modeinitset_v$ in $\stateset_v$. After that, it check which cells of the grid over $\stateset_v$ are occupied by $\modeinitset_v$. For \SV, while computing $\Reach_{\ha_v}$, for any initial set $\modeinitset_v$ in $\stateset_v$ for which a reachset segment has to be computed, it checks which cells of the grid over $\stateset_v$ are occupied by $\modeinitset_v$.
%For \SV, since it would be computing the $\Reac$it would have $\modeinitset_v$ from the beginning.
Then, for \NS, it computes the reachsets for each of these cells, and unions the results. 
For \SC\ and \SV, it checks first which cells have their reachsets computed before and cached. For those cells, the reachsets are retrieved from the cache instead of computed. For cells with no cached reachsets, they get computed and cached. 
After that, our implementation unions the cells reachsets and transform the result back to $\stateset$ using $\Phi$.
Doing this would fix the over-approximation error added by the reachability analysis tool, since that depends on the size of the initial set it is asked to compute the reachset for. 
% \hussein{griding the quality of the reachability analysis for a given tool depends on the size of the initial set} is that when the reachability analysis tool being used and its parameters thereof are fixed, the quality of a reachset is determined by the size of the initial set. 
% Hence, by fixing the size of the cells from which the reachsets are computed, we are fixing the quality of the reachsets that we obtain from any of the three methods that we consider above. 
Thus, fixing the grid size would enable quantifying the over-approximation error added because of using symmetry after fixing the error added by the reachability tool. Moreover, this would allow for a fairer comparison of computation time between the methods.

%The second method \SC, follows the same steps as \NS. However, for any initial set $\modeinitset$ for which a reachset segment has to be computed, it transforms it 

For \SV, after reaching a fixed point in the computation of $\Reach_{\ha_v}$, the resulting $\permodedict $ is used to get the reachset segments of $\ha$, without the griding-based method described.
% before reaching the fixed point in the computation of $\Reach_{\ha_v}$, whenever a reachset segment is computed as described above, it gets transformed directly to a reachset segment of $\ha$. $\permodedict $ would be used to compute a

%\hussein{check for better place} However, when using the $3^{\exprd}$ method, once a fixed point is reached, and since we only need the reachset segments saved in $\permodedict $, our code transforms them directly to get the reachset segments of $\ha$. 

 Any of the existing reachability analysis tools can be used, such as DryVR~\cite{FanQM18} and Flow*~\cite{Flow}. Our implementation has both of these tools as options for the user. However, for comparison purposes in our experiments, it simulates the continuous dynamics from the center state of a given cell using an ODE solver, and then places hyper-rectangles equal to the size of the cell at each state in the simulation. The union of these rectangles is considered the reachset. 
 %This would produce reasonable, but unsound, reachsets. 
 This is computationally cheaper than any of the existing reachability analysis tools. This will ensure that the computation time improvement of \SC\ and \SV\ over \NS\ is not due to using a computationally expensive reachability tool. Such a tool will make the advantage \SC\ and \SV\ even bigger, as they retrieve and transform some of the reachsets versus computing them.
 % This will ensure that our computation improvement in transforis not due to using 
 %For example, the hybrid systems verification tool DryVR~\cite{FanQM18}, %option in $\ourtacastool$ for reachsets computation.
%would simulate the dynamics starting from several random states in the cell, solve an optimization problem,  and then bloat the center simulation
% to compute sensitivity, before bloating the center simulation with the computed sensitivity, 
%to obtain the reachset.
% Still, transforming existing reachset of the cell is computationally cheaper than computing a new one from scratch using this method. For any professional use of $\ourtacastool$, it has the options of using DryVR or Flow*\cite{Flow}.
%we will show in our results in Table~\ref{table:experimental_results}
%

% After the reachsets of all the cells are obtained, $\ourtacastool$ unions them together to get the reachset segment. 
%We do that to fix the quality of the computed reachsets segments as the over-approximation error of the reachset computed by DryVR increases as the size of the initial set increases~\cite{FanQMV:CAV2017}.

%We plan to extend the implementation to the unbounded case using SMT solvers as in (\ref{eq:unboundedverif_SMT}).
Finally, to check if a scenario is safe, our implementation intersects the computed reachset and the unsafe set.

\subsection{Scenarios and metrics }
\label{sec:scenarios}

% We tested our code in different scenarios, and we show the results in several tables and figures in this section. In all of the scenarios considered, we modeled the roads to be the modes of the automaton $\ha$.

We consider several scenarios and virtual maps in our experiments, with the following features:
%We describe them as follows: 
%For each aspect we present the notation that we use to refer to in the tables and figures. 
%The {\em scenarios} meta-column describes the scenarios of the examples we consider in our experiments: the {\em id} column is for easy referencing for the scenarios in the paper.
\begin{enumerate}%[(1)]
\item  {\em dynamic} function of the agent in the scenario:
\begin{inparaenum}[(a)]
	\item  {\em robot}, that of equation~(\ref{eg_item:def_2_robot_dynamics}) of Example~\ref{sec:single_robot_example_edges}, and 
	\item  {\em linear}, the linear stable dynamics of the form $\stateinstance = \mathit{diag}([-3,-3,-1]) (\stateinstance - p)$, where $\stateinstance$ and $p \in \reals^3$,
\end{inparaenum}
\item {\em path} followed by the agent as a sequence of modes: 
\begin{inparaenum}[(a)]
\item {\em rectangle} (rect.), that of Figure~\ref{fig:problem_description} followed for 4 full turns for a total of 16 roads,
\item {\em Koch snowflake} (Ko.), a truncated Koch snowflake\footnote{The actual Koch snowflake is a fractal. Here we truncate the construction after a finite number of iterations to get a snowflake shape with finite edges. } path with 16 roads (see Figure~\ref{fig:linear-snowflake-t-r-sym2-real}),
\item {\em random} (rand.), a random path of 14 roads (see Figure~\ref{fig:linear-random-t-r-sym2-real}), and 
\item finally, {\em S-shaped}, an $S$-shaped path with 16 roads (see Figure~\ref{fig:linear-S-t-r-sym2-real}),
\end{inparaenum}
\item virtual map $\Phi$ used to construct $\ha_v$: 
\begin{inparaenum}[(a)]
\item {\em T}, translation map that translates the coordinates of both waypoints of a road so that the end point is the origin, and
\item {\em TR}, combines {\em T} with rotation of axes so that the road is the new $x[1]$-axis. That is in contrast with Example~\ref{sec:transformation_example_edges}, where we chose the road to be the $x[0]$-axis, and
\end{inparaenum}
\item {\em sym} $\in \{$\NS, \SC, \SV$\}$, method used to computed $\Reach_\ha$.
%: $1,2,3$ refer to the first, second, and $3^{\exprd}$ methods, respectively. We can think of it as an ordered set of methods with increasing level of using symmetry.
\end{enumerate}

For each of the scenarios, we collected several statistics:
\begin{enumerate}
\item {\em computed} (\#{\em co}),  {\em retrieved} (\#{\em re}), {\em copied} (\#{\em cp}), {\em total} (\#{\em tot.}): are the numbers of cells reachsets that has been computed from scratch, has been retrieved from the cache, number of reachset segments that has been transformed from $\permodedict $, after a fixed point has been reached, and the sum of all three numbers, respectively,
\item \#$m / e$, are the numbers of modes and edges of $\ha_v$,
% \# {\em pre-f} represents the mode reachset segments computed before reaching a fixed point.
\item {\em time}, in minutes, is the total time needed to compute the reachset of $\ha$, and
\item over-approximation {\em error} (\%) is the error added to the reachset of $\ha$, due to using symmetry:
\[
\mathit{error} = \mathit{avg}_{i\in [\mathit{path}.\mathit{len}]}\frac{\mathit{Vol}(\modeinitset_{j,i}) - \mathit{Vol}(\modeinitset_{\text{\NS},i})}{\mathit{Vol}(\modeinitset_{\text{\NS},i})} \times 100,
\]
where $\mathit{Vol}(\cdot)$ returns the volume of the given hyperrectangle and $\modeinitset_{j,i}$ denotes the $i^{\mathit{th}}$ mode initial set of the path using method $j \in \{$\NS, \SC, \SV$\}$.

% Finally, comparing the computational expense of computing a cell reachset from scratch 
%It is meant to measure an estimate of the over-approximation error added due to using symmetry to compute the reachset of $\ha$. 
%It is the average of the relative volume of the initial sets of the modes  of the ratios of the volumes  of the mode initial sets to the corresponding mode initial sets when not using symmetry. 
% shows the difference between the total number of grid reachset segments using $\ourtool$ or $\ourtacastool$ while using symmetry versus using  $\ourtacastool$ without symmetry.
\end{enumerate}

% for longer paths, that are not necessarily infinite, reaching a fixed point would improve computation time  
%However, for this particular rectangle path, 
% One can see that more tubes has been computed when using the virtual automaton 
%The last two rows show the results where we use the translation virtual map instead.
\begin{comment}
\begin{figure*}[t!]
	%\caption{Reachtubes for drone~\ref{fig:linear_nosym_1_agents_sym} and linear~\ref{fig:linear_nosym_3_agents_sym} models  using {\sf Sym-Flow*}. Three agents \label{fig:flow_linear}}
	\centering
	\includegraphics[width=0.25\textwidth]{Figures/robot-rect-t-r-sym0/real.png}
	\caption{robot, $\mathit{sym = 0}$\label{fig:robot-rect-t-r-sym0-real}}
	\vspace{\floatsep}
\end{figure*}
\end{comment}

\begin{figure*}[t!]
	%\caption{Reachtubes for drone~\ref{fig:linear_nosym_1_agents_sym} and linear~\ref{fig:linear_nosym_3_agents_sym} models  using {\sf Sym-Flow*}. Three agents \label{fig:flow_linear}}
\begin{subfigure}[t]{0.25\textwidth}
	\centering
	\includegraphics[width=\textwidth]{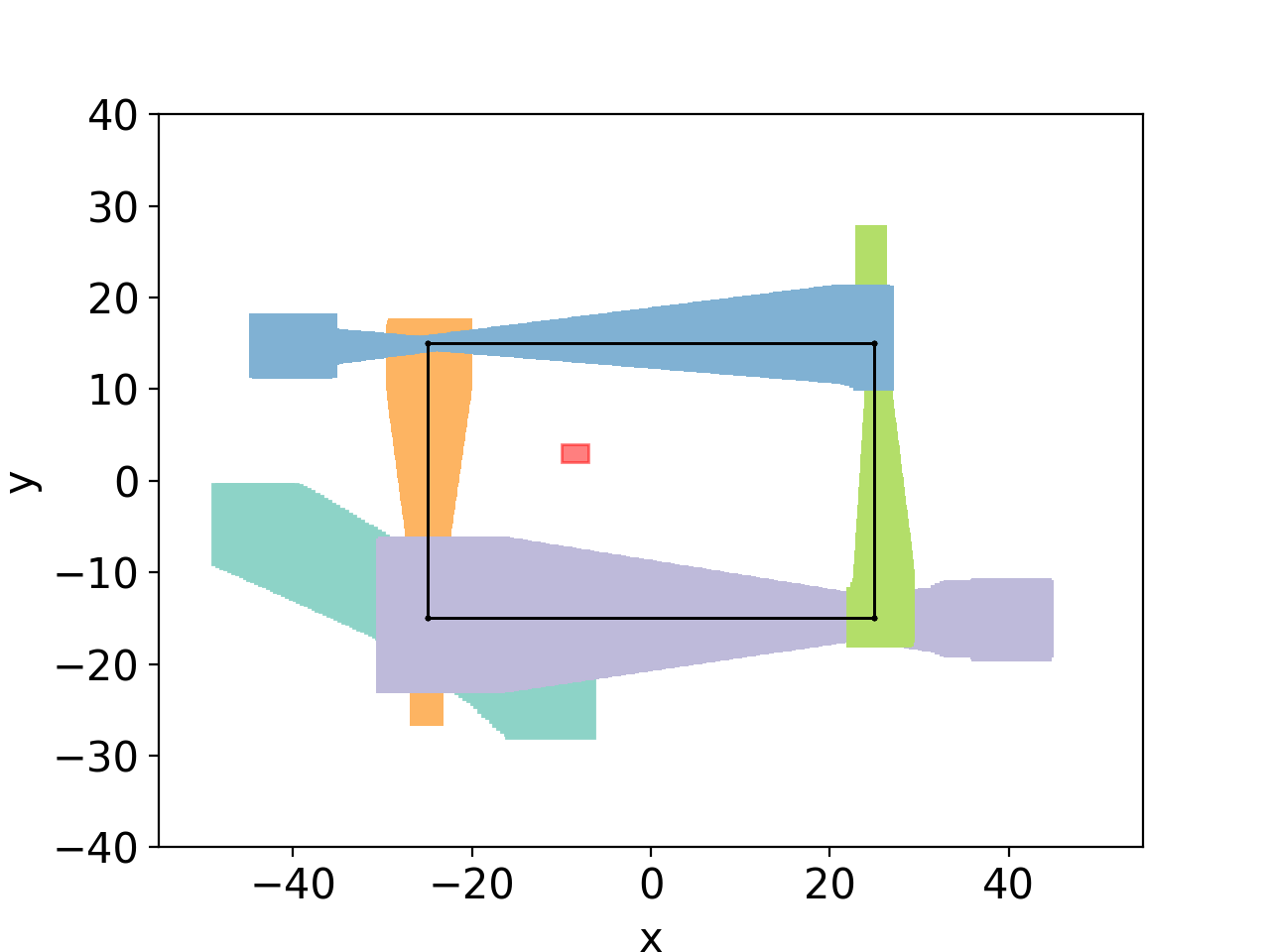}
	\caption{\label{fig:robot-rect-t-sym2-real}}
	\vspace{\floatsep}
\end{subfigure}
\hspace{-0.1in}
\begin{subfigure}[t]{0.25\textwidth}
	\centering
	\includegraphics[width=\textwidth]{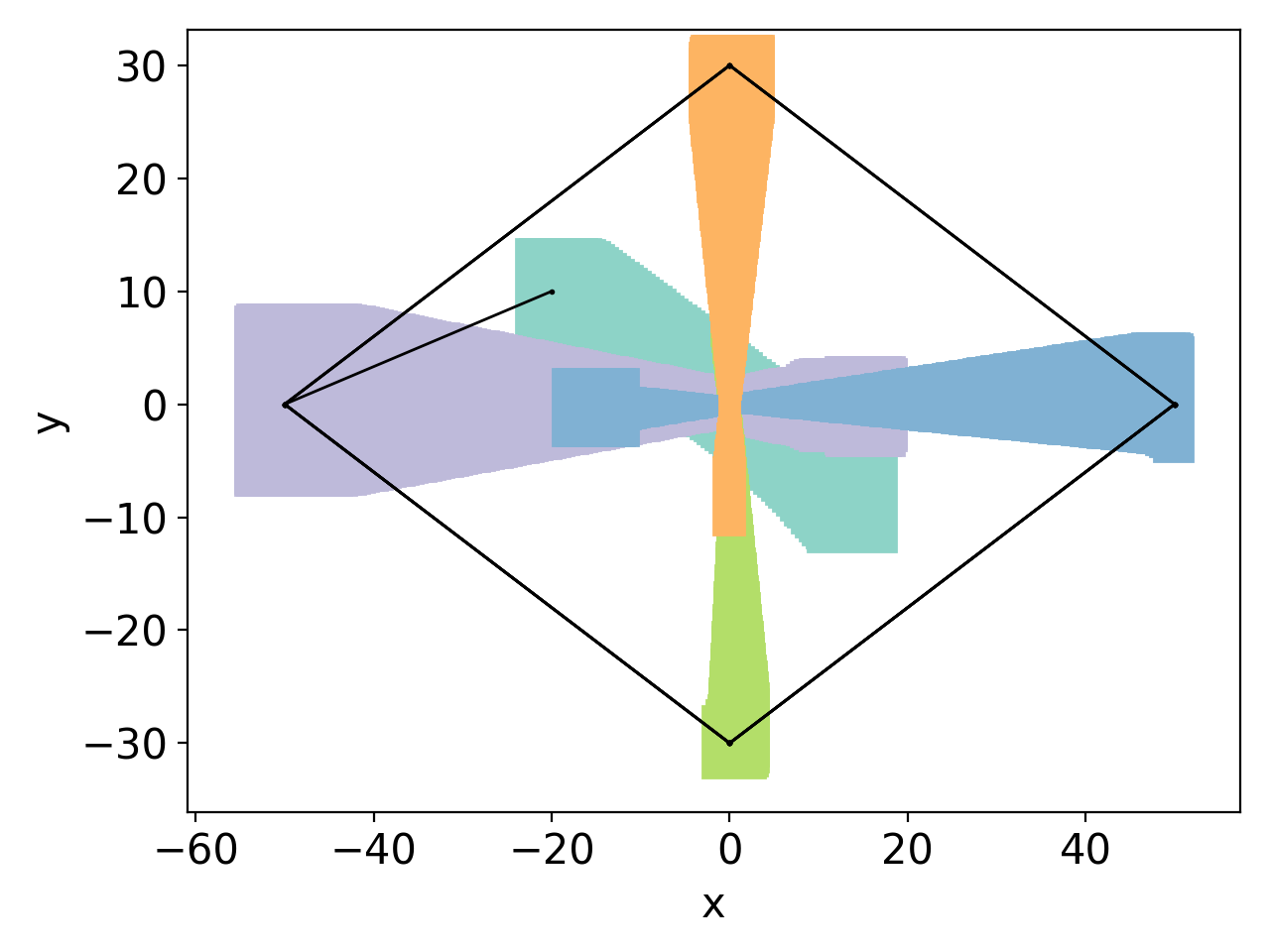}
	\caption{\label{fig:robot-rect-t-sym2-virtual}}
	\vspace{\floatsep}
\end{subfigure}
\hspace{-0.1in}
\begin{subfigure}[t]{0.25\textwidth}
	\centering
	\includegraphics[width=\textwidth]{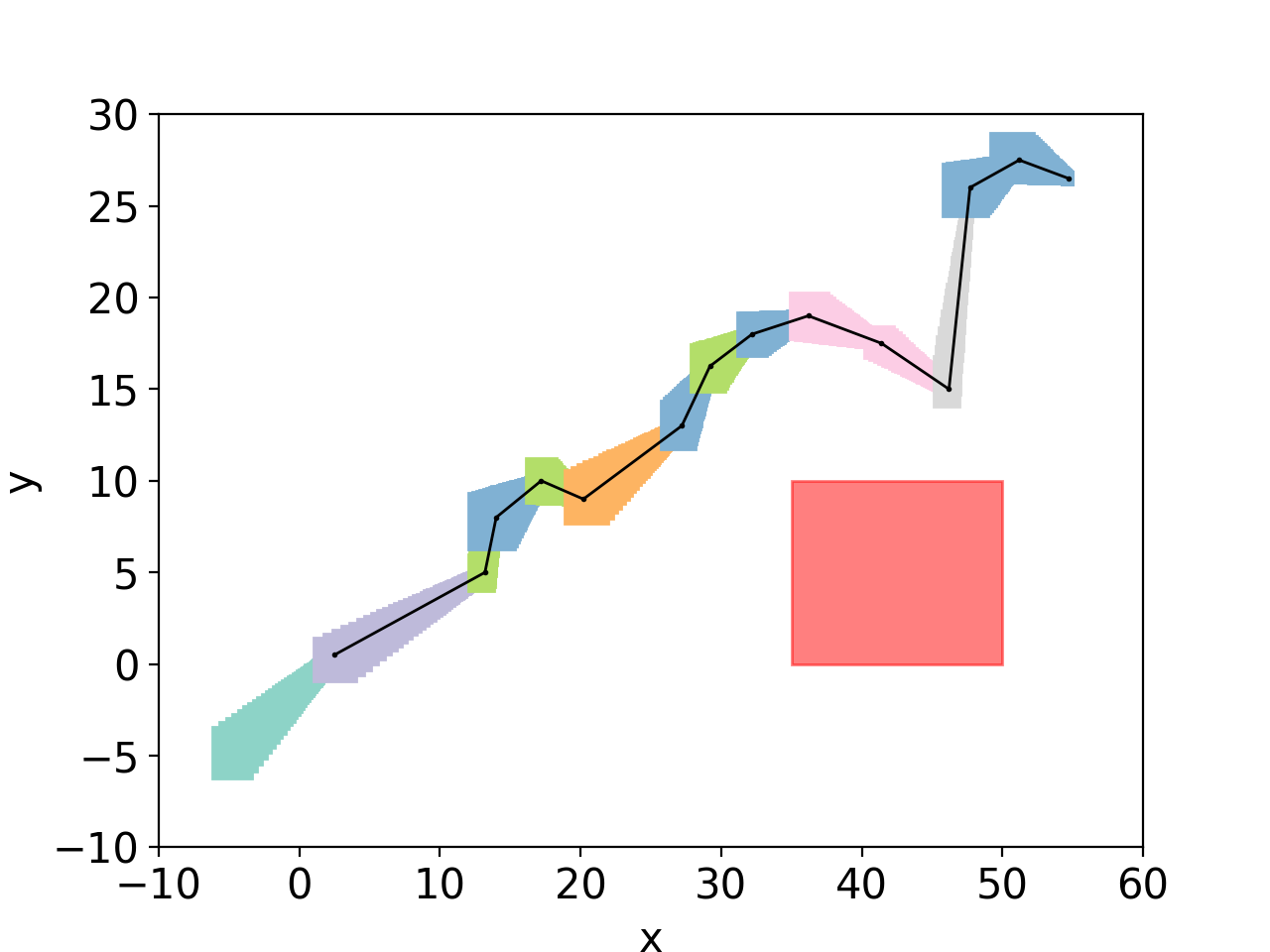}
	\caption{\label{fig:linear-random-t-r-sym2-real}}
	\vspace{\floatsep}
\end{subfigure}
\hspace{-0.1in}
\begin{subfigure}[t]{0.25\textwidth}
	\centering
	\includegraphics[width=\textwidth]{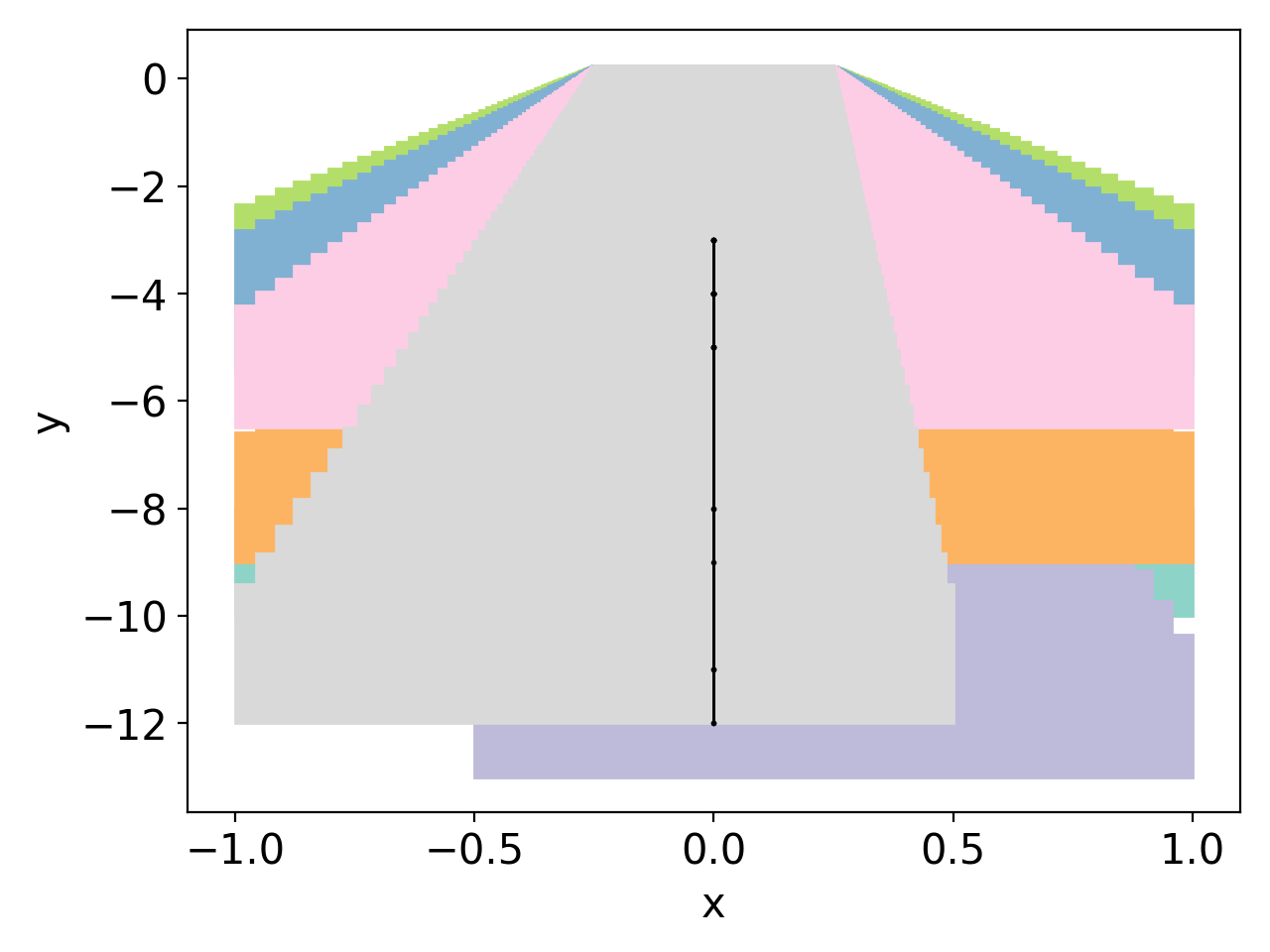}
	\caption{\label{fig:linear-random-t-r-sym2-virtual}}
	\vspace{\floatsep}
\end{subfigure}
\hspace{-0.1in}
	\begin{subfigure}[t]{0.25\textwidth}
		\centering
		\includegraphics[width=\textwidth]{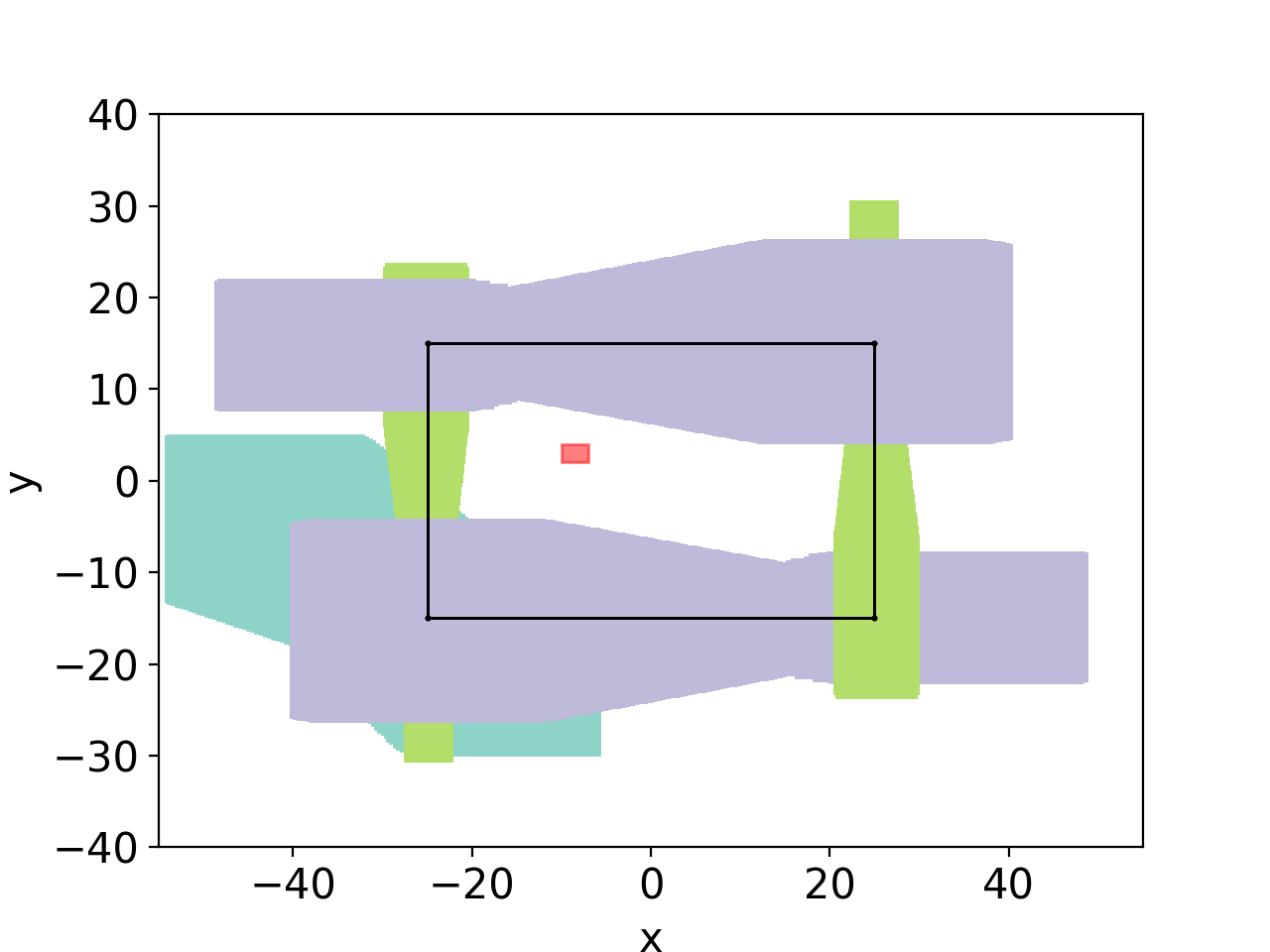}
		\caption{\label{fig:robot-rect-t-r-sym2-real}}
		\vspace{\floatsep}
	\end{subfigure}
\hspace{-0.1in}
	\begin{subfigure}[t]{0.25\textwidth}
		\centering
		\includegraphics[width=\textwidth]{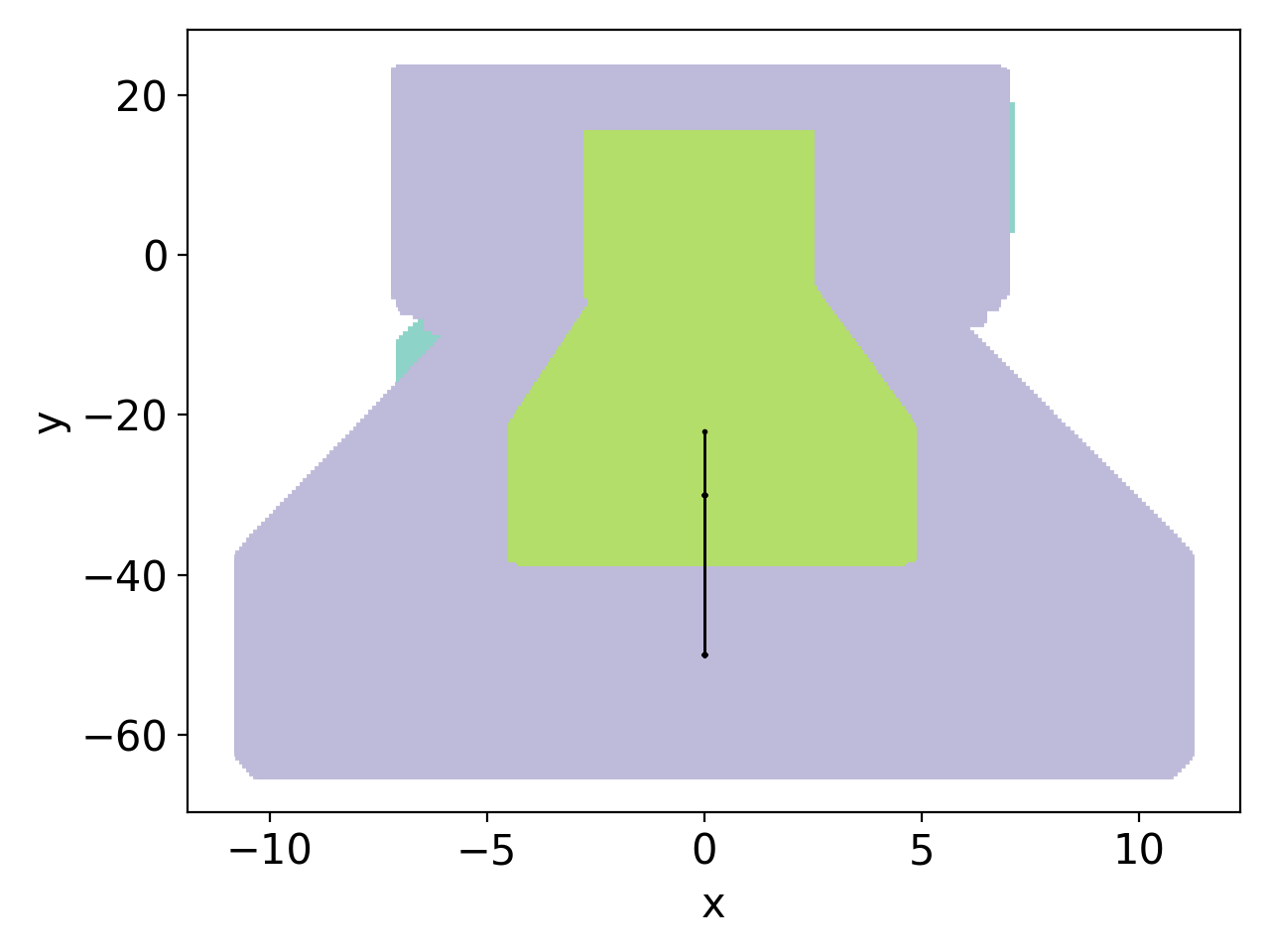}
		\caption{\label{fig:robot-rect-t-r-sym2-virtual}}
		\vspace{\floatsep}
	\end{subfigure}
\hspace{-0.1in}
\begin{subfigure}[t]{0.25\textwidth}
	\centering
	\includegraphics[width=\textwidth]{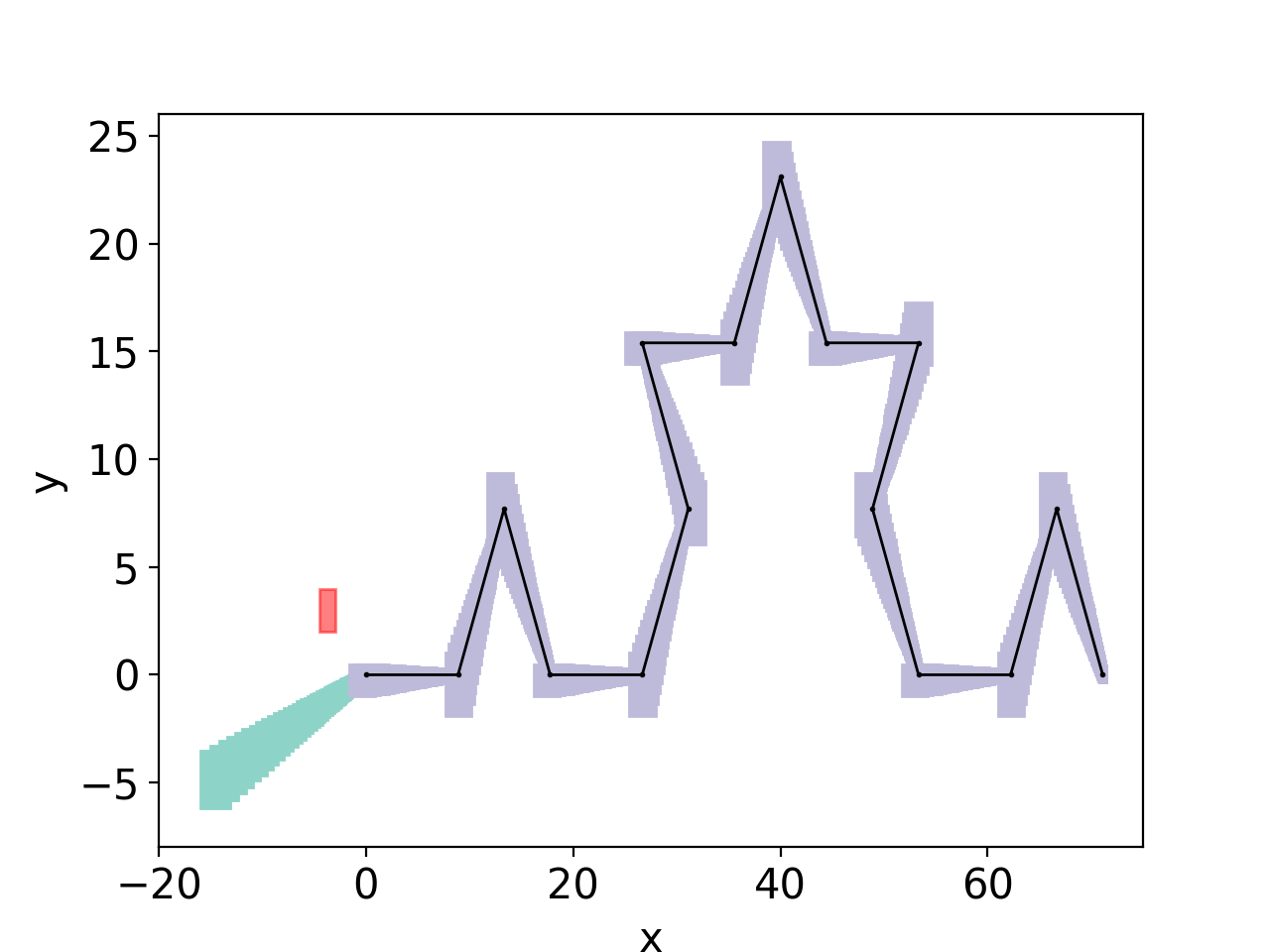}
	\caption{\label{fig:linear-snowflake-t-r-sym2-real}}
	\vspace{\floatsep}
\end{subfigure}
\hspace{-0.1in}
\begin{subfigure}[t]{0.25\textwidth}
	\centering
	\includegraphics[width=\textwidth]{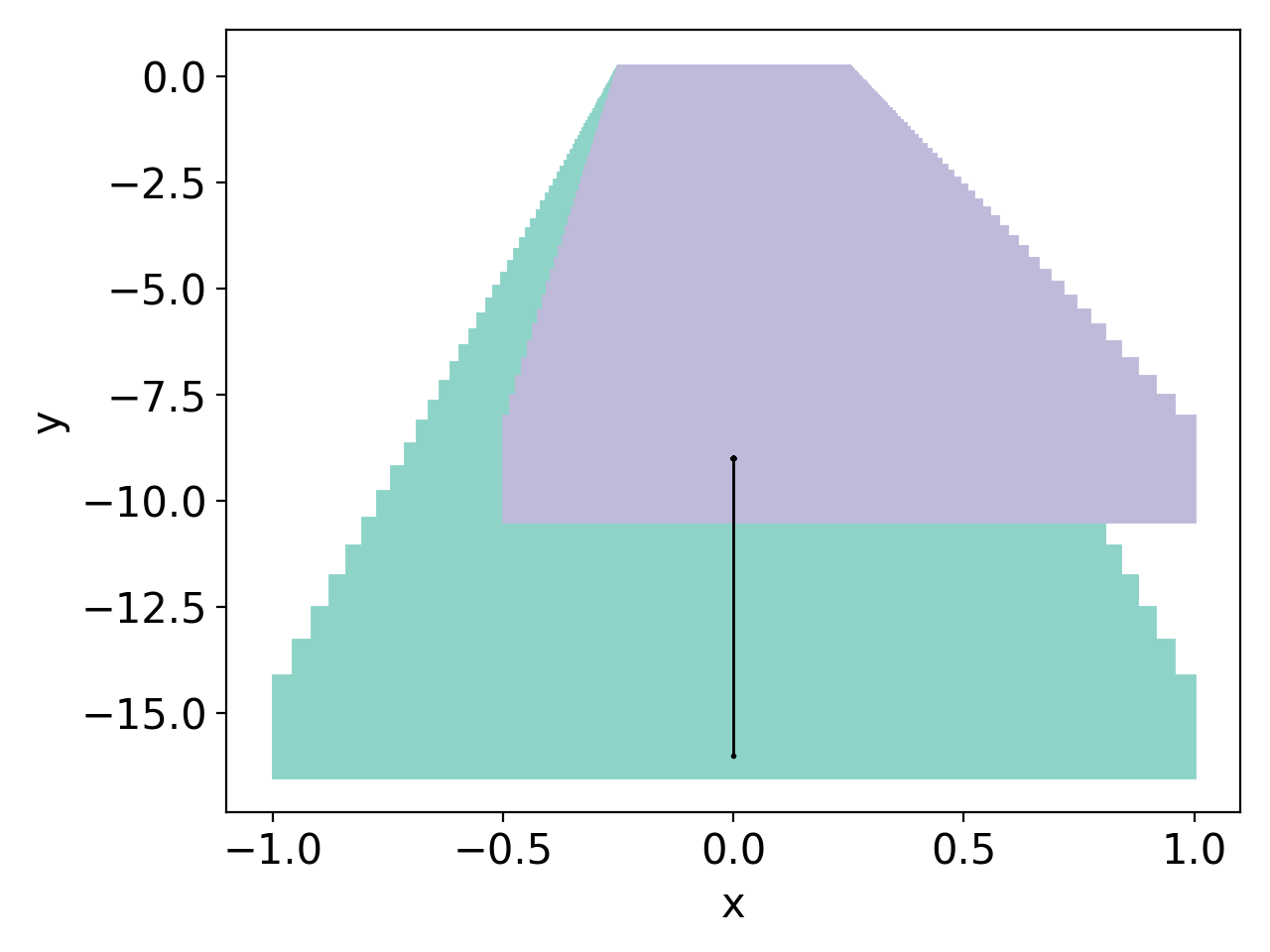}
	\caption{\label{fig:linear-snowflake-t-r-sym2-virtual}}
	\vspace{\floatsep}
\end{subfigure}
\hspace{-0.1in}
	\begin{subfigure}[t]{0.25\textwidth}
		\centering
	\includegraphics[width=\textwidth]{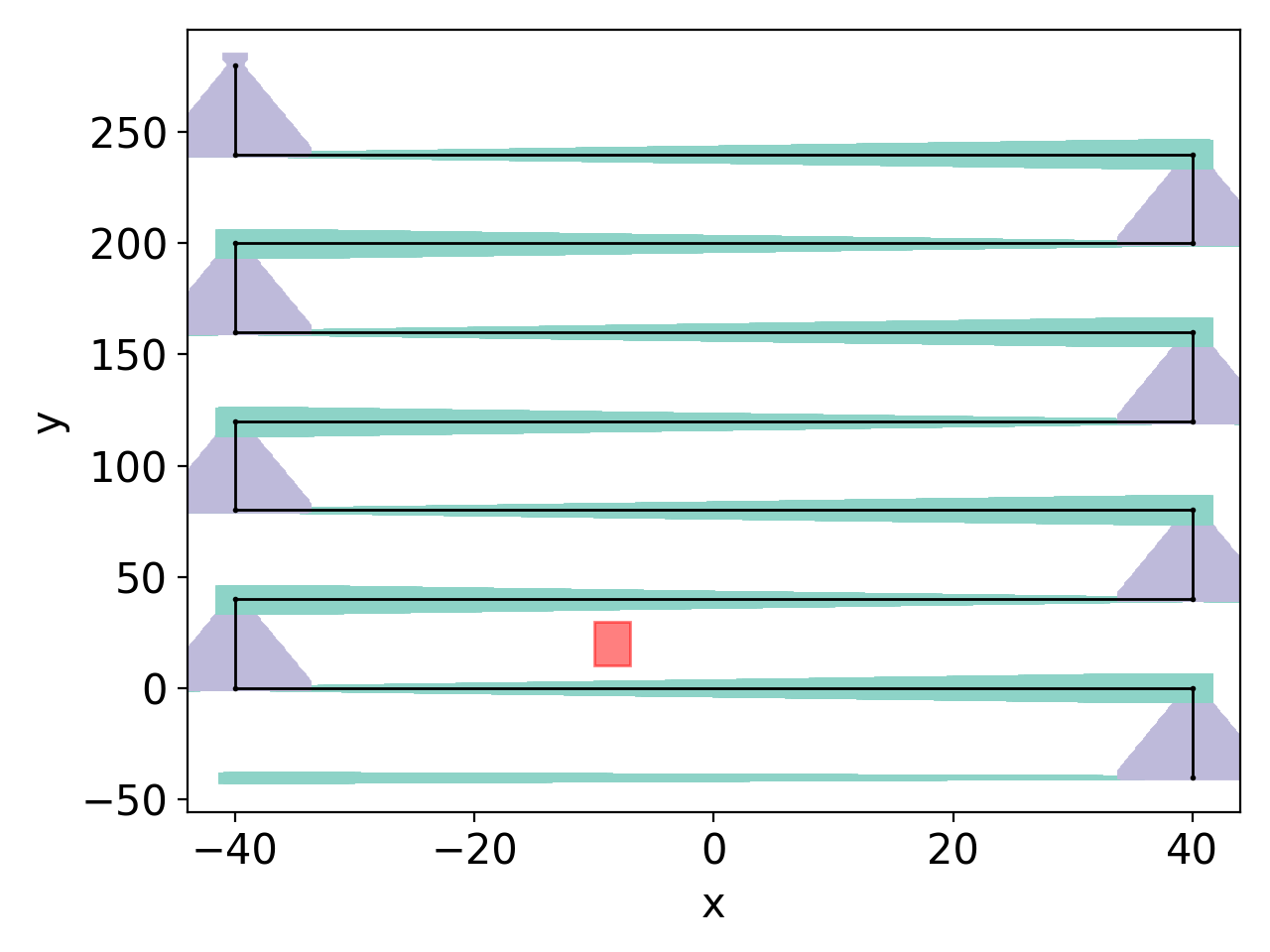}
	\caption{ \label{fig:robot-S-t-r-sym2-real}}
	\vspace{\floatsep}
\end{subfigure}
\hspace{-0.1in}
\begin{subfigure}[t]{0.25\textwidth}
	\centering
	\includegraphics[width=\textwidth]{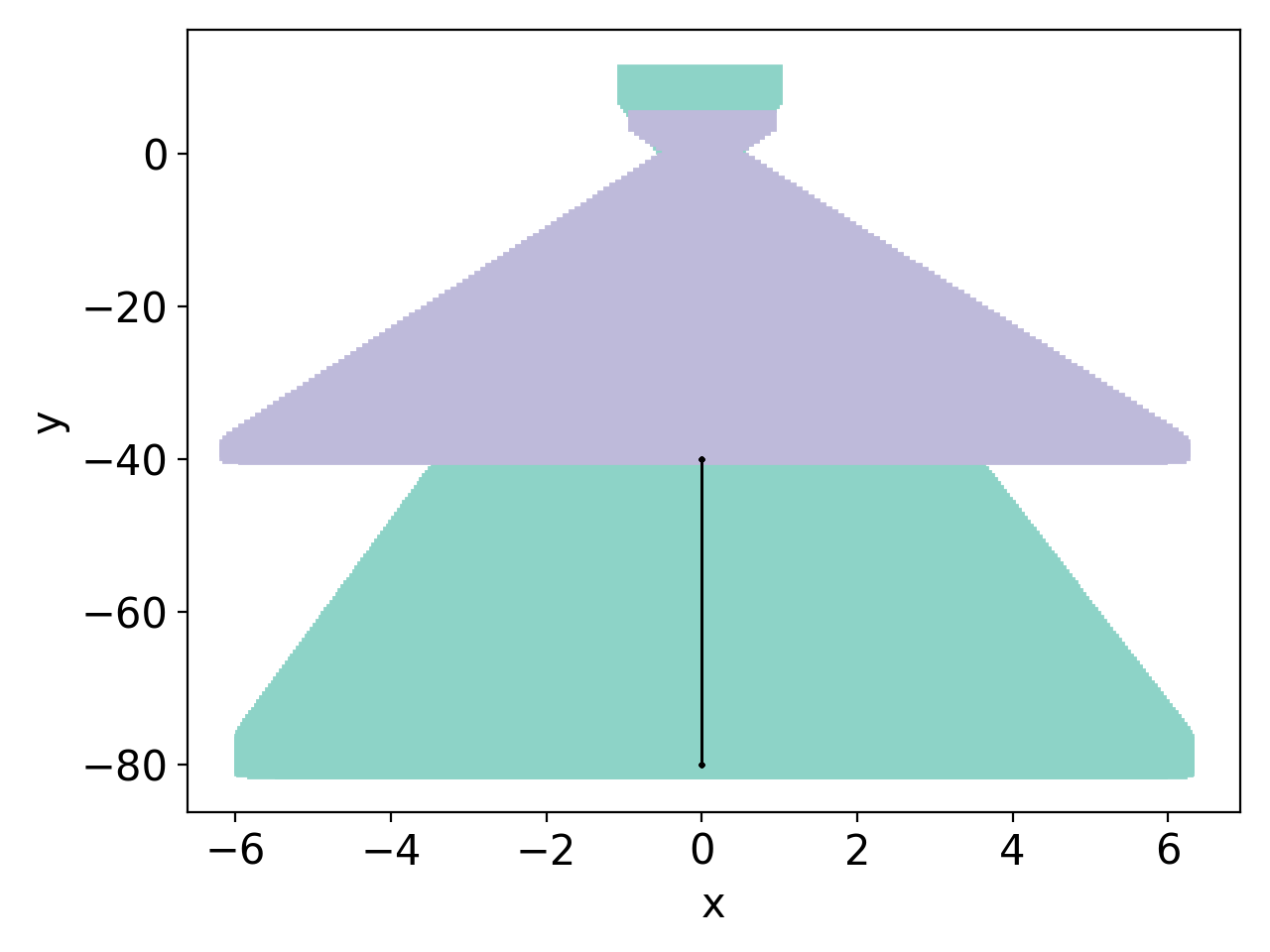}
	\caption{\label{fig:robot-S-t-r-sym2-virtual}}
	\vspace{\floatsep}
\end{subfigure}
\hspace{-0.1in}
\begin{subfigure}[t]{0.25\textwidth}
	\centering
	\includegraphics[width=\textwidth]{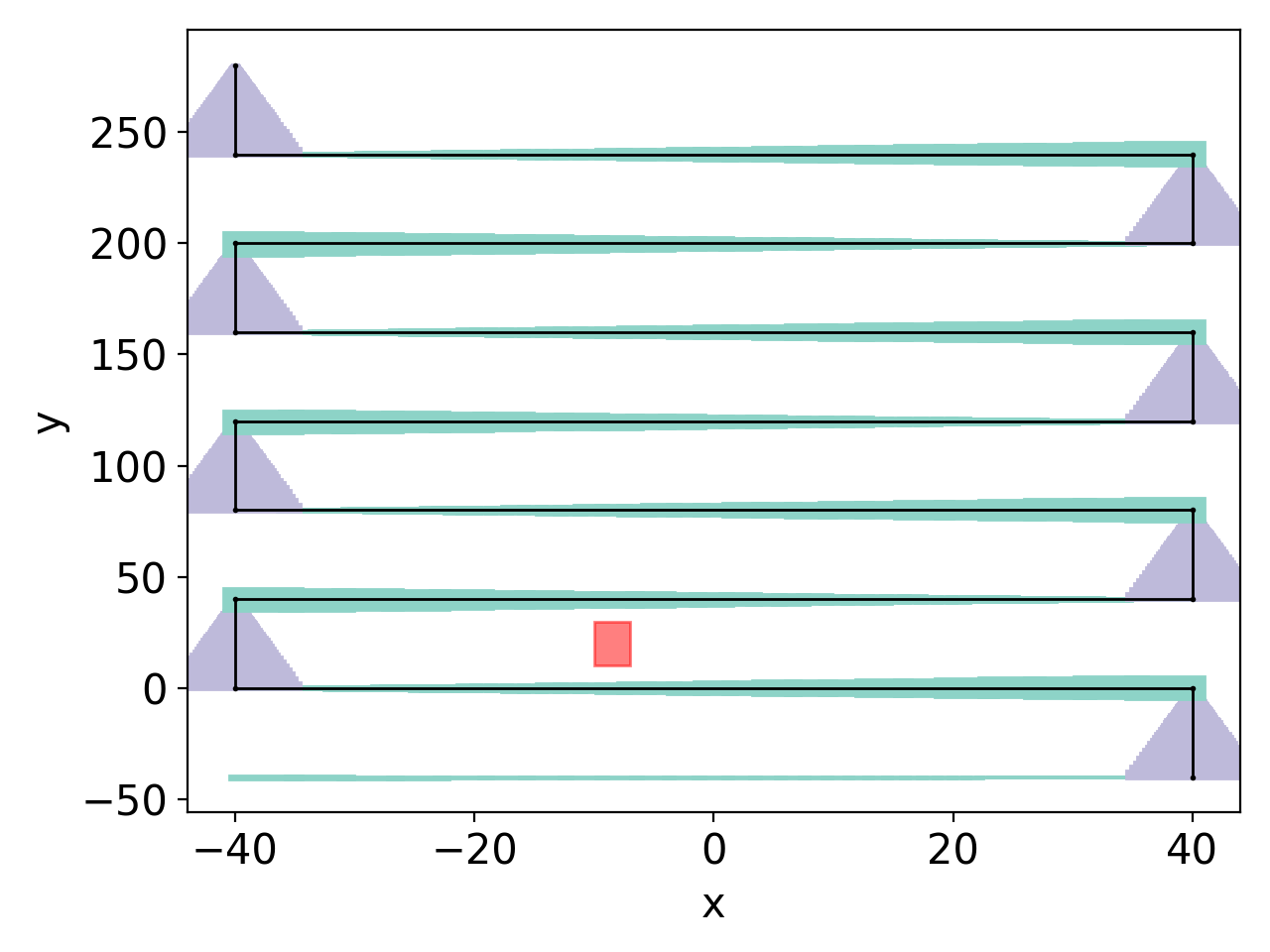}
	\caption{ \label{fig:linear-S-t-r-sym2-real}}
	\vspace{\floatsep}
\end{subfigure}
\hspace{-0.1in}
\begin{subfigure}[t]{0.25\textwidth}
	\centering
	\includegraphics[width=\textwidth]{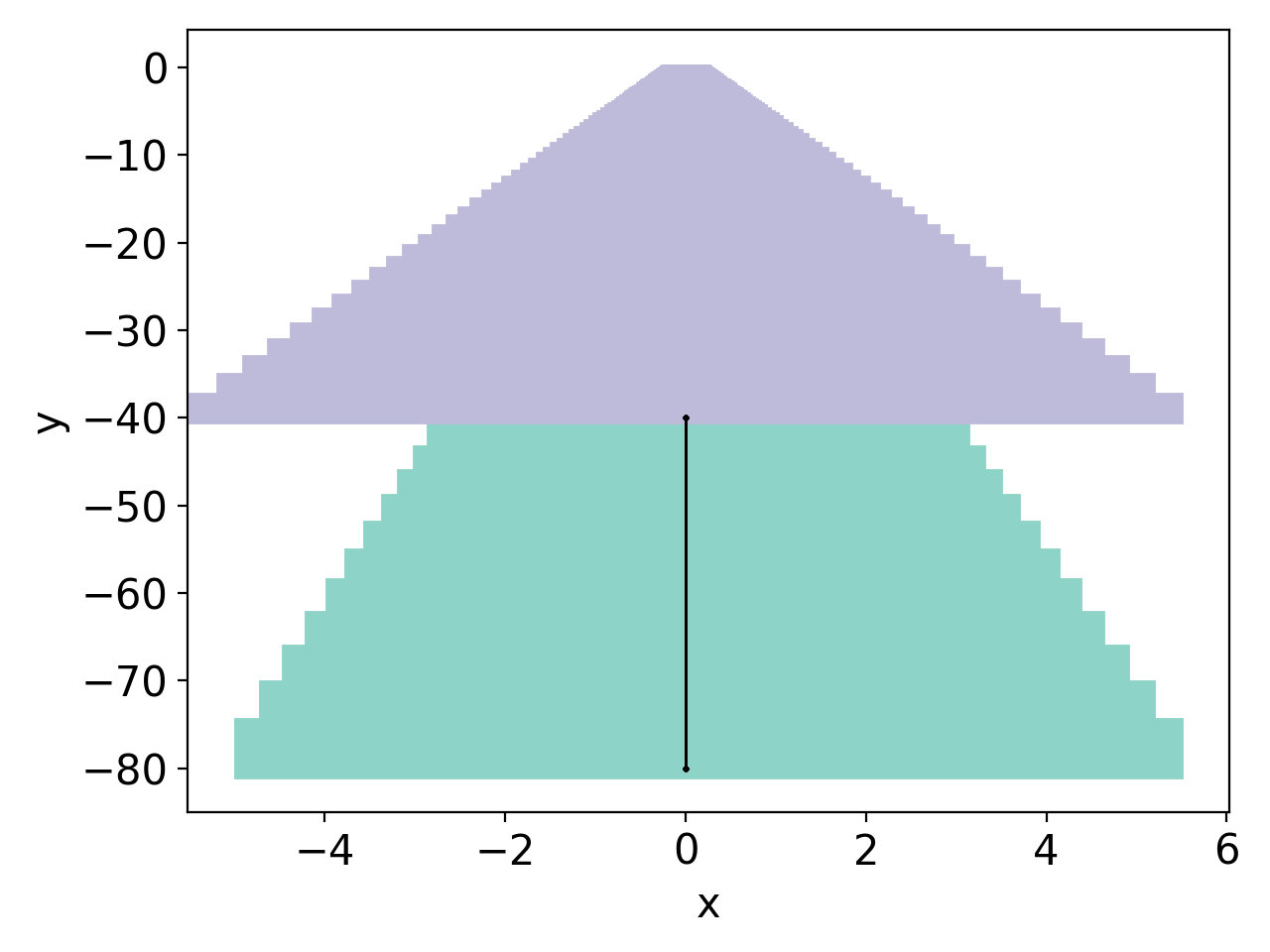}
	\caption{\label{fig:linear-S-t-r-sym2-virtual}}
	\vspace{\floatsep}
\end{subfigure}
\hspace{-0.1in}

\begin{comment}
\begin{subfigure}[t]{0.25\textwidth}
	\centering
	\includegraphics[width=\textwidth]{Figures/linear-snowflake-t-sym2/real.png}
	\caption{linear, T\label{fig:linear-snowflake-t-sym2-real}}
	\vspace{\floatsep}
\end{subfigure}
\hspace{-0.1in}
\begin{subfigure}[t]{0.25\textwidth}
	\centering
	\includegraphics[width=\textwidth]{Figures/linear-snowflake-t-sym2/virtual.png}
	\caption{linear, T, vir, Ko\label{fig:linear-snowflake-t-sym2-virtual}}
	\vspace{\floatsep}
\end{subfigure}
\end{comment}
	\caption{\scriptsize 
		% The concrete and virtual reachsets of different experiments. In all figures, unless otherwise stated, {\em sym} is 2. {\em vir.} refers to virtual automaton. {\em Ko} refers to the Koch snowflake path. When the path is clear from the figures, it is omitted from the caption. Finally,
		%we present Figures~\ref{fig:linear-kock-t-vs-tr} and \ref{fig:linear-random} to 
		%in the figures corresponding to the different experiments, we show the concrete and virtual systems along with their computed reachsets. Segments corresponding to different virtual modes are colored differently. The concrete automaton reachset segments in the figures have the same color as their corresponding reachset segment in the virtual system figures. Finally, the unsafe set is colored red. \hussein{add to the caption instead of here }
		Reachsets, projected to the position-in-the-plane part of the state, for vehicles visiting sequences of waypoints.
		Each horizontal pair depicts $\Reach_\ha$ (left) and $\Reach_{\ha_v}$ (right) of a scenario. {\em sym} is $\SV$ in all scenarios.
		{\em dynamic} is {\em robot} (left) and {\em linear} (right).
		$\Phi$ is $T$ in the $1^{\expst}$ row (Figures~\ref{fig:robot-rect-t-sym2-real} through \ref{fig:linear-random-t-r-sym2-virtual}), and  $\TR$ in the rest of the figures (Figures~\ref{fig:robot-rect-t-r-sym2-real} through \ref{fig:linear-S-t-r-sym2-virtual}). 
		Figure \ref{fig:robot-rect-t-sym2-virtual} has larger $\#m/e$ and smaller reachset segments than Figure~\ref{fig:robot-rect-t-r-sym2-virtual}, while modeling the same scenario, where {\em path} is {\em rectangle}, as shown in Figures~ \ref{fig:robot-rect-t-sym2-real} and~\ref{fig:robot-rect-t-r-sym2-real}. 
		That shows that smaller $\#m/e$ means larger reachset segments. Figure~\ref{fig:linear-random-t-r-sym2-virtual}, corresponding to {\em path} being {\em random} in Figure~\ref{fig:linear-random-t-r-sym2-real}, has larger $\#m/e$ that \ref{fig:linear-snowflake-t-r-sym2-virtual}, corresponding to {\em path} being  {\em snowflake} in Figure~\ref{fig:linear-snowflake-t-r-sym2-real}, respectively. That shows that when $\Phi$ can relate the trajectories of more modes, it results in smaller $\#m/e$.
		Figure \ref{fig:robot-S-t-r-sym2-virtual} with {\em dynamic} being {\em robot} has larger reachset segments than \ref{fig:linear-S-t-r-sym2-virtual} with {\em dynamic} being {\em linear}, while corresponding to the same $S$-shaped path. That shows that more complex dynamics results in larger reachset segments.
		%, and thus more expensive computations. 
		%Reachset with symmetry (right) is computed in virtual coordinates. 
		% The last two figures in the top row illustrate how using both translational and rotational symmetries, a single reachability computation can cover all the segments of the path.
 }
	%} %\ref{fig:robot-rect-t-r-sym0-real} show the reachset of $\ha$ when {\em sym} equals 0, \ref{fig:robot-rect-t-r-sym2-real} and \ref{fig:robot-rect-t-r-sym2-virtual} show the same but when {\em sym} equals 2. \label{fig:robot-rect}}
\end{figure*}

\subsection{Results analysis and discussion}
%In this section, we list and discuss few observations from the results shown in Table~\ref{table:experimental_results}.
We will discuss several observations using the results of some experiments. To check the same observations for different scenarios, we refer the reader to Table~\ref{table:experimental_results} in the Appendix.

\subsubsection*{\SV\ vs. \NS\ and \SC}
\begin{comment}
The results of the experiment on the scenario of Example~\ref{sec:virtual_example_roads} are shown in the first, fourth, and $5^{\expth}$ rows of Table~\ref{table:robot_example}. 
The number of computed reachsets \#{\em co} is larger when using the virtual automaton versus using the method in \cite{Sibai:TACAS2020} or not using symmetry. 
This is because of the conservativeness or the over-approximation error added by using the virtual automaton, as can be seen the {\em error} column. Moreover, this is a path where a fixed point can be reached without using symmetry or a virtual system since its reachset would be a bounded set. However, for the $S$-shaped path or the snowflake paths, where the reachset can grow unbounded, reaching a fixed point would be infeasible without abstraction. 
Moreover, with the right choice of virtual map, one can improve the computation time as well as the added over-approximation error, even for the bounded case, as can be seen in the last row of the table.
\end{comment}
The results of running our implementation on a scenario using all combinations of the three methods with the two virtual maps,
% where {\em dynamic} and {\em path} are {\em robot} and {\em S-shaped}, respectively,
 are shown in Table~\ref{table:robot_S_example}. 
When using the \SV, the numbers of computed reachsets \#{\em co} are 1749 and 1513, for $\Phi$ being $T$ and $\TR$, respectively. These are significantly smaller than those when using the other methods in the first three rows.
The reason is that 11 out of the 16 segments were transformed from $\permodedict $, when $\Phi$ was $T$, and 13 out of 16 when $\Phi$ was $\TR$. Hence, only five segments needed to be computed using the griding method, before a fixed point was reached when $\Phi$ was $T$, and three when $\Phi$ was $\TR$. This resulted in significant decrease in the total number of reachsets requested (\#{\em tot.}). Consequently, the computation time is around 89\% and 56\% less than that of \NS\ and \SC, respectively, when $\Phi$ was $T$, and by 89\% and 83\%, when $\Phi$ was $\TR$. We can see that \SC, that uses symmetry and caching, maintained advantage in computation time over \NS, the standard one, as in \cite{Sibai:TACAS2020}.

When using \SV, the over-approximation errors were 23\% and 140.8\%, for $\Phi$ being $T$ and $\TR$, respectively. These are larger than those of \SC, which were 0\% and 49.4\%, respectively.
This is because of the conservativeness of the abstraction as we discussed at the end of Section~\ref{sec:fsr}. 
%Moreover, this is a path where a fixed point can be reached without using symmetry or a virtual system since its reachset would be a bounded set. 
% However, for the $S$-shaped path or the snowflake paths, where the reachset can grow unbounded, reaching a fixed point would be infeasible without abstraction. 

% Moreover, with the right choice of virtual map, one can improve the computation time as well as the added over-approximation error, even for the bounded case, as can be seen in the last row of the table.

\begin{table}[h]
	\vspace{0.1in}
	\caption{ \scriptsize Results showing the advantage of using \SV\ over \NS\ and \SC\ to compute $\Reach_\ha$ for the scenario where {\em dynamic} and {\em path} are {\em robot} and {\em S-shaped}, respectively. \#$m / e$ is  3/4 and  2/2 for the $4^{\mathit{th}}$ and $5^{\mathit{th}}$ rows, respectively.
		\label{table:robot_S_example} }
	\centering
	\begin{tabular}{|l |l| l | l| l | l |l|l|}
		\hline 
		%\multicolumn{1}{|l|}{Sym}    
		$\Phi$ & sym 
		%& \multicolumn{}{|l|}{Results}  \\
		%\hline
		& \#co & \#re &  \#cp & \#tot. & time & error \\
		\hline 
		 - &\NS &18591  &- &-&18591 &1.3 &- \\
		\hline 
		 T &\SC &6831  &11760 & - &18591  &0.8 & 0\\
		\hline 
		 TR &\SC &2421 &3864 &- &6285& 0.32 & 49.3 \\
		\hline
		T &\SV &1749  &84 &11 &1844 & 0.14 &23.4 \\
		\hline 
		TR&\SV &1513&0&13&1526 &0.14 &140.8 \\
		\hline 
	\end{tabular}
\end{table}
\begin{comment}
	\begin{tabular}{|l |l| l | l| l | l |l|l|}
\hline 
%\multicolumn{1}{|l|}{Sym}    
$\Phi$ & sym 
%& \multicolumn{}{|l|}{Results}  \\
%\hline
& \#co & \#re &  \#cp & \#tot. & time & error \\
\hline
- & 1 &8001 &-& -& 8001 &0.6 &  - \\
\hline 
T &2 &3945 &4056 &- &8001 &0.42 & 0 \\
\hline 
T &3 &9306 &0 &11 &9317 &0.84 & 326\\
\hline 
TR&2 &4206 &2008 & - &6214  &0.43 &  $-30$ \\
\hline 
TR&3 &30456 &273 &13&30742 &3.42 & 2142\\
\hline 
\end{tabular}
\end{comment}

\subsubsection*{Smaller vs. larger virtual automata}

%We show first how to interpret the results for the different experiments before discussing what properties of the concrete hybrid automaton, modeling such waypoint following scenarios, affects the abstraction and verification processes in the rest of the section.  

In general, the fewer the modes and edges \#$m/e$ of $\ha_v$,  the less is the number of reachset segments have to be computed before reaching a fixed point. This can be seen, for example, by comparing the $4^{\expth}$ and $5^{\expth}$ rows of Table~\ref{table:robot_S_example}. In the $4^{\expth}$ row, where $\Phi$ was $T$, \#$m/e$ were $3/4$. In the $5^{\expth}$ row, where $\Phi$ was $\TR$, \#$m/e$ were $2/2$. 
Out of the 16 reachset segments had to be computed for the 16 roads in the $S$-shaped path, in the $4^{\expth}$ row, 11 segments were transformed from $\permodedict $, while in the $5^{\expth}$ row, 13 segments. That means that the fixed point was reached earlier when \#$m/e$ of $\ha_v$ were smaller.

On the other hand, smaller \#$m/e$ means larger equivalent sets of modes and edges in $\ha$. Hence, more modes of $\ha$ would be mapped to the same mode of $\ha_v$. That means each reachset segment in $\permodedict $ should cover more cases. That leads to larger reachset segments and more conservativeness of the abstraction. This can be seen by comparing the over-approximation error in the $4^{\expth}$ and $5^{\expth}$ rows. In the $5^{\expth}$ row, it had $140.8$\% error versus $23.4$\% of the $4^{\expth}$ row.  It can also be seen for the rectangle path with robot dynamics scenario by checking  Figures~\ref{fig:robot-rect-t-sym2-real} and ~\ref{fig:robot-rect-t-r-sym2-real}.
% We compare the sizes of the reachset segments corresponding to the same roads in the two figures. 
Figure~\ref{fig:robot-rect-t-r-sym2-real} has smaller \#$m/e$, but larger reachset than Figures~\ref{fig:robot-rect-t-sym2-real}. 

Additionally, the larger \#$m/e$, the more expensive the fixed point check of equation~(\ref{eq:fixed_point_condition}) is. This can be seen by again comparing the $4^{\expth}$ and $5^{\expth}$ rows. The $4^{\expth}$ row had larger \#{\em co} and \#{\em tot.} than the $5^{\expth}$ row, and the latter reached fixed point earlier. Both of these reasons should have lead the $5^{\expth}$ row to have faster computation time if there was no fixed-point check after each segment computation. However, they were the same.
%Those in Figure~\ref{fig:linear-Koch-t-r-2d-real} are much larger than those in Figure~\ref{fig:linear-Koch-t-2d-real}. \hussein{mention the error, it is much more accurate.}

Larger reachsets are also more expensive to compute since, as mentioned earlier, they are computed as a union of reachsets with smaller initial sets of fixed size, the grid cells size. Thus, more cells reachsets are computed because of larger reachsets causing larger per-mode initial sets.
%For example, the number of cell reachsets computed from scratch \# {\em co} is 6792 in $3^{\mathit{rd}}$ row, much larger than $3219$, that of the $5^{\mathit{th}}$ row. 
%It is also worth noting that for the $3^{\mathit{rd}}$ row, the computed reachsets are for the first three segments only, since the rest of 13 segments got transformed from $\permodedict$ as can be seen in the \# {\em cp} column.
This overhead would lead to larger computation time for some of the scenarios with smaller $\ha_v$ than those with larger ones, despite reaching the fixed point earlier. For example, the computation time for the $3^{\exprd}$ row was larger than that of the $2^{\expnd}$ row in Table~\ref{table:robot_example}, although the former had smaller $\ha_v$ and reached fixed point earlier.
\begin{table}[h]
	\vspace{0.1in}
	\caption{ \scriptsize Results showing that smaller $\# m/e$ might cause more computations and larger reachsets than one with larger $\# m/e$. Scenario has {\em dynamic} and {\em path} being {\em robot} and {\em rectangle}, respectively. \NS\ is used in the $1^{\expst}$ row  and \SV\ in the $2^{\expnd}$ and $3^{\exprd}$ rows.
		\label{table:robot_example} }
	\centering
	\begin{tabular}{|l |l| l | l| l | l |l|l|}
	\hline 
	%\multicolumn{1}{|l|}{Sym}    
	$\Phi$ & \#$m/e$ 
	%& \multicolumn{}{|l|}{Results}  \\
	%\hline
	& \#co & \#re &  \#cp & \#tot. & time & error \\
	\hline
	- & - &8001 &-& -& 8001 &0.6 &  - \\
	\hline 
	T & 5/5 &9306 &0 &11 &9317 &0.84 & 326\\
	\hline 
	TR&3/3 &30456 &273 &13&30742 &3.42 & 2142\\
	\hline 
\end{tabular}
\end{table}

%This was only overcame in the nonlinear fixed-wing scenarios in the last two rows.

% \vspace{-0.1in}
\subsubsection*{Effective vs. less effective symmetries}
The more effective $\Phi$ is in grouping different modes of $\ha$, the smaller $\ha_v$ is. In other words, 
the coarser the partition of $P$ that $\Phi$ creates, the smaller are $\#m/e$. This can be seen in our experiments by comparing $\#m/e$ for $\Phi$ being $T$ or $\TR$ for the scenarios in Tables~\ref{table:robot_S_example} and \ref{table:robot_example}. For example, in Table~\ref{table:robot_example}, \#$m/e = 5/5$ when $\Phi$ is $T$, and \#$m/e = 3/3$, when $\Phi$ being $\TR$. This reflects that $\TR$ is able to group more modes than $T$. This is expected since $\TR$ relates translated and rotated trajectories of the robot, not just translated ones.
%we show the corresponding graphs and reachsets in Figure~\ref{fig:linear-kock-t-vs-tr}. When using translation only, the virtual system has 6 modes and 7 edges (Figure~\ref{fig:linear-Koch-t-2d-virtual}) while when using rotation as well, it has only 2 modes and 2 edges (Figure~\ref{fig:linear-Koch-t-r-2d-virtual}). 
%In both cases, one of the modes is $\bot$ corresponding to the initial set.
% \vspace{-0.1in}
%\subsubsection{Structured vs. less structured paths}
%The more structured the set of edges with respect to the virtual map is, the less are the numbers of modes and edges of the virtual one. 
Another point of view can be seen by examining the scenarios considered in Table~\ref{table:structured_vs_unstructured}. Although the total number of segments for the Koch-snowflake path is 16, larger than that of the random path which is 14, $\#m/e$ of the latter were 7/11, much larger than those of the former 2/2. This can also be seen by comparing Figures~\ref{fig:linear-snowflake-t-r-sym2-virtual} and and~\ref{fig:linear-random-t-r-sym2-virtual}. That means that $\TR$ was not able to group the different modes of the random path, due to the different lengths of the roads. On the other hand, it was able to group all the modes of the snowflake path in two modes, since all the roads, except the first, are translated and rotated versions of each other.

\begin{table}[h]
	\vspace{0.1in}
	\caption{\scriptsize Results showing that $\Phi = \TR$ results in smaller $\ha_v$, i.e. smaller $\#m/e$, than $\Phi = T$. Dynamics are linear and {\em sym} is \SV.
		\label{table:structured_vs_unstructured} }
	\centering
	\begin{tabular}{|l|l | l | l| l | l |l|l|}
		\hline 
		%\multicolumn{1}{|l|}{Sym}    
		path & $\Phi$ & \#$m/e$ & 
		%& \multicolumn{}{|l|}{Results}  \\
		%\hline
		\#co & \#re &  \#cp & \#tot. & time  \\
		\hline
		rand. & T & 12/13 & 321 &93& 0& 414 &0.29  \\
		\hline
		rand. & TR & 7/11 & 228.6 &224.4& 2& 455 &0.33  \\
		\hline 
		Ko.& T & 6/8 & 392 &208& 4 &604  &0.39  \\
		\hline 
		Ko.& TR & 2/2 & 180 &0& 2 &194  &0.17  \\
		\hline 
	\end{tabular}
\end{table}

\subsubsection*{Simple vs. complex continuous dynamics}
The more complex the continuous dynamics are, the more the cells reachsets have to be computed. This is shown by comparing \#{\em tot.} of the linear and robot dynamics in Table~\ref{table:linear_vs_nonlinear}. They have the same sizes of initial sets, same $\#m/e$, and they reach the fixed point after the same number of reachset segments. Yet, the robot had \#{\em tot.} of 1844 versus 654 for the linear dynamics. Also, this can be seen by comparing the sizes of the reachsets of Figures~ \ref{fig:robot-S-t-r-sym2-virtual} and \ref{fig:linear-S-t-r-sym2-virtual}. The green reachset segment of the robot has an initial set width from -6 to 6, while that of the linear model ranges from around -5 to 5. That is because the robot is less stable than the linear one.

\begin{table}[h]
	\vspace{0.1in}
	\caption{\scriptsize Results showing that stable linear dynamics results in fewer reachsets need to be computed, than the more complex robot dynamics. {\em path} is {\em S-shaped}.
		\label{table:linear_vs_nonlinear} }
	\centering
	\begin{tabular}{|l | l| l | l| l | l |l|l|}
		\hline 
		%\multicolumn{1}{|l|}{Sym}    
		dyn. & $\Phi$ &  
		%& \multicolumn{}{|l|}{Results}  \\
		%\hline
		 \# co & \# re &  \# cp & \# tot. & time & error \\
		\hline
		robot & T & 1749 &84 & 11& 1844 &0.14 &  23.4 \\
		\hline
		robot & TR &1513 &0 & 13& 1526 &0.14 &  140.8 \\
		\hline 
		linear & T &603 &40& 11&654  &0.46 &  7.3 \\
		\hline 
		linear & TR &601 &36& 13 &650  &0.45 &  103.8 \\
		\hline 
	\end{tabular}
\end{table}

\begin{comment}
\begin{figure}
	%\caption{Reachtubes for drone~\ref{fig:linear_nosym_1_agents_sym} and linear~\ref{fig:linear_nosym_3_agents_sym} models  using {\sf Sym-Flow*}. Three agents \label{fig:flow_linear}}
\caption{Scenario 6 (left) and Scenario 8 (right), concrete systems (top) and virtual systems (bottom), linear vs. nonlinear \label{fig:S-linear-vs-nonlinear}}
\end{figure}
\end{comment}

% In summary, we decompose the multi-agent hybrid system so that the different agents and their different modes share their reachtube computations. 
\begin{comment}
 The section is organized as follows: we describe the multi-agent verification algorithm $\multiagentverif$ that we designed and implemented in $\ourtacastool$ in Section~\ref{sec:multi-verif-algorithm}, describe the $\symcompute$ implementation in Section~\ref{sec:cacheimplementation}, 
 %describe of the dynamic safety algorithm along with the pseudo code in Section~\ref{sec:dynamicsafety}, 
 and finish with the results of the experiments and corresponding analysis in Section~\ref{sec:results}.
\end{comment}

%\subsubsection{Transforming vs. computing reachset}

\section{Conclusion}
\label{sec:conclusion}
We presented the first symmetry-based abstractions of hybrid automata. Our abstractions create automata with fewer number of modes and edges than the concrete ones by representing sets of modes with single ones. Symmetry maps transform trajectories of a concrete mode to trajectories of its corresponding abstract one, and vice versa.
We showed a forward simulation relation that proves the soundness of our abstraction.
Moreover, we showed how these abstractions would accelerate reachset computation and enable unbounded-time safety verification. The fewer number of modes of the abstract automaton makes its reachset computation reach a fixed point earlier than that of the concrete one. 
%This method utilizes the structure of the mode transition graph along with the symmetries to build a simpler system with few number of modes and edges. All of these modes share the same continuous dynamics, which means they can share their reachsets as well. Under nice conditions, this abstracted model would be able to reach a fixed point even if the concrete one may not.
Such a fixed point would result on a per-mode reachsets that can be transformed to construct the reachset of any concrete mode. 
We implemented our approach in Python 3 and showed the advantage of our approach over existing methods and the different parameters that affect our abstraction quality, in a sequence of reachability analysis experiments.

\begin{comment}
\subsection{Symmetry transformations and effectiveness}
\label{sec:symeffectiveness}

For Theorem~\ref{thm:fixed_point_condition} to be useful, i.e. for $(\ref{eq:fixed_point_condition})$ to be satisfied, several conditions must be satisfied:
\begin{enumerate}
	\item The number of virtual parameters $N_v$, which is the number of reachsets we compute, is small, that is the role of the $\rho_p$s. Remember that $N_v = |P_{D_v}|$, where $P_{D_v} = \{\rho_{p_{2}}(p_{1})$, where $p_{1}$ and $p_{2}$ are consecutive pairs in an execution of the real model$\}$. Hence, the paths of the real model and the symmetric transformations of the parameters should be structured enough that there are not many unique virtual parameters.
	\item For any virtual parameter $i$, the guards of all transitions to $i$ should not be spread out, but rather overlapping, i.e. $\guard_v(p_v^{(j)}, p_v^{(i)})$
\end{enumerate}
In the theorem above, we assume that $R$

The set of symmetry transformations $\Gamma_v$ should:

\begin{enumerate}
	\item Decrease the number of relative virtual parameters, i.e. decrease the cardinality of the set $P_{\disctrans_v}$.
	\item Decrease the size of the sets $\guardsingle_v(p_v)$ and $\resetsingle_v(p_v)$ for all $p_v \in P$.
	\item stable system results, how many reachtube segments we need to compute till we don't need to compute reachtubes anymore.
	\item levels of symmetry, edges or nodes levels.
\end{enumerate}
\end{comment}

 \bibliographystyle{IEEEtran}
\bibliography{sayan1,hussein}

\appendices

\begin{table*}
	\centering
	\small
	\caption{\small Experimental results. 
		\label{table:experimental_results}}
	\begin{tabular}{|l | l | l | l | l | l | l| l| l | l |l|l|l|}
		\hline 
		\multicolumn{5}{|l|}{Scenarios}    & \multicolumn{8}{|l|}{Results}  \\
		\hline
		id & model & path & virtual map & sym & \#co & \#ret. &  \#cp  & \#tot. & $\#m/e$ & \# pre-f & time & error \\
		\hline
		1 & robot & rectangle &  - & \NS &8001 &-& -& 8001 &- &- &0.6 &  - \\
		\hline 
		2 & robot & rectangle & trans& \SC &3945 &4056& - &8001 &- &- & 0.42&  0 \\
		\hline 
		3 & robot & rectangle & trans.+rot.& \SC &4206 &2008& - &6214 &- &- &0.43 &  -30 \\
		\hline 
		4 & robot & rectangle & trans.&\SV &9306 &0& 11 &9317 &5/5 &5/16 &0.84 & 326  \\
		\hline 
		5 & robot & rectangle & trans.+rot. &\SV &30456& 273 &13 &30742 &3/3 &3/16 &3.42 & 2142\\
		\hline 
		6 & linear & snowflake & - &\NS &835&-&-&835 &- &- &0.52 & -  \\
		\hline 
		7 & linear & snowflake & trans.&\SC &475&388& - &863 &- &- &0.52 &  3.7 \\
		\hline 
		8 & linear & snowflake & trans.+rot.&\SC &234&732& - &966 &- &- &0.5 & 93.5  \\
		\hline 
		9 & linear & snowflake & trans.&\SV &392&208 &4& 604 &6/8 &12/16 &0.39 & 3.2  \\
		\hline 
		10 & linear & snowflake & trans.+rot.&\SV &180&0&14&194 &2/2 &2/16 &0.17 & 137.4 \\
		\hline  
		11 & linear & random & - &\NS &351& - & - &351 &- &- &0.23 & - \\
		\hline 
		12 & linear & random & trans.&\SC &321& 93 &- &414 &- &- &0.37 &19  \\
		\hline 
		13 & linear & random & trans.+rot.&\SC &186.3&260.7& - &447  &- &- &0.28 & 118.7 \\
		\hline 
		14 & linear & random & trans.&\SV &321 &93&0&414 &12/13 &14/14 &0.29 & 19 \\
		\hline 
		15 & linear & random & trans.+rot.&\SV &228.6&224.4&2&455 &7/11 &12/14 &0.33 & 165.23  \\
		\hline
		16 & linear & S-shaped & - &\NS &5963&-&-&5963 &- &- &3.85 & -\\
		\hline
		17 & linear & S-shaped & trans.&\SC &2419&3544&-&5963 &- &- &3.14 & 0 \\
		\hline
		18 & linear & S-shaped & trans.+rot.&\SC &2167&3808& -&5975  &- &- &2.97 & 0.6 \\
		\hline
		19 & linear & S-shaped & trans.&\SV &603&40&11&654 &3/4 &5/16 &0.46 &7.3 \\
		\hline
		20 & linear & S-shaped & trans.+rot.&\SV &601&36&13&650  &2/2 &3/16 &0.45 &103.8 \\
		\hline 
		21 & robot & S-shaped & - &\NS &18591  &- &-&18591&-&- &1.3 &- \\
		\hline 
		22 & robot & S-shaped & trans.&\SC &6831  &11760 & - &18591 & -&- &0.8 & 0\\
		\hline 
		23 & robot & S-shaped & trans.+rot.&\SC &2421 &3864 &- &6285& -&- &0.32 & 49.3 \\
		\hline
		24 & robot & S-shaped & trans.&\SV &1749  &84 &11 &1844 & 3/4& 5/16 &0.14 &23.4 \\
		\hline 
		25 & robot & S-shaped & trans.+rot.&\SV &1513&0&13&1526 &2/2 &3/16 &0.14 &140.8 \\
		\hline
	\end{tabular}
\end{table*}

\section*{Acknowledgment}
The authors are supported by a research grant from The Boeing Company and a research grant from NSF (CPS 1739966). We would like to thank John L. Olson and Arthur S. Younger from The Boeing Company for valuable technical discussions.

\end{document}